\documentclass[12pt]{article}
\usepackage[margin=1in]{geometry}

\usepackage[ruled]{algorithm2e}
\usepackage{algorithmic}
\usepackage{amsmath,amssymb,amsthm,amsfonts}
\usepackage{times}
\usepackage[mathscr]{euscript}
\usepackage{xcolor}
\usepackage{xfrac}
\usepackage[colorlinks,linkcolor=black,bookmarksopen,
  bookmarksnumbered,citecolor=black,urlcolor=black]{hyperref}

\definecolor{addblue}{rgb}{0.0,0,0.8}

\usepackage[normalem]{ulem}

\usepackage{tikz}

\DeclareMathOperator*{\supp}         {supp}

\newcommand*{\real}                 {\mathbb{R}}

\theoremstyle{plain}
\newtheorem{theorem}{Theorem}[section]

\newtheorem{proposition}[theorem]{Proposition}
\newtheorem{corollary}[theorem]{Corollary}
\theoremstyle{definition}
\newtheorem{definition}[theorem]{Definition}
\newtheorem{assumption}[theorem]{Assumption}

\newtheorem{example}[theorem]{Example}

\usepackage{xcolor}

\usepackage{authblk}

\title{Structural Validation Of Synthetic Power Distribution Networks Using The Multiscale Flat Norm}

\author[1]{Kostiantyn Lyman}
\author[2]{Rounak Meyur}
\author[1]{Bala Krishnamoorthy}
\author[2]{Mahantesh Halappanavar}
\affil[1]{Washington State University}
\affil[2]{Pacific Northwest National Laboratory}

\begin{document}

\maketitle

\vspace*{-0.6in}
\tableofcontents

\clearpage
\begin{abstract}
  We study the problem of comparing a pair of geometric networks that may not be similarly defined, i.e., when they do not have one-to-one correspondences between their nodes and edges.
  Our motivating application is to compare power distribution networks of a region.
  Due to the lack of openly available power network datasets, researchers synthesize realistic networks resembling their actual counterparts.
  But the synthetic digital twins may vary significantly from one another and from actual networks due to varying underlying assumptions and approaches.
  Hence the user wants to evaluate the quality of networks in terms of their structural similarity to actual power networks.
  But the lack of correspondence between the networks renders most standard approaches, e.g., subgraph isomorphism and edit distance, unsuitable. 

  We propose an approach based on the \emph{multiscale flat norm}, a notion of distance between objects defined in the field of geometric measure theory, to compute the distance between a pair of planar geometric networks.
  Using a triangulation of the domain containing the input networks, the flat norm distance between two networks at a given scale can be computed by solving a linear program.
  In addition, this computation automatically identifies the 2D regions (patches) that capture where the two networks are different.
  We demonstrate through 2D examples that the flat norm distance can capture the variations of inputs more accurately than the commonly used Hausdorff distance.
  As a notion of \emph{stability}, we also derive upper bounds on the flat norm distance between a simple 1D curve and its perturbed version as a function of the radius of perturbation for a restricted class of perturbations.
  We demonstrate our approach on a set of actual power networks from a county in the USA.
  Our approach can be extended to validate synthetic networks created for multiple infrastructures such as transportation, communication, water, and gas networks.

  \noindent {\bfseries Keywords:} geometric networks, multiscale flat norm, stability, synthetic network validation.
\end{abstract}

\section{Introduction}\label{sec:intro}

The power grid is the most vital infrastructure that provides crucial support for the delivery of basic services to most segments of society. 
Once considered a passive entity in power grid planning and operation, the power distribution system poses significant challenges in the present day. 
The increased adoption of rooftop solar photovoltaics (PVs) and electric vehicles (EVs) augmented with residential charging units has altered the energy consumption profile of an average consumer.
Access to extensive datasets pertaining to power distribution networks and residential consumer demand is vital for public policy researchers and power system engineers alike. 
However, the proprietary nature of power distribution system data hinders their public availability.
This has led researchers to develop frameworks that synthesize realistic datasets pertaining to the power distribution system~\cite{highres_net,overbye_2020,nrel_net,rounak2020,anna_naps,schweitzer}. 
These frameworks create digital replicates similar to the actual power distribution networks in terms of their structure and function. 
Hence the created networks can be used as \emph{digital duplicates} in simulation studies of policies and methods before implementation in real systems. 

The algorithms associated with these frameworks vary widely---ranging from first principles based approaches~\cite{rounak2020,anna_naps} to learning statistical distributions of network attributes~\cite{schweitzer} to using deep learning models such as generative adversarial neural networks~\cite{feeder_gan}. 
Validating the synthetic power distribution networks with respect to their physical counterpart is vital for assessing the suitability of their use as effective digital duplicates. 
Since the underlying assumptions and algorithms of each framework are distinct from each other, some of them may excel compared to others in reproducing digital replicates with better precision for selective regions. 
To this end, we require well-defined metrics to rank the frameworks and judge their strengths and weaknesses in generating digital duplicates of power distribution networks for a particular geographic region.

The literature pertaining to frameworks for synthetic distribution network creation include certain validation results that compare the generated networks to the actual counterpart~\cite{highres_net,validate2020,schweitzer}. But the validation results are mostly limited to comparing the statistical network attributes such as degree and hop distributions and power engineering operational attributes such as node voltages and edge power flows. Since power distribution networks represent real physical systems, the created digital replicates have associated geographic embedding. Therefore, a structural comparison of synthetic network graphs to their actual counterpart becomes pertinent for power distribution networks with geographic embedding.
Consider an example where a digital twin is used to analyze impact of a weather event~\cite{samiul2021}.
Severe weather events such as hurricanes, earthquakes and wild fires occur in specific geographic trajectories, affecting only portions of societal infrastructures. 
In order to correctly identify them during simulations, the digital twin should structurally resemble the actual infrastructure.

\medskip
\noindent\textbf{Problem Statement:} 
In recent years, the problem of evaluating the quality of reconstructed networks has been studied for street maps.
Certain metrics were defined to compare outputs of frameworks that use GPS trajectory data to reconstruct street map graphs~\cite{roads_2015,roads2014}.
The abstract problem can be stated as follows: \emph{compute the similarity between a given pair of embedded planar graphs}.
This is similar to the well known subgraph isomorphism problem~\cite{eppstein_1995} wherein we look for isomorphic subgraphs in a pair of given graphs.
A major precursor to this problem is that we require a one-to-one mapping between nodes and edges of the two graphs.
While such mappings are well-defined for street networks, the same cannot be inferred for power distribution networks.
Since power network datasets are proprietary, the node and edge labels are redacted from the network before it is shared.
The actual network is obtained as a set of ``drawings'' with associated geographic embeddings. Each drawing can be considered as a collection of line segments termed a \emph{geometry}. Hence the problem of comparing a set of power distribution networks with geographic embedding can be stated as the following: \emph{compute the similarity between a given pair of geometries lying on a geographic plane}.

\medskip
\noindent\textbf{Our Contributions:} 
We propose a new distance measure to compare a pair of geometries using the \emph{flat norm}, a notion of distance between generalized objects studied in geometric measure theory~\cite{Federer1969,Morgan2016}.
This distance combines the difference in length of the geometries with the area of the patches contained between them. The area of patches in between the pair of geometries accounts for the lateral displacement between them.
We employ a \emph{multiscale} version of the flat norm~\cite{MoVi2007} that uses a scale parameter $\lambda \geq 0$ to combine the length and area components
(for the sake of brevity, we refer to the multiscale flat norm simply as the flat norm).
Intuitively, a smaller value of $\lambda$ captures larger patches of area between the geometries while a large value of $\lambda$ captures more of the (differences in) lengths of the geometries.
Computing the flat norm over a range of values of $\lambda$ allows us to compare the geometries at multiple scales.
For computation, we use a discretized version of the flat norm defined on simplicial complexes~\cite{IbKrVi2013}, which are triangulations in our case.
It may not be possible to derive standard stability results for this distance measure that imply only small changes in the flat norm metric when the inputs change by a small amount---there is no alternative metric to measure the \emph{small change in the input}.
We demonstrate through 2D examples (in Fig.~\ref{fig:currents:example-currents-and-neighborhoods}), for instance, that the commonly used Hausdorff metric cannot be used for this purpose.
Instead, we have derived upper bounds on the flat norm distance between a piecewise linear 1-current and its perturbed version as a function of the radius of perturbation under certain assumptions provided the perturbations are performed carefully (see Section \ref{subsec:FN-BOUND}).

A lack of one-to-one correspondence between edges and nodes in the pair of networks prevents us from performing one-to-one comparison of edges.
Instead we can sample random regions in the area of interest and compare the pair of geometries within each region.
For performing such local comparisons, we define a \emph{normalized flat norm} where we normalize the flat norm distance between the parts of the two geometries by the sum of the lengths of the two parts in the region. 
Such comparison enables us to characterize the quality of the digital duplicate for the sampled region.
Further, such comparisons over a sequence of sampled regions allows us to characterize the suitability of using the entire synthetic network as a duplicate of the actual network. 

Our \textbf{main contributions} are the following:
(i) we propose a distance measure for comparing a pair of geometries embedded in the same plane using the flat norm that accounts for deviation in length and lateral displacement between the geometries; and
(ii) we perform a region-based characterization of synthetic networks by sampling random regions and comparing the pair of geometries contained within the sampled region.
The proposed distance allows us to perform global as well as local comparisons between a pair of network geometries.

\subsection{Related Work} \label{ssec:relwrk}
Several well defined graph structure comparison metrics such as subgraph isomorphism and edit distance have been proposed in the literature along with algorithms to compute them efficiently.
Tantardini et al.~\cite{Tantardini2019} compare graph network structures for the entire graph (global comparison) as well as for small portions of the graph known as motifs (local comparison).
Other researchers have proposed methodologies to identify structural similarities in embedded graphs~\cite{BaDiBiChSuWa2019,graphsim-2020}.
However, all these methods depend on one-to-one correspondence of graph nodes and edges rather than considering the node and edge geometries of the graphs.
The edit distance, i.e., the minimum number of edit operations to transform one network to the other, has been widely used to compare networks having structural properties~\cite{PaGaMiHa2018,riba2020,xu2015}.
Riba et al.~\cite{riba2020} used the Hausdorff distance between nodes in the network to compare network geometries.
Majhi et al.~\cite{MaWe2024} modified the traditional definition of graph edit distance to be applicable in the context of ``geometric graphs'' embedded in a Euclidean space.
Along with the usual insertion and deletion operations, the authors have proposed a cost for translation in computing the geometric edit distance between the graphs.
However, the authors also show that the problem of computing this metric is NP-hard.
As a computationally tractable alternative, Majhi studied \cite{Ma2023} a graph mover's distance (GMD) between two geometric graphs with at most $n$ vertices that can be computed in $O(n^3)$ time.
At the same time, the GMD requires that the input geometric graphs are \emph{ordered}, i.e., that their vertices are ordered or index in a specific manner.

Meyur et al.~\cite{rounak_pnas} compared network geometries using the Hausdorff distance after partitioning the geographic region into small rectangular grids and comparing the geometries for each grid.
However, the Hausdorff metric is sensitive to outliers as it focuses only on the maximum possible distance between the pair of geometries.
When the geometries coincide almost entirely except in a few small portions, the Hausdorff metric still records the discrepancy in those small portions without accounting for the similarity over the majority of portions.
The similar approach used by Brovelli et al.~\cite{roads_haus} to compare a pair of road networks in a geographic region suffers from the same drawback.
This necessitates a well-defined distance metric between networks with geographic embedding~\cite{roads2014}.

Several comparison methods have been proposed in the context of planar graphs embedded in a Euclidean space~\cite{streets2006,efficiency2020}.
They include local and global metrics to compare road networks.
The local metrics characterize the networks based on cliques and motifs, while the global metrics involve computing the \emph{efficiency} of constructing the infrastructure network.
The most efficient network is assumed to be the one with only straight line geometries connecting node pairs.
Albeit useful to characterize network structures, these methods are not suitable for a numeric comparison of network geometries.

\section{Preliminaries}\label{sec:prelim}
\begin{definition}[Geometric graph]
A graph $\mathscr{G}\left(\mathscr{V},\mathscr{E}\right)$ with node set $\mathscr{V}$ and edge set $\mathscr{E}$ is said to be a geometric graph of $\,\mathbb{R}^d$ if the set of nodes $\mathscr{V}\subset\mathbb{R}^d$ and the edges are Euclidean straight line segments $\left\{\overline{uv}~|~e:=\left(u,v\right)\in\mathscr{E}\right\}$ which intersect (possibly) at their endpoints.
\end{definition}
\begin{definition}[Structurally similar geometric graphs] \label{def:strsmlrggrfs}
  \ Two geometric graphs $\mathscr{G}_0\left(\mathscr{V}_0,\mathscr{E}_0\right)$ and $\mathscr{G}_1\left(\mathscr{V}_1,\mathscr{E}_1\right)$ are said to be \emph{structurally similar} at the level of $\gamma \geq 0$, termed $\gamma$-similar, if \\
  $\,\operatorname{dist}\left(\mathscr{G}_0,\mathscr{G}_1\right)\leq\gamma$ for the distance function $\operatorname{dist}$ between the two graphs.
\end{definition}
We could consider a given network as a set of edge geometries.
Hence we could consider the problem of comparing geometric graphs $\mathscr{G}_0$ and $\mathscr{G}_1$ as that of comparing the set of edge geometries $\mathscr{E}_0$ and $\mathscr{E}_1$. In this paper, we propose a suitable distance that allows us to compare between a pair of geometric graphs or a pair of geometries.
We use the multiscale flat norm, which has been well explored in the field of geometric measure theory, to define such a distance between the geometries. 

\subsection{Multiscale Flat Norm}
\label{subsec:multiscale-flat-norm}
We use the multiscale simplicial flat norm proposed by Ibrahim et al.~\cite{IbKrVi2013} to compute the distance between two networks.
We now introduce some background for this computation. 
A $d$-dimensional \emph{current} $T$ (referred to as a $d$-current) is a generalized $d$-dimensional geometric object with orientations (or direction) and multiplicities (or magnitude).
An example of a $2$-current is a surface with finite area (multiplicity) and a specific orientation (clockwise or counterclockwise).
We use $\mathcal{C}_d$ to denote the set of all $d$-currents,
and $\mathcal{C}_d(\real^{p})$ to denote the set of $d$-currents embedded in $\real^{p}$.
$\mathbf{V}_d(T)$ or $\left|{T}\right|$ denotes the $d$-dimensional \emph{volume} of $T$, e.g., length in 1D or area in 2D.
The boundary of $T$, denoted by $\partial T$, is a $(d-1)$-current.
The \emph{multiscale flat norm} of a $d$-current $T \in \mathcal{C}_d$, at scale $\lambda \geq 0$ is defined as
\begin{equation}
    \mathbb{F}_\lambda\left(T\right) = \min_{S \in \mathcal{C}_{d + 1}}\left\{\mathbf{V}_d\left(T-\partial S\right)+\lambda \mathbf{V}_{d+1}\left(S\right)\right\},
\label{eq:flatnorm-def}
\end{equation}
where the minimum is taken over all $(d+1)$-currents $S$. 
Computing the flat norm of a 1-current (curve) $T$ identifies the optimal 2-current (area patches) $S$ that minimizes the sum of the length of current $T-\partial S$ and the area of patch(es) $S$.
Fig.~\ref{fig:demo-flatnorm-basic} shows the flat norm computation for a generic 1D current $T$ (blue).
The 2D area patches $S$ (magenta) are computed such that the expression in Eq.~(\ref{eq:flatnorm-def}) is minimized for the chosen value of $\lambda$ that ends up using most of the patch under the sharper spike on the left but only a small portion of the patch under the wider bump to the right.
\begin{figure}[ht!]
    \centering
    \includegraphics[width=0.98\textwidth]{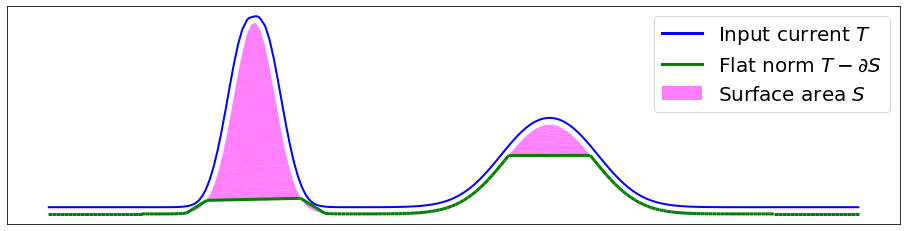}
    \caption{Multiscale flat norm of a 1D current $T$ (blue).
      The flat norm is the sum of length of the resulting 1D current $T-\partial S$ (green) and the area of 2D patches $S$ (magenta).
      We show $T-\partial S$ slightly separated for easy visualization.}
    \label{fig:demo-flatnorm-basic}
\end{figure}

The scale parameter $\lambda$ can be intuitively understood as follows.
Rolling a ball of radius $1/\lambda$ on the $1$-current $T$ traces the output current $T-\partial S$ and the untraced regions constitute the patches $S$.
Hence we observe that for a large $\lambda$, the radius of the ball is very small and hence it traces major features while smoothing out (i.e., missing) only minor features (wiggles) of the input current.
But for a small $\lambda$, the ball with a large radius smoothes out larger scale features (bumps) in the current.
Note that for smaller $\lambda$, the cost of area patches is smaller in the minimization function and hence more patches are used for computing the flat norm. We can use the flat norm to define a natural distance between a pair of 1-currents $T_1$ and $T_2$ as follows~\cite{IbKrVi2013}.
\begin{equation}
    \mathbb{F}_\lambda\left(T_1,T_2\right) = \mathbb{F}_\lambda\left(T_1-T_2\right)\label{eq:flatnorm-def2}
\end{equation}

We compute the flat norm distance between a pair of input geometries (synthetic and actual) as the flat norm of the current $T=T_1-T_2$ where $T_1$ and $T_2$ are the currents corresponding to individual geometries.
Let $\Sigma$ denote the set of all line segments in the input current $T$.
We perform a constrained triangulation of $\Sigma$ to obtain a $2$-dimensional finite oriented simplicial complex $K$.
A constrained triangulation ensures that each line segment $\sigma_i\in\Sigma$ is an edge in $K$,
and that $T$ is an oriented $1$-dimensional subcomplex of $K$.

Let $m$ and $n$ denote the numbers of edges and triangles in $K$.
We can denote the input current $T$ as a $1$-chain $\sum_{i=1}^{m}t_i\sigma_i$ where $\sigma_i$ denotes an edge in $K$ and $t_i$ is the corresponding multiplicity.
Note that $t_i=-1$ indicates that orientation of $\sigma_i$ and $T$ are opposite, $t_i=0$ denotes that $\sigma_i$ is not contained in $T$, and $t_i=1$ implies that $\sigma_i$ is oriented the same way as $T$.
Similarly, we define the set $S$ to be the $2$-chain of $K$ and denote it by $\sum_{i=1}^{m}s_i\omega_i$ where $\omega_i$ denotes a $2$-simplex in $K$ and $s_i$ is the corresponding multiplicity.

The boundary matrix $\left[\partial\right]\in\mathbb{Z}^{m\times n}$ captures the intersection of the $1$ and $2$-simplices of $K$.
The entries of the boundary matrix $\left[\partial\right]_{ij}\in\{-1,0,1\}$.
If edge $\sigma_i$ is a face of triangle $\omega_j$, then $\left[\partial\right]_{ij}$ is nonzero and it is zero otherwise.
The entry is $-1$ if the orientations of $\sigma_i$ and $\omega_j$ are opposite and it is $+1$ if the orientations agree. 

We can respectively stack the $t_i$'s and $s_i$'s in $m$ and $n$-length vectors $\mathbf{t}\in\mathbb{Z}^m$ and $\mathbf{s}\in\mathbb{Z}^n$. The $1$-chain representing $T-\partial S$ is denoted by $\mathbf{x}\in\mathbb{Z}^m$ and is given as $\mathbf{x} = \mathbf{t} - \left[\partial\right]\mathbf{s}$.
The multiscale flat norm defined in Eq.~(\ref{eq:flatnorm-def}) can be computed by solving the following optimization problem:
\begin{equation}
\begin{aligned}
    \mathbb{F}_\lambda\left(T\right) = \min_{\mathbf{s}\in\mathbb{Z}^n} & \sum_{i=1}^{m}w_i\left|x_i\right|+\lambda \left(\sum_{j=1}^{n}v_j\left|s_j\right|\right)\\
    \textrm{s.t.} & \quad \mathbf{x} = \mathbf{t} - \left[\partial\right]\mathbf{s}, 
    \quad \mathbf{x} \in \mathbb{Z}^m,
\end{aligned}
\label{eq:opt-flatnorm-def}
\end{equation}
where $\mathbf{V}_d\left(\tau\right)$ in Eq.~(\ref{eq:flatnorm-def}) denotes the volume of the $d$-dimensional simplex $\tau$.
We denote volume of the edge $\sigma_i$ as $\mathbf{V}_1(\sigma_i)=w_i$ and set it to be the Euclidean length, and volume of a triangle $\tau_j$ as $\mathbf{V}_2(\tau_j)=v_j$ and set it to be the area of the triangle.

In this work, we consider geometric graphs embedded on the geographic plane and are associated with longitude and latitude coordinates.
We compute the Euclidean length of edge $\sigma_i$ as $w_i=R\Delta\phi_i$ where $\Delta\phi_i$ is the Euclidean normed distance between the geographic coordinates of the terminals of $\sigma_i$ and $R$ is the radius of the earth.
Similarly, the area of triangle $\tau_j$ is computed as $v_j=R^2\Delta\Omega_j$ where $\Delta\Omega_j$ is the solid angle subtended by the geographic coordinates of the vertices of $\tau_j$.

Using the fact that the objective function is piecewise linear in $\mathbf{x}$ and $\mathbf{s}$, the minimization problem can be reformulated as an integer linear program (ILP) as follows:
\begin{subequations}
\begin{align}
    \mathbb{F}_\lambda\left(T\right) = \min & \sum_{i=1}^{m}w_i\left(x_i^{+}+x_i^{-}\right)+\lambda \left(\sum_{j=1}^{n}v_j\left(s_j^{+}+s_j^{-}\right)\right)\\
    \textrm{s.t.} & \quad \mathbf{x}^{+}-\mathbf{x}^{-} = \mathbf{t} - \left[\partial\right]\left(\mathbf{s}^{+}-\mathbf{s}^{-}\right)\\
    & \quad \mathbf{x}^{+},\mathbf{x}^{-} \geq 0,\quad \mathbf{s}^{+},\mathbf{s}^{-} \geq 0\\
    & \quad \mathbf{x}^{+},\mathbf{x}^{-} \in \mathbb{Z}^m,\quad \mathbf{s}^{+},\mathbf{s}^{-} \in \mathbb{Z}^n \label{seq:intger-constraints}
\end{align}
\label{eq:opt-flatnorm}
\end{subequations}
\noindent The linear programming relaxation of the ILP in Eq.~(\ref{eq:opt-flatnorm}) is obtained by ignoring the integer constraints Eq.~(\ref{seq:intger-constraints}).
We refer to this relaxed linear program (LP) as the \emph{\bfseries flat norm LP}.
Ibrahim et al.~\cite{IbKrVi2013} showed that the boundary matrix $[\partial]$ is totally unimodular for our application setting.
Hence the flat norm LP will solve the ILP, and hence the flat norm can be computed in polynomial time.

Algorithm~\ref{alg:distance} describes how we compute the distance between a pair of geometries with the associated embedding on a metric space $\mathcal{M}$.
We assume that the geometries (networks) $\mathscr{G}_1\left(\mathscr{V}_1,\mathscr{E}_1\right)$ and $\mathscr{G}_2\left(\mathscr{V}_2,\mathscr{E}_2\right)$ with respective node sets $\mathscr{V}_1,\mathscr{V}_2$ and edge sets $\mathscr{E}_1,\mathscr{E}_2$ have no one-to-one correspondence between the $\mathscr{V}_i$'s or $\mathscr{E}_i$'s.
Note that each vertex $v\in\mathscr{V}_1,\mathscr{V}_2$ is a point and each edge $e\in\mathscr{E}_1,\mathscr{E}_2$ is a straight line segment in $\mathcal{M}$. We consider the collection of edges $\mathscr{E}_1,\mathscr{E}_2$ as input to our algorithm.
First, we orient the edge geometries in a particular direction (left to right in our case) to define the currents $T_1$ and $T_2$, which have both magnitude and direction.
Next, we consider the bounding rectangle $\mathscr{E}_{\textrm{bound}}$ for the edge geometries and define the set $\Sigma$ to be triangulated as the set of all edges in either geometry and the bounding rectangle.
We perform a constrained Delaunay triangulation \cite{Si2010} on the set $\Sigma$ to construct the $2$-dimensional simplicial complex $K$.
The constrained triangulation ensures that the set of edges in $\Sigma$ is included in the simplicial complex $K$.
Then we define the currents $T_1$ and $T_2$ corresponding to the respective edge geometries $\mathscr{E}_1$ and $\mathscr{E}_2$ as $1$-chains in $K$.
Finally, the flat norm LP is solved to compute the simplicial flat norm.

\begin{algorithm}[ht!]
\caption{Distance between a pair of geometries}
\label{alg:distance}
\textbf{Input}: Geometries $\mathscr{E}_1,\mathscr{E}_2$\\
\textbf{Parameter}: Scale $\lambda$
\begin{algorithmic}[1]
\STATE Orient each edge in the edge sets from left to right: $\Tilde{\mathscr{E}}_1:=\mathsf{Orient}\left(\mathscr{E}_1\right);~~\Tilde{\mathscr{E}}_2:=\mathsf{Orient}\left(\mathscr{E}_2\right)$.
\STATE Find bounding rectangle for the pair of geometries: \\$\mathscr{E}_{\textrm{bound}}=\mathsf{rect}\left(\Tilde{\mathscr{E}}_1,\Tilde{\mathscr{E}}_2\right)$.
\STATE Define the set of line segments to be triangulated: \\$\Sigma=\Tilde{\mathscr{E}}_1\cup\Tilde{\mathscr{E}}_2\cup\mathscr{E}_{\textrm{bound}}$.
\STATE Perform constrained triangulation on set $\Sigma$ to construct $2$-dimensional simplicial complex $K$.
\STATE Define the currents $T_1,T_2$ as $1$-chains of oriented edges\\
$\Tilde{\mathscr{E}}_1$ and $\Tilde{\mathscr{E}}_2$ in $K$.
\STATE Solve the flat norm LP to compute flat norm $\mathbb{F}_\lambda\left(T_1-T_2\right)$.
\end{algorithmic}
\textbf{Output}: Flat norm distance $\mathbb{F}_\lambda\left(T_1-T_2\right)$.
\end{algorithm}

\subsection{Normalized Flat Norm}
\label{subsec:normalized-fn}
Recall that in our context of synthetic power distribution networks, the primary goal of comparing a synthetic network to its actual counterpart is to infer the quality of the replica or the \emph{digital duplicate} synthesized by the framework.
The proposed approach using the flat norm for structural comparison of a pair of geometries provides us a method to perform global as well as local comparison.
While we can produce a global comparison by computing the flat norm distance between the two networks, it may not provide us with complete information on the quality of the synthetic replicate.
On the other hand, a local comparison can provide us details about the framework generating the synthetic networks.
For example, a synthetic network generation framework might produce higher quality digital replicates of actual power distribution networks for urban regions as compared to rural areas.
A local comparison highlights this attribute and identifies potential use case scenarios of a given synthetic network generation framework.

Furthermore, availability of actual power distribution network data is sparse due to its proprietary nature.
We may not be able to produce a global comparison between two networks due to unavailability of network data from one of the sources.
Hence, we want to restrict our comparison to only the portions in the region where data from either network is available, which also necessitates a local comparison between the networks.

For a local comparison, we consider uniform sized regions and compute the flat norm distance between the pair of geometries within the region.
However, the computed flat norm is dependent on the length of edges present within the region from either network.
Hence we define the \emph{normalized} multiscale flat norm, denoted by $\widetilde{\mathbb{F}}_{\lambda}$, for a given region as
\begin{equation}
\label{eq:flat-norm-normalized-def}
    \widetilde{\mathbb{F}}_{\lambda}\left(T_1-T_2\right) = \frac{\mathbb{F}_{\lambda}\left(T_1-T_2\right)}{|T_1|+|T_2|}\,.
\end{equation}

For a given parameter $\epsilon$, a local region is defined as a square of size $2\epsilon\times 2\epsilon$ steradians.
Let $T_{1,\epsilon}$ and $T_{2,\epsilon}$ denote the currents representing the input geometries inside the local region characterized by $\epsilon$. Note that the ``amount'' or the total length of network geometries within a square region varies depending on the location of the local region. In this case, the lengths of the network geometries are respectively $|T_{1,\epsilon}|$ and $|T_{2,\epsilon}|$. Therefore, we use the ratio of the total length of network geometries inside a square region to the parameter $\epsilon$ to characterize this ``amount'' and denote it by $|T|/\epsilon$ where 
\begin{equation}
    |T|/\epsilon = \frac{|T_{1,\epsilon}|+|T_{2,\epsilon}|}{\epsilon}\,.
\end{equation}
Note that while performing a comparison between a pair of network geometries in a local region using the multiscale flat norm, we need to ensure that comparison is performed for similar length of the networks inside similar regions.
Therefore, the ratio $|T|/\epsilon$, which indicates the length of networks inside a region scaled to the size of the region, becomes an important aspect of characterization while performing the flat norm based comparison.

\section{Implementation of Flat Norm}
\label{sec:implementation}

\subsection{Geometry Comparison using Flat Norm}
\label{subsec:implementation:geom}
We show a simple example depicting the use of flat norm to compute the distance between a pair of geometries that are two line segments of equal length meeting at their midpoints in Fig.~\ref{fig:demo-flatnorm}.
As the angle between the two line segments decreases from $90$ to $15$ degrees, the computed flat norm also decreases.
We used $\lambda=1$ for this illustration.

\begin{figure}[ht!]
    \centering
    \includegraphics[width=0.47\textwidth]{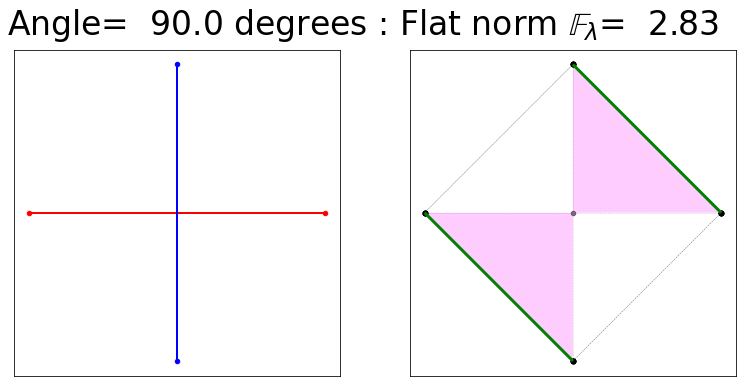}
    \includegraphics[width=0.47\textwidth]{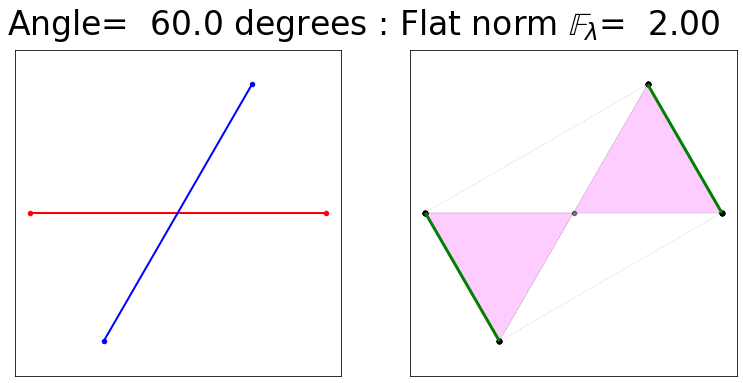}
    \includegraphics[width=0.47\textwidth]{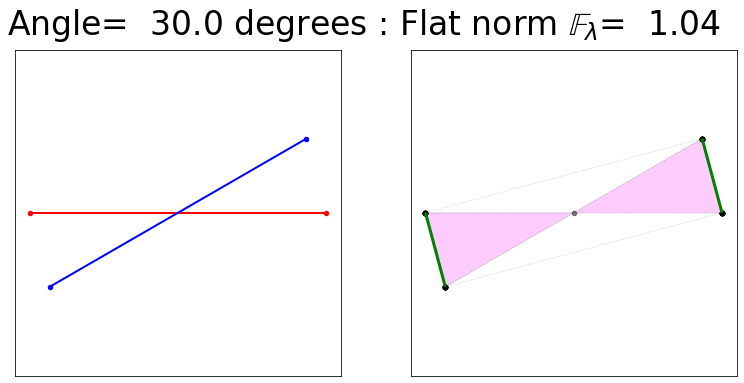}
    \includegraphics[width=0.47\textwidth]{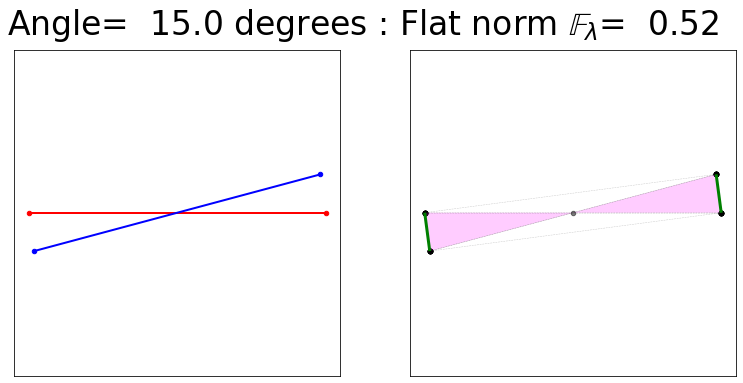}
    \caption{Variation in flat norm for pairs of geometries shown in red and blue as the angle between them decreases.
      Components of the flat norm decomposition is shown in green (1D) and pink (2D) in each case.
      When the geometries are perpendicular to each other, flat norm distance is the maximum and it decreases as the angle decreases.}
    \label{fig:demo-flatnorm}
\end{figure}

\subsection{Flat Norm Computation for a Pair of Geometries}
\label{subsec:implementation:fn-geom}
~We demonstrate the steps involved in computing the flat norm for a pair of input geometries in Fig.~\ref{fig:demo-toy}.
The input geometries are a collection of line segments shown in blue and red (top left). We construct the set $\Sigma$ by combining all the edges of either geometry along with the bounding rectangle (top right). Thereafter, we perform a constrained triangulation to construct the $2$-dimensional simplicial complex $K$ (bottom left). Finally, we compute the multiscale simplicial flat norm with $\lambda=1$ (bottom right).
Note that this computation captures the length deviation (shown by green edges) and the lateral displacement (shown by the magenta patches).
\begin{figure}[ht!]
    \centering
    \includegraphics[width=0.98\textwidth]{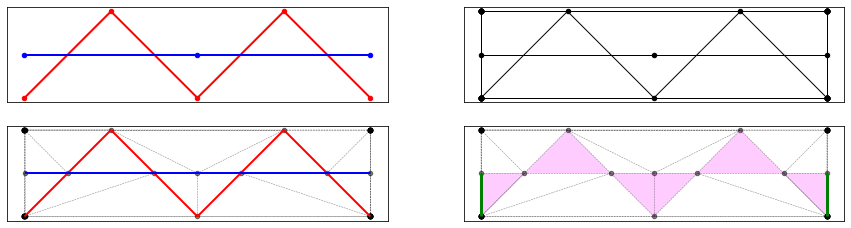}
    \caption{Steps in computing the flat norm for a pair of input geometries.}
    \label{fig:demo-toy}
\end{figure}

\subsection{Flat Norm Computation for a Pair of Networks}
\label{subsec:implementation:fn-power-grids}
The steps involved in computing the multiscale flat norm distance between a pair of geometries are shown in Fig.~\ref{fig:example-1}.
They include the actual power distribution network (red) for a region in a county from USA and the synthetic network (blue) constructed for the same region~\cite{rounak2020}.
First, we orient each edge in either network from left to right.
Thereafter, we find the enclosing rectangular boundary around the pair of networks.
We perform a constrained Delaunay triangulation which ensures that the edges in the geometries and the convex boundary are selected as edges of the triangles.
Finally the flat norm LP (relaxation of the ILP in Eq.~(\ref{eq:opt-flatnorm})) is solved to compute the flat norm distance between the networks.
\begin{figure}[p!]
    \centering
    \includegraphics[height=0.80\textheight]{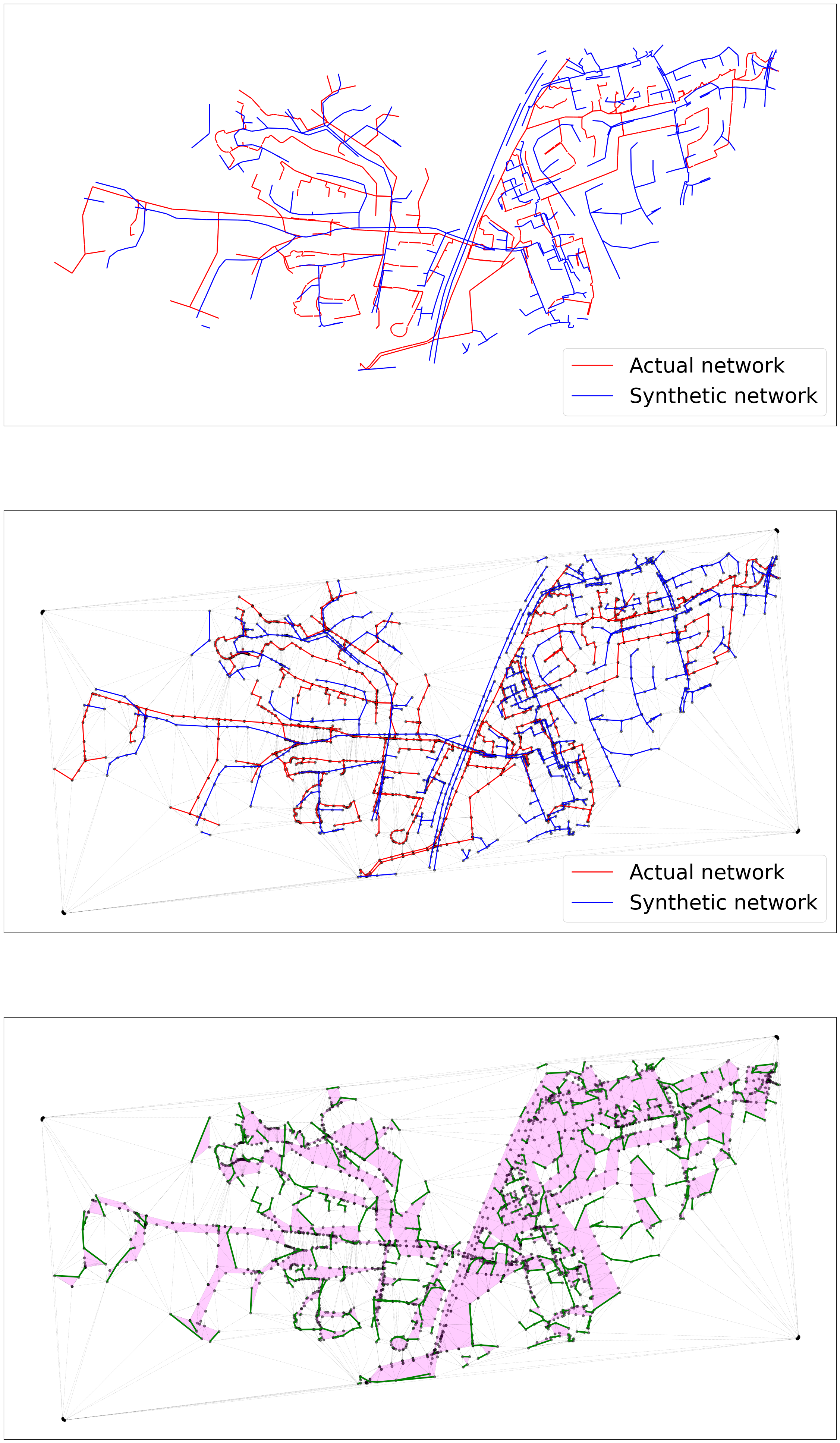}
    \caption{Steps showing the flat norm distance computation between two networks (shown in blue and red in the top two plots).
      First, the convex rectangular boundary around the pair of networks is identified.
      A constrained triangulation is computed such that the edges in the networks and convex boundary are edges of triangles (middle).
      The flat norm LP is solved to compute the simplicial flat norm, which includes the sum of areas of the magenta triangles and lengths of green edges (bottom).
    }
    \label{fig:example-1}
\end{figure}

\begin{figure}[p!]
    \centering
    \includegraphics[width=0.65\textwidth]{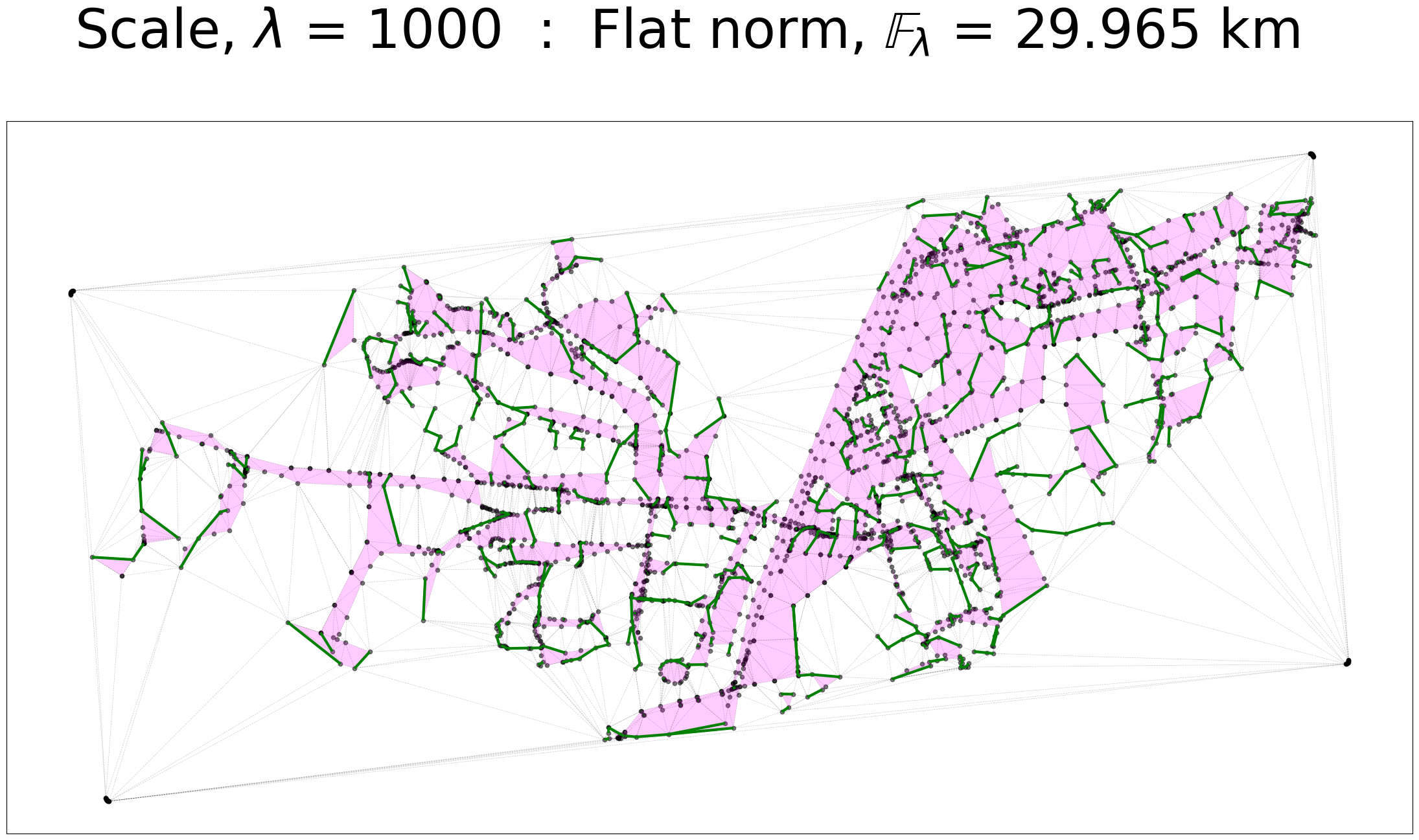} \\
    \vspace*{0.2in}
    \includegraphics[width=0.65\textwidth]{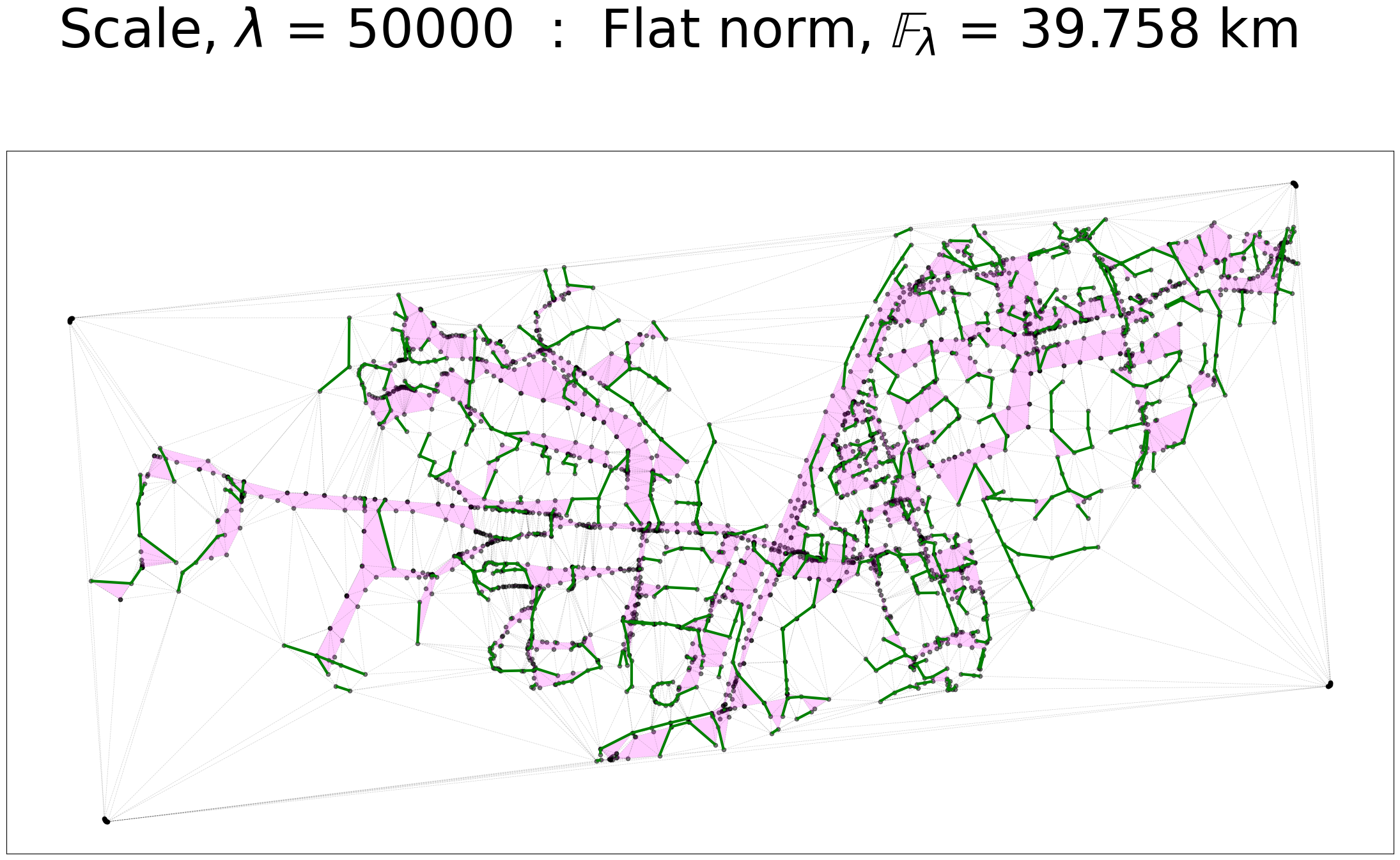} \\
    \vspace*{0.2in}
    \includegraphics[width=0.65\textwidth]{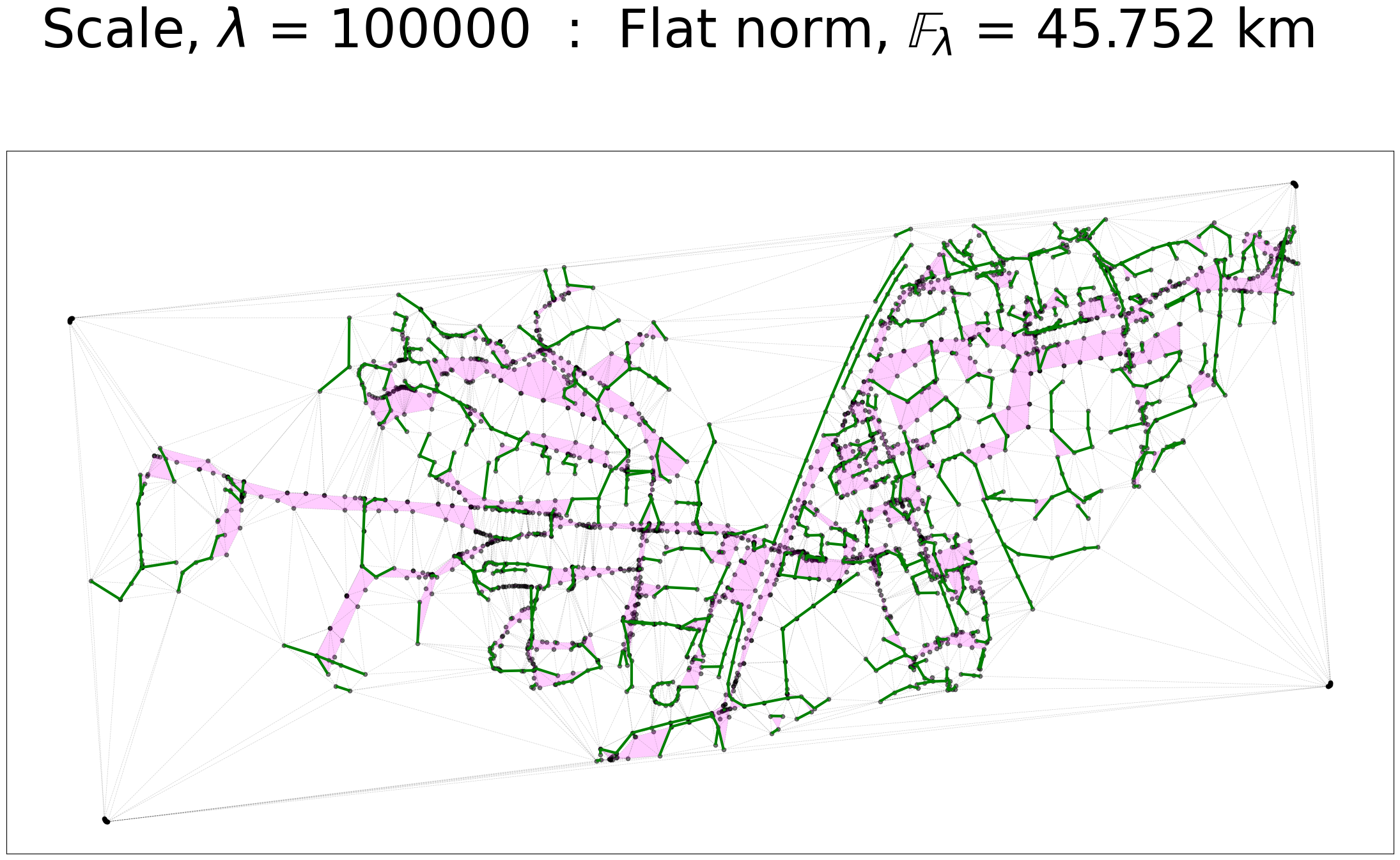}
    \caption{The flat norm computed between the pair of network geometries for three values of the scale parameter $\lambda$ ranging between $\lambda=1000$ to $\lambda=10000$.}
    \label{fig:lambda-var}
\end{figure}

The multiscale flat norm produces different distance values for different values of the scale parameter $\lambda$.
Fig.~\ref{fig:lambda-var} shows the flat norm distance between the actual and synthetic power network for the same region for multiple values of the scale parameter $\lambda$.
We observe that as $\lambda$ becomes larger, the 2D patches used in computing the flat norm become smaller as it becomes more expensive to use the area term in the flat norm LP minimization problem.

The variation of the computed flat norm for different values of the scale parameter is summarized in Fig.~\ref{fig:scale-variation}.
As the scale parameter is increased, fewer area patches are considered in the simplicial flat norm computation.
This is captured by the blue decreasing curve in the plot.
The computed flat norm increases for larger values of the scale parameter $\lambda$ as more and more individual currents contribute their unscaled length toward the flat norm value instead of becoming a boundary of some area component, which, if there are any, now also contribute more because of the increased scale  $\lambda$.
We show the plot with two different vertical scales: the left scale indicates the deviation in length (measured in km) and the right scale shows the deviation expressed through the area patches (measured in sq.km).
\begin{figure}[ht!]
    \centering
    \includegraphics[width=0.90\textwidth]{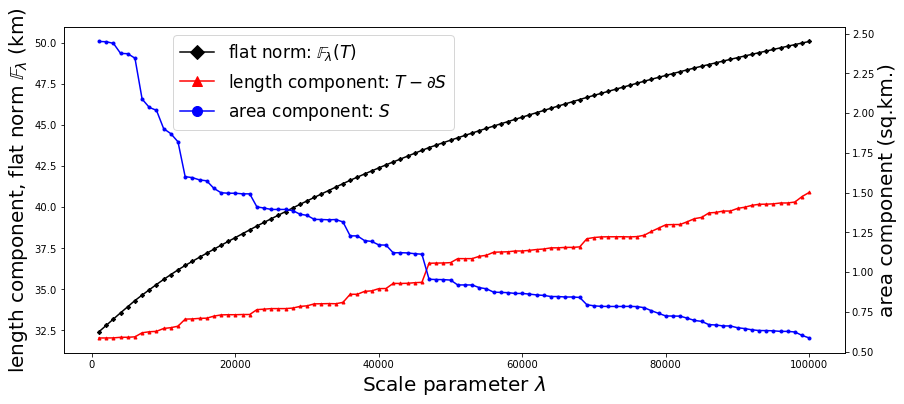}
    \caption{Effect of varying the scale parameter $\lambda$ in the flat norm computation.
      The flat norm for a $1$-dimensional current consists of two parts: a length component and a scaled surface area component.
      The variations in the length component and the unscaled surface area component (right vertical scale) are also shown.}
    \label{fig:scale-variation}
\end{figure}

\subsection{Comparing Network Geometries}
The primary goal of computing the flat norm is to compare the pair of input geometries.
As mentioned earlier, the flat norm provides an accurate measure of difference between the geometries by considering both the length deviation and area patches in between the geometries.
Further, we normalize the computed flat norm to the total length of the geometries.
In this section, we show examples where we computed the normalized flat norm for the pair of network geometries (actual and synthetic) for a few regions.

\begin{figure}[ht!] 
    \centering
    \includegraphics[width=0.46\textwidth]{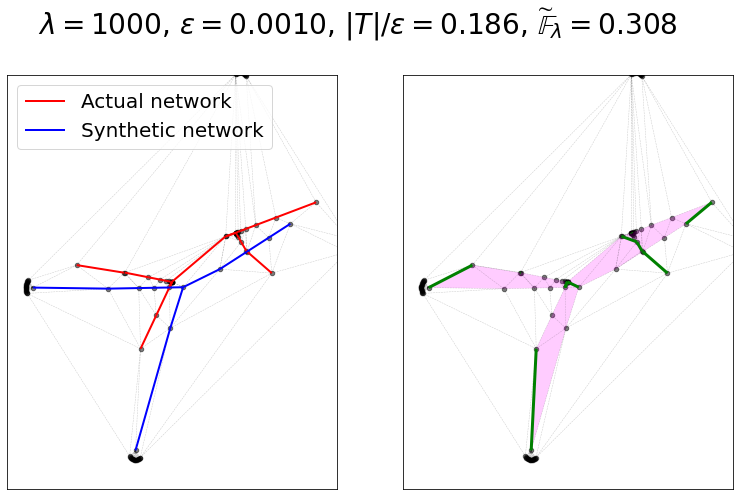}
    \includegraphics[width=0.46\textwidth]{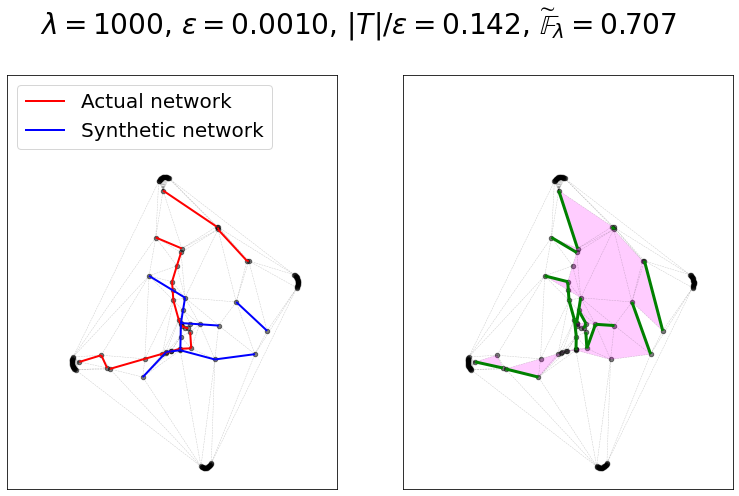}
    \includegraphics[width=0.46\textwidth]{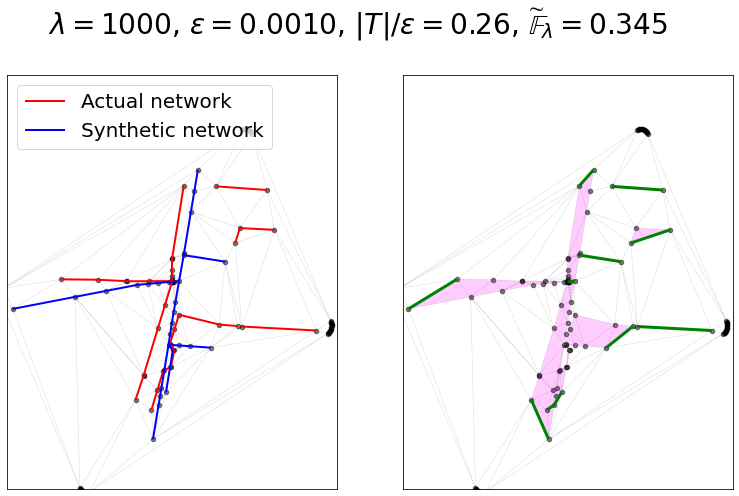}
    \includegraphics[width=0.46\textwidth]{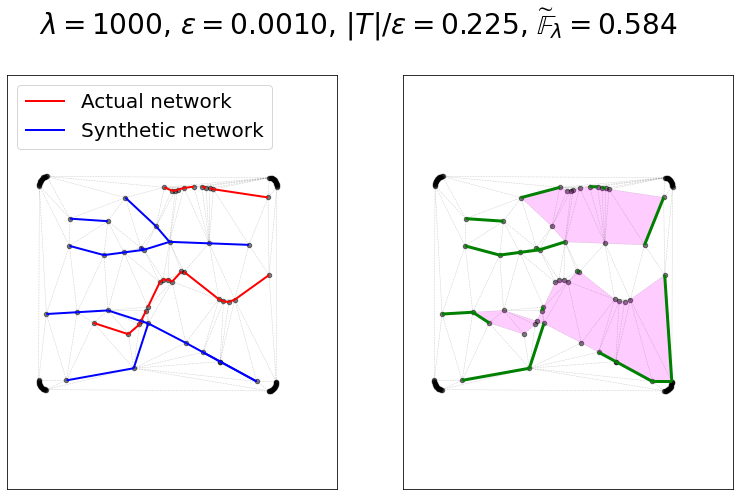}
    \caption{Normalized flat norm (with scale $\lambda=1000$) distances for pairs of regions in the network of same size ($\epsilon=0.001$) with similar $|T|/\epsilon$ ratios (two pairs each in the top and bottom rows).
      The pairs of geometries for the first plot (on left) are quite similar, which is reflected in the low flat norm distances between them.
      The network geometries on the right plots are more dissimilar and hence the flat norm distances are high.}
    \label{fig:comparing}
\end{figure}
The top two plots in Fig.~\ref{fig:comparing} show two regions characterized by $\epsilon=0.001$ and almost similar $|T|/\epsilon$ ratios. This indicates that the length of network scaled to the region size is almost equal for the two regions. From a mere visual perspective, we can conclude that the first pair of network geometries resemble each other where as the second pair are fairly different. This is further validated from the results of the flat norm distance between the network geometries computed with the scale $\lambda=1000$, since the first case produces a smaller flat norm distance compared to the latter. The bottom two plots show another example of two regions with almost similar $|T|/\epsilon$ ratios and enable us to infer similar conclusions. The results strengthens our case of using flat norm as an appropriate measure to perform a local comparison of network geometries.
We can choose a suitable $\gamma > 0$ in Definition \ref{def:strsmlrggrfs} which differentiates between these example cases and use the proposed flat norm distance metric to identify structurally similar network geometries.
However, the choice of $\gamma$ has to be made empirically.
This necessitates a statistical study of randomly chosen local regions in different sections of the networks, which is performed in the Section~\ref{sec:stats}.

\newcommand*{\dsp}                  {\displaystyle}
\newcommand*{\mc}[1]                {\mathcal{#1}}

\newcommand*{\bd}                   {\partial}
\newcommand*{\iimplies}             {\Longrightarrow}
\newcommand*{\iiff}                 {\Longleftrightarrow}
\newcommand*{\pprime}               {\prime\prime}
\newcommand*{\tr}                   {\scalebox{0.5}{$\mathsf{T}$} }

\newcommand*{\Half}                 {\frac{1}{2}}
\newcommand*{\half}                 {\sfrac{1}{2}}

\newcommand*{\abs}[1]               {|{#1}|}
\newcommand*{\Abs}[1]               {\left|{#1}\right|}
\newcommand*{\aBs}[1]               {\big|{#1}\big|}

\newcommand*{\norm}[1]              {\left\| {#1} \right\|}
\newcommand*{\inner}[2]             {\left\langle  {#1}, {#2} \right\rangle}

\newcommand*{\tld}[1]               {\widetilde{#1}}
\newcommand*{\hhat}[1]              {\widehat{#1}}
\newcommand*{\bbar}[1]              {\overline{#1}}

\newcommand*{\e}                    {\delta}
\newcommand*{\perturb}              {\rightsquigarrow}
\newcommand*{\perturbe}             {\overset{\e}{\rightsquigarrow}}
\newcommand*{\perturbx}[1]          {\overset{#1}{\rightsquigarrow}}
\newcommand*{\perturbeplus}         {\overset{\e^{\scalebox{0.4}{+}}}{\rightsquigarrow}}

\newcommand{\splus}                 {\scalebox{0.6}{+}}
\newcommand{\ssplus}                {\scalebox{0.4}{+}}
\newcommand{\sminus}                {\scalebox{0.6}{-}}

\definecolor{AzureBlue}             {RGB}{0, 127, 255}
\definecolor{Pink}                  {RGB}{255, 81, 199}
\definecolor{DeepPink}              {RGB}{255, 20, 147}
\definecolor{BrightRed}             {RGB}{255, 32, 50}
\definecolor{SpringGreen}           {RGB}{0, 178,	88}
\definecolor{GrassGreen}            {RGB}{126, 211,	33}
\definecolor{ForestGreen}           {RGB}{34,139,34}
\definecolor{RoyalBlueDark}         {RGB}{0, 35, 102}
\definecolor{NavyBlue}              {RGB}{0, 0, 128}
\definecolor{RoyalBlue}             {RGB}{65, 105, 225}
\definecolor{RoyalBlueOld}          {cmyk}{1, 0.50, 0, 0}
\definecolor{NavyBlueLight}         {RGB}{25, 116, 210}
\definecolor{DarkYellow}            {RGB}{210, 177, 8}
\definecolor{DarkOrange}            {RGB}{245, 166, 35}
\definecolor{DodgerPurple}          {RGB}{144, 19, 254}
\definecolor{Aqua}                  {RGB}{92, 225, 225}

\section{Notion of Stability for the Flat Norm Distance}
\label{sec:fn-stability}

We now investigate two approaches to define a notion of stability for the flat norm distance.
For any measure of discrepancy between objects, the notion of \emph{stability} is not only desirable from a theoretical standpoint but is also necessary for practical applications. 
The comparison metric is said to be \emph{stable} if \emph{small changes in the input geometries lead to only small changes in the measured discrepancy}.
But such a formulation introduces a ``chicken and an egg'' problem---in order to evaluate the stability of a proposed metric we need an alternative baseline metric to measure the small change in the input. 
Of course, the baseline metric should be stable as well. 
This constitutes the first approach. 
The alternative approach is to derive directly an upper bound on the proposed metric under well-defined controlled perturbations of input geometries.

The Hausdorff distance metric $\mathbb{D}_H$, which has been extensively used in the literature for comparing geometrically embedded networks, is stable in the former sense,
i.e., a small change in the input geometry will result in only a small change in $\mathbb{D}_H$.
Hence, it is a natural choice for use as the baseline metric.
At the same time, the Hausdorff distance is not sensitive enough to adequately measure small changes in the input geometries. 
Let us consider a $\e$-ball around each node in the network for a chosen perturbation radius $\e > 0$.
We then uniformly sample a point in each circular region and use them as the perturbed embeddings of the nodes.
The Hausdorff distance will change if the perturbation is either \emph{significantly large} to overshadow the current value by moving some node far enough,
or it is \emph{very specific} and affects the maximizer nodes of $\mathbb{D}_H$.
As will be shown in the next subsection (Sec.~\ref{subsec:FN-VS-HDF}) on a few simple counter-examples and the real-world networks introduced in the previous section (Sec.~\ref{subsec:implementation:fn-power-grids}),
knowing the value of the Hausdorff distance or how it changed is not enough to infer any useful information about the flat norm distance, and vice versa.
Even though our examples are 1-dimensional, the behavior can be observed in any dimension.

In the subsequent subsection (Sec.~\ref{subsec:FN-BOUND}),
under some mild assumptions about the scale and the radius of perturbation, 
we derive an upper bound on the flat norm distance for the case of simple piecewise linear curves in $\real^2$.
We formalize these curves as a special class of integral 1-currents, namely the piecewise linear currents,
and study a class of (positive) $\e$-perturbations of their nodes. 
It allows us to track the change of the components of the flat norm distance while perturbing each node one by one,
which in turn allows us to construct a non-trivial upper bound on $\mathbb{F}_{\lambda}$ between the original 1-current and its final perturbed version.

\subsection{Comparing the Hausdorff and the flat norm distance}
\label{subsec:FN-VS-HDF}
We refer the reader to standard textbooks on geometric measure theory \cite{Federer1969,Morgan2016} for the formal definition of integral currents and other related concepts.
For our purposes, it is sufficient to consider an integral current as a collection of oriented manifolds with or without boundary, and with integer multiplicities as well as orientations for each submanifold.
Recall from Sec.~\ref{subsec:multiscale-flat-norm} that $\mc{C}_d(\real^{d + 1})$ denotes the set of all oriented $d$-dimensional integral currents ($d$-current) embedded in $\real^{d + 1}$,
and $\supp(T) \subset \real^{d+1}$ is the $d$-dimensional support for $T \in \mc{C}_d(\real^{d + 1})$.
Let $X \in \mc{C}_d(\real^{d + 1})$ with $(d-1)$-boundary $\bd X \in \mc{C}_{d - 1}(\real^{d + 1})$,
then the set of all $d$-currents embedded in $\real^{d+1}$ spanned by the boundary of $X$ 
is denoted as $\mc{C}_{d}[\bd X; \real^{d + 1}] \subset \mc{C}_{d}(\real^{d + 1})$,
or simply $\mc{C}_{d}[\bd X]$ if the embedding space is clear from the context.

Let $T_0, T_1 \in \mc{C}_d(\real^{d + 1})$ be two integral $d$-currents in $\real^{d + 1}$, 
and $\norm{v - u}_d$  be the Euclidean distance between $u,v \in \real^d$.
The Hausdorff distance between currents $T_1$ and $T_0$ is given by
\begin{align*}
    \mathbb{D}_H(T_1, T_0)
    &=   \max \left\lbrace 
                \sup\limits_{v \in \supp{T_0}} \mathbb{D}(v, T_1),
                \sup\limits_{v \in \supp{T_1}} \mathbb{D}(v, T_0),
    \right\rbrace
\end{align*}
where $\mathbb{D}(v, T)$ is  the distance from a point to a current given as
\begin{align*}
    \mathbb{D}(v, T) = \inf\limits_{u \in \supp{T}} \norm{v - u}_d\,.
\end{align*}

Let $X = T_1 - T_0 - \bd S$ be the $d$-component of the flat norm decomposition in Eq.~\eqref{eq:flatnorm-def}.
Note that $\bd X = \bd T_1 - \bd T_0 - \bd \bd S = \bd T_1 - \bd T_0$,
i.e., $X \in C_d[\bd T_1 - \bd T_0]$,
and it can be rendered to zero only if $\bd T_1 - \bd T_0 \equiv 0$. 
Hence, the volume of the minimal $d$-current spanned by $\bd T_1 - \bd T_0$ provides a lower bound on $\mathbb{F}_{\lambda}(T_1 - T_0)$.
On the other hand the Hausdorff distance between boundaries  $\mathbb{D}_H(\bd T_1, \bd T_0)$ doesn't provide any meaningful insights about the actual value of  $\mathbb{D}_H(T_1, T_0)$.
This prompts us to suspect that the Hausdorff distance between two $d$-currents does not have any meaningful relations with the value of the flat norm distance.
In fact, we expect $\mathbb{F}_{\lambda}(T_1 - T_0)$ to be more sensitive to perturbations of input geometries.
Moreover, as can be seen from the examples below and the discussion in the following Section \ref{subsec:FN-BOUND},
given that the scale $\lambda > 0$ is small enough,
when $T_1$ is perturbed within a $\e$-tube
the range of incurred changes of the flat norm distance depends, mainly, on the size of perturbation $\e > 0$ and the input volume $\mathbf{V}_d(T_1)$, and not on $\mathbb{D}_H(T_1, T_0)$.
Although the examples in this Section are given for 1-currents in $\real^2$ for illustrative purposes, this conclusion holds in the general case of $d$-currents as well.

\begin{example}[Fixing $\mathbb{D}_H(\tld{T}_1, T_0)$] \label{eg:Hdf-example-1}
Let $T_1$ and $T_0$ be two 1-currents in $\real^2$
with common boundaries, $\bd T_1 = - \bd T_0$,
and let $H = \mathbb{D}_H(T_1, T_0)$ be the Hausdorff distance between them.
For some small $\e > 0$, consider $\tld{T}_1$---a perturbation of  $T_1$  within  an $\e$-tube---such that the boundaries and the Hausdorff distance do not change:
\begin{align*}
\begin{array}{ccccc}\dsp
    \bd \tld{T}_1 = \bd T_1 
    &\text{ and }&\dsp
    \mathbb{D}_H(\tld{T}_1, T_0) = H.
\end{array}
\end{align*}
Note that $\mathbb{D}_H(\bd T_1, \bd T_0) = \mathbb{D}_H(\bd T_1, \bd \tld{T}_1) = 0$.
See Fig.~\ref{fig:currents:example-currents-and-neighborhoods} for  examples of perturbations $\tld{T}_1$.

\begin{figure}[ht!]
    \begin{minipage}{.50\textwidth}
                \centering
                \includegraphics[width=0.95\textwidth]{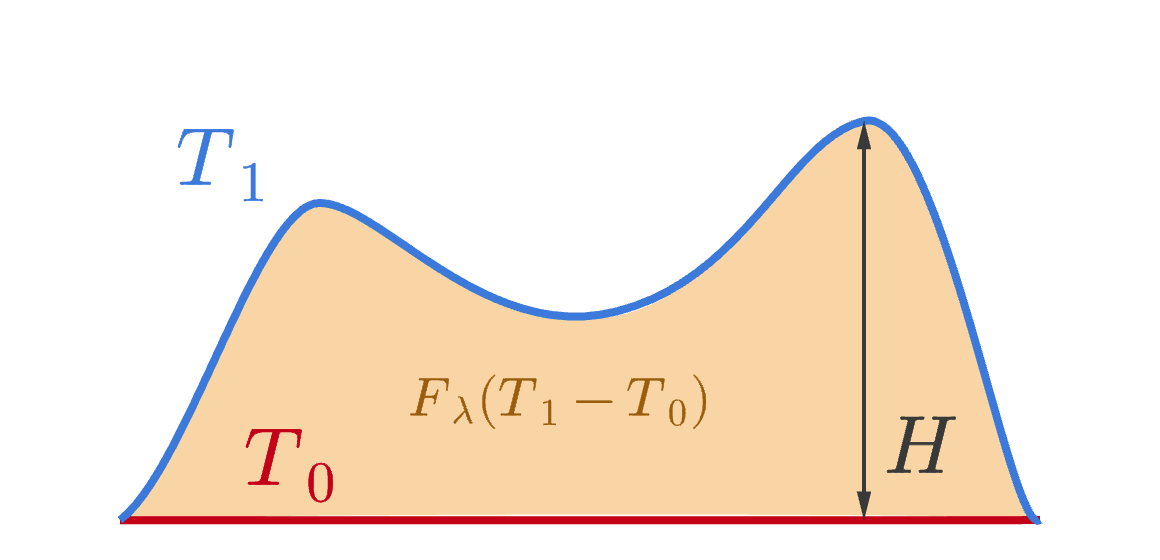} \\
                \includegraphics[width=0.9\textwidth]{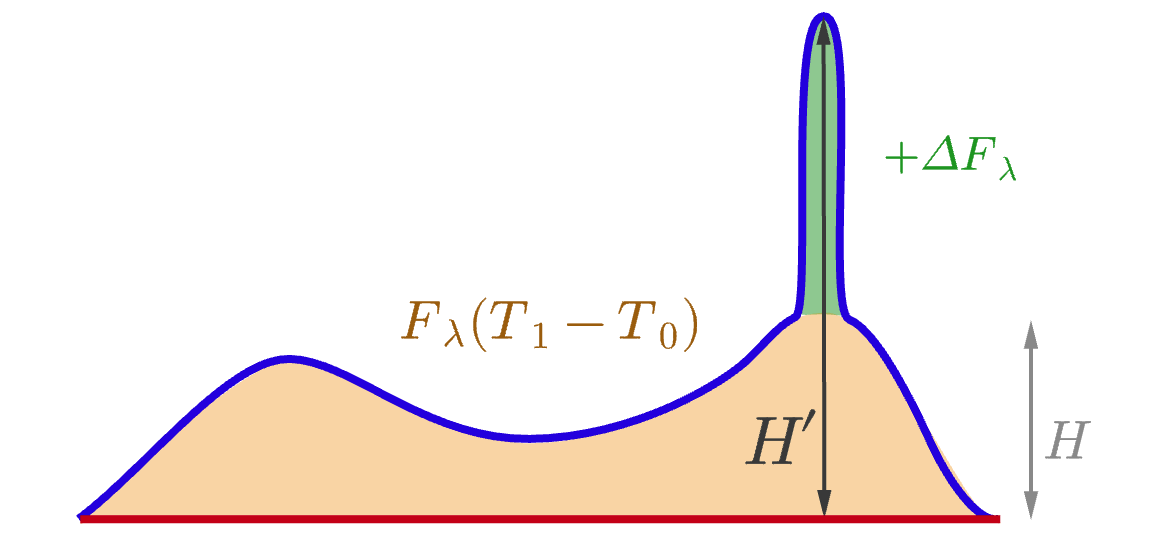}
    \end{minipage}%
    \begin{minipage}{.50\textwidth}
                \centering
                \includegraphics[width=0.95\textwidth]{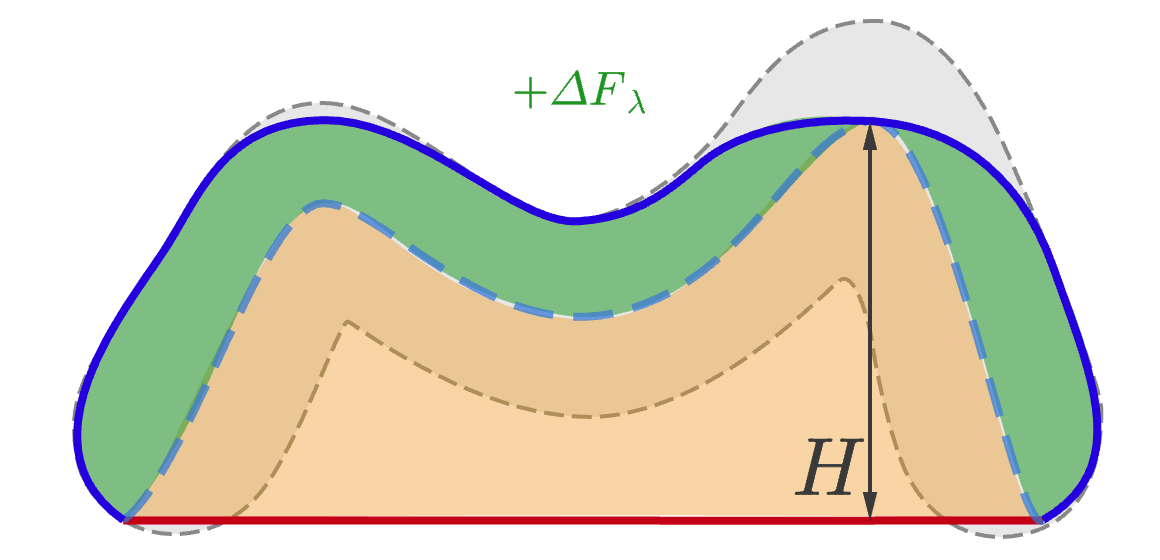} \\
                \includegraphics[width=0.95\textwidth]{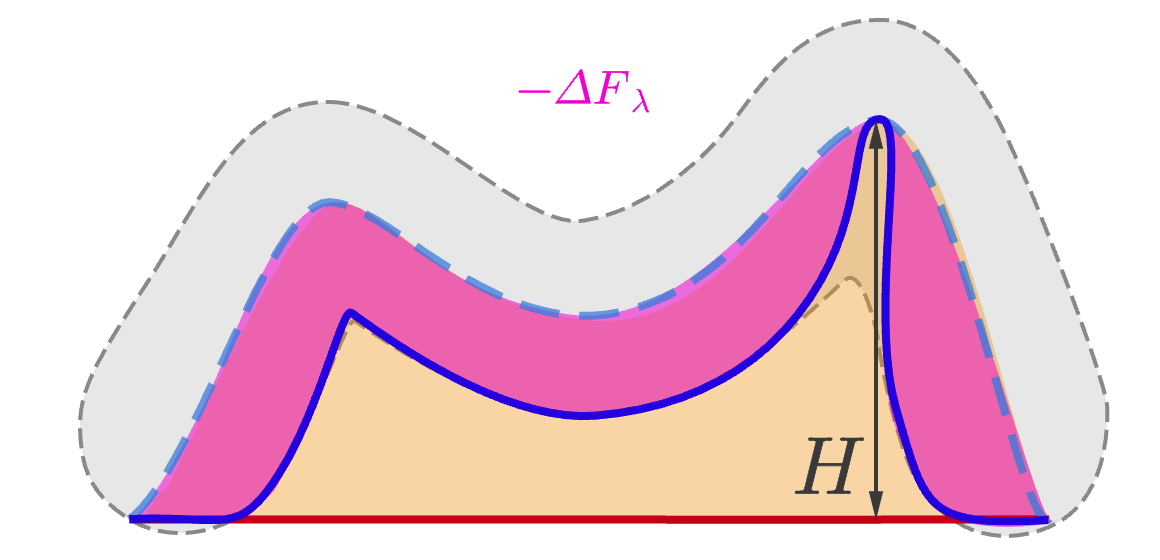}
    \end{minipage}\\ 
    \caption{
        \textbf{Top Left:}
            The input curves \textcolor{AzureBlue}{$T_1$} and \textcolor{BrightRed}{$T_0$} with shared endpoints
            and Hausdorff distance $\mathbb{D}_H(T_1, T_0) = H$.
            At a small enough scale $\lambda > 0$, the flat norm distance \textcolor{orange}{$\mathbb{F}_{\lambda}(T_1 - T_0)$} corresponds to the orange patch in between them.
        \textbf{Right: Example 1.}
            The example perturbations \textcolor{blue}{$\tld{T}_1$} (solid  blue)  that lie within a $\e$-neighborhood (gray) of \textcolor{AzureBlue}{$T_1$}(dashed blue).
            The Hausdorff distance between \textcolor{blue}{$\tld{T}_1$} and \textcolor{BrightRed}{$T_0$} remains same, i.e., $H$.
            The green patch (Right, Top) captures the  \textcolor{ForestGreen}{increment $\Delta \mathbb{F}$},
            and beneath it (Right, Bottom) the pink area corresponds to the \textcolor{DeepPink}{decrement $\Delta \mathbb{F}$}.
        \textbf{Bottom Left: Example 2.}
            The example perturbation of $T_1$ that moves \emph{only} the maximizer of $\mathbb{D}_H(T_1, T_0)$
            further away from $T_0$ so that $\mathbb{D}_H(\tld{T}_1, T_0) = H^{\prime} >> H$.
            The flat norm distance increases by  \textcolor{ForestGreen}{$\Delta \mathbb{F}$} that corresponds to the area of the created spike, which can be arbitrary small.
    }
\label{fig:currents:example-currents-and-neighborhoods}
\end{figure}

We could have cases where $\tld{T}_1$ lies mostly at the upper envelope of this $\e$-tube,
causing the flat norm distance to increase by $\Delta \mathbb{F}_{\lambda} = \Abs{\mathbb{F}_{\lambda}(T_1 - T_0) - \mathbb{F}_{\lambda}(\tld{T}_1 - T_0)}$ (highlighted in green),
or mostly at the lower envelope causing a decrease in the flat norm distance, respectively (highlighted in pink).
In both cases, one would expect the ideal measure of discrepancy between $\tld{T}_1$ and $T_0$ to change significantly as well (compared to the one between $T_1$ and $T_0$).
The flat norm distance accurately captures all such changes (to keep the example simple, we consider the default flat norm distance and not the normalized version).
At the same time, both such variations could have the same Hausdorff distance $H$ from $T_0$ as $T_1$, which completely misses all the changes applied to $T_1$ in either case.
\end{example}

\begin{example}[Fixing $\mathbb{F}_{\lambda}(\tld{T}_1 - T_0)$] \label{eg:Hdf-example-2}
A modification of this example can illustrate the other extreme case---when Hausdorff distance changes by a lot but the flat norm distance does not change much at all, see bottom row in Fig.~\ref{fig:currents:example-currents-and-neighborhoods}.
Consider moving \emph{only} the highest point on $T_1$ further away from $T_0$ so that Hausdorff distance becomes $H^{\prime} >> H$,
as shown on the bottom left figure of Fig.~\ref{fig:currents:example-currents-and-neighborhoods}.
We keep $T_1$ a connected curve, thus creating a sharp spike in it.
While the Hausdorff distance between the curves has increased dramatically, the flat norm distance sees only a minute increase as measured by the area under the spike.
Moreover, the increment $\Delta \mathbb{F}_{\lambda}$ can be decreased to almost zero by narrowing the spike.
Once again, the flat norm distance accurately captures the intuition that the curves have \emph{not} changed much when just a single point moves away while the rest of the curve stays the same.
Hence the flat norm provides a more robust metric that better captures significant changes while maintaining stability to small perturbations (also see Section \ref{subsec:FN-BOUND} for theoretical bounds).
\end{example}

\subsection{Empirical study} \label{ssec:Hdf-example-empirical}
We observe similar behavior to those illustrated by the theoretical example (Fig.~\ref{fig:currents:example-currents-and-neighborhoods}) in our computational experiments.
Fig.~\ref{fig:flatnorm-stability-empirical} shows scatter plots denoting empirical distribution of percentage deviation of the two metrics from the original values $\left(\%\Delta\mathbb{D}_{H},\%\Delta\widetilde{\mathbb{F}}_{\lambda}\right)$ for a local region. The perturbations are considered for three different radii shown in separate plots. We note that the percentage deviations in the two metrics are comparable in most cases. In other words, neither metric behaves abnormally for a small perturbation in one of the networks.
\begin{figure}[htb!]
    \centering
    \includegraphics[width=0.95\textwidth]{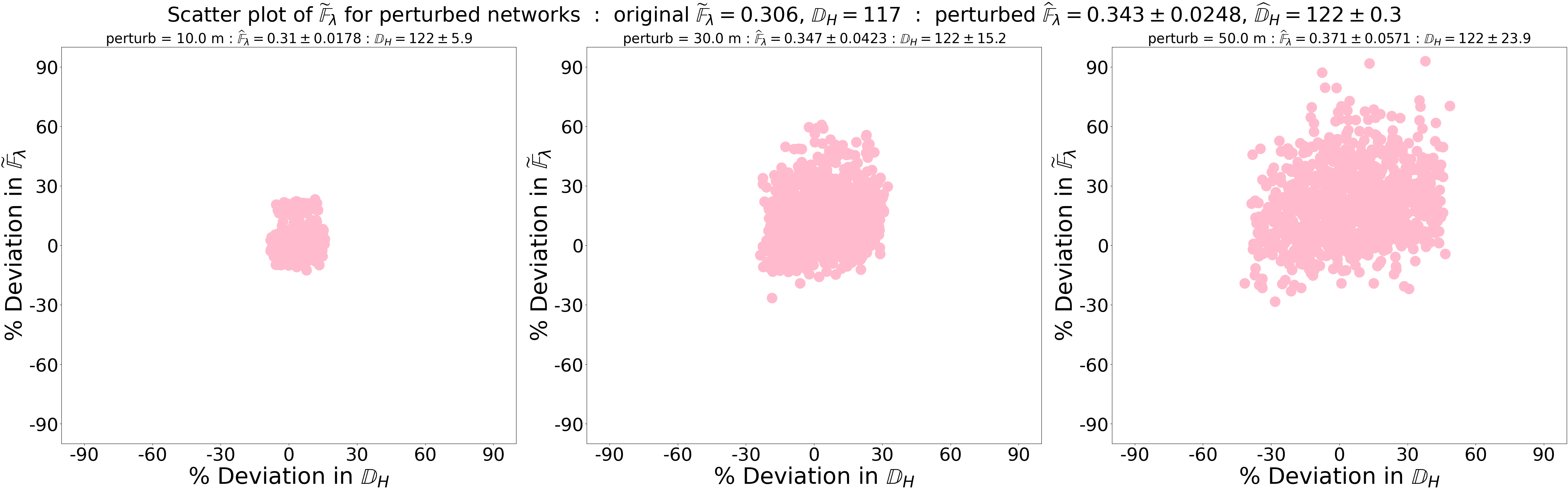}
    \caption{Scatter plots showing the effect of network perturbation on the normalized flat norm and Hausdorff distances for a local region.
      The percentage deviation in the metrics for the perturbations is shown along each axis.
      We do not observe significantly large deviations in any one metric for a given perturbation.}
    \label{fig:flatnorm-stability-empirical}
\end{figure}

Next, we compare the sensitivity of the two metrics to outliers.
Here, we consider a single random node in one of the networks and perturb it.
Fig.~\ref{fig:flatnorm-outlier-empirical} shows the sensitivity of the metrics to these outliers.
The original normalized flat norm and Hausdorff distance metrics are shown by the horizontal and vertical dashed lines respectively.
The points along the horizontal dashed line denote the cases where the Hausdorff distance metric is more sensitive to the outliers, while the normalized flat norm metric remains the same.
These cases occur when the perturbed random node determines the Hausdorff distance, similar to the second Example where Hausdorff distance increased from $H$ to $H^{\prime}$.
On the flip side, the points along the vertical dashed line denote the Hausdorff distance remaining unchanged while the normalized flat norm metric shows variation.
Just as in the theoretical example (Fig.~\ref{fig:currents:example-currents-and-neighborhoods}), such variation in the normalized flat norm metric implies a variation in the network structure.
However, such variation is not captured by the Hausdorff distance metric.
Hence, our proposed metric is capable of identifying structural differences due to perturbations while remaining stable when widely separated nodes (which are involved in Hausdorff distance computation) are perturbed.
The other points which are neither on the horizontal nor the vertical dashed lines indicate that either metric can identify the structural variation due to the perturbation.

\begin{figure}[htb!]
    \centering
    \includegraphics[width=0.95\textwidth]{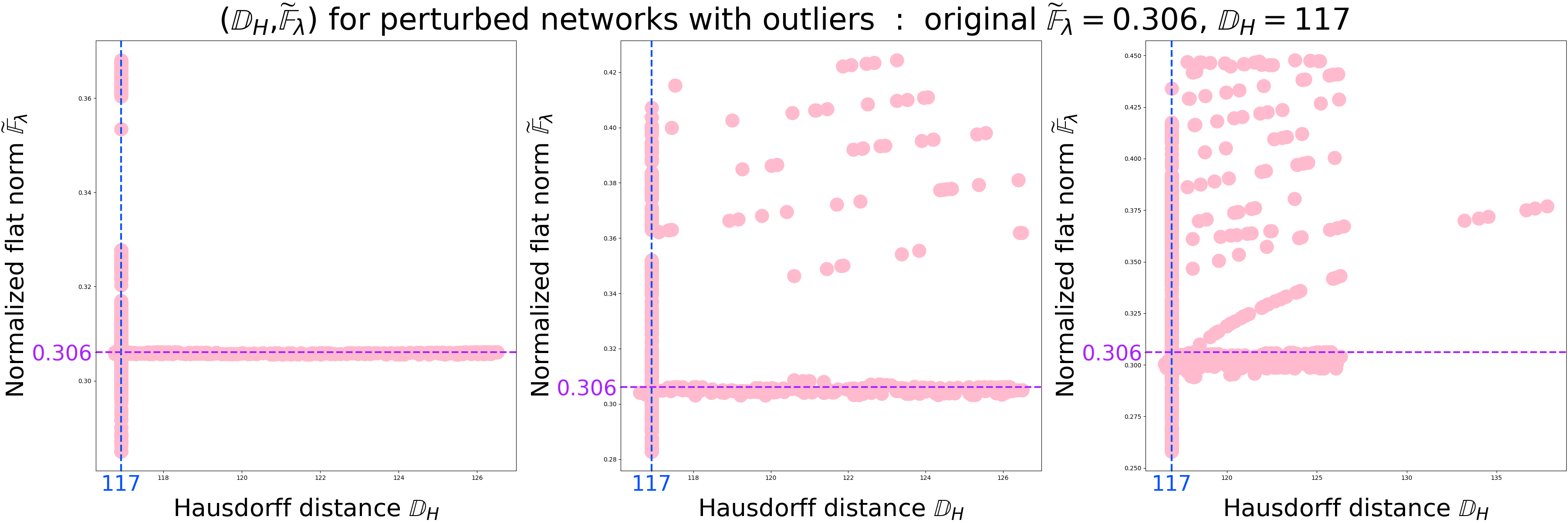}
    \caption{Scatter plots showing the effect of a few outliers on normalized flat norm and Hausdorff distance for a local region. The original normalized flat norm and Hausdorff distance are highlighted by the dashed horizontal and vertical lines. We observe multiple cases where the Hausdorff distance is more sensitive to outliers compared to $\tilde{\mathbb{F}}_{\lambda}$.}
    \label{fig:flatnorm-outlier-empirical}
\end{figure}

\subsection{Stability of the Flat Norm}
\label{subsec:FN-BOUND}

The results of the previous subsection are applicable in all dimensions, i.e., one cannot expect to bound changes in the flat norm by constant, or even polynomial, functions of the changes in the Hausdorff distance.
In this subsection, we adopt a more direct approach to the investigation of the stability of our discrepancy measure between geometric objects. 
Our goal is to construct an upper bound on $\mathbb{F}_{\lambda}$  from the bottom up based only on the input geometries and an appropriately defined radius of perturbation. 
To this end, we consider simple piecewise linear (PWL) curves spanned by a pair of points with no self-intersections embedded in $\real^2$.
Here, \emph{simple} means that there is no self intersection or branching in the curve that connects two points.
Despite its apparent simplicity,  this class of curves is of particular interest, since they can potentially approximate any continuous non-intersecting curve in $\real^2$.
More directly, power grid networks that form the main motivation for our work can be seen as collections of such simple curves.
Although, the results of this subsection are proven only for a pair of simple PWL currents, the empirical findings on the real-world networks, presented in the next section (Sec.~\ref{sec:stats}) comply surprisingly well with the upper bounds established for simple curves (see Fig.~\ref{fig:flatnorm-hists-lambdas}).

We conceptualize a simple piecewise linear curve $\mc{T} \subset \real^2$ between points $s$ and $t$ as the 1-current $T$ embedded in $\real^2$ that is equipped with an edge set $\mathbf{E}(T) = \left\lbrace e_{1}, e_{2} , \ldots, e_{n-1}, e_{n} \right\rbrace$ given by the linear segments of $\mc{T}$,
and a node set $\mathbf{N}(T) = \left\lbrace v_0, v_1, \ldots, v_{n - 1}, v_{n} \right\rbrace$, where $v_0 = s$ and $v_n = t$, and $v_i = e_i \cap e_{i+1}$ for $i=1,\dots,n-1$.
Such currents are defined as a formal sum of their edges $T = e_1 + \ldots + e_n$, and can be thought of as discretized linear approximations of simple continuous curves in $\real^2$.
We refer to this type of currents as \emph{PWL currents},
and denote the set of all PWL currents with $n$ edges spanned by $s$ and $t$ (i.e., starting at $s$ and ending and $t$) as $\mc{L}_n[s, t]$ or $\mc{L}_n[s, t; \real^2]$.
A \emph{subcurrent} of $T \in \mc{L}_{n}[s, t]$ spanned by nodes $v_i$ and $v_j$, where $i < j$, is denoted $T[v_i, v_j] \subseteq T$, and $\mathbf{E}\big( T[v_i, v_j] \big) = \left\lbrace e_{i + 1} , \ldots, e_{j} \right\rbrace$,
which implies $T[v_i, v_j] \in \mc{L}_{j - i}[v_i, v_j]$.
The length of $T$ is given by $\abs{T} = \sum_{j = 1}^{n} \abs{e_j}$.

We start with a PWL current $T_0$ in $\real^2$ and its copy $T_1$.
Next, we consider a sequence of perturbations of $T_1$'s nodes within some $\e$-neighborhood to obtain $\tld{T}_1$,
while tracking the components of the flat norm distance between the original and the perturbed copy.
Recall that the components of the flat norm distance between generic inputs $T_0$ and $T_1$ are the perimeter of the unfilled void $\abs{T_1 - T_0 - \bd S}$ and the area of a 2-current $S$ given as $\mc{A} \left(S\right) = \mathbf{V}_2\left(S\right)$, see Fig.~\ref{fig:demo-flatnorm-basic}.
We refer to them as the \emph{length} (or \emph{perimeter}) and the \emph{area} components of flat norm distance, respectively:
\begin{align}
\label{eq:flat-norm-1-currents}   
    \mathbb{F}_\lambda\left(T_1 - T_0\right) 
    &= 
        \min\limits_{S \in \mc{C}_2(\real^2) }
        \big\{ \Abs{T_1 - T_0 -\bd S}  + \lambda \mc{A} \left(S\right) \big\}.
\end{align}

Since $T_0$ and $T_1$ are identical to start with, the flat norm distance between them is zero (for $S = \emptyset$).
Now let us take a look at the components of the flat norm distance $\mathbb{F}_\lambda(\tld{T}_1 - T_0)$ between the original current and the perturbed copy $\tld{T}_1$.
Let $X \subset \real^2$ be the 2-dimensional void with the boundary given by $\tld{T}_1 - T_0$,
and $S \in C_2(\real^2)$ be a 2-current that fills in, possibly partially, the void $X$.
The area component in the optimal decomposition of $\mathbb{F}_\lambda(\tld{T}_1 - T_0)$ is  bounded by the area of $X$, $\mc{A}(S) \leq \mc{A}(X)$,
and is maximized when the void is fully filled, i.e., $S \equiv X$.
On the other hand, the length component can be, potentially, arbitrarily large due to the complexity of $\bd S$:
\begin{align}
\label{eq:pwl-fn-length-bound}    
  0 \leq \abs{\tld{T}_1 - T_0 -\bd S} 
  \leq \abs{\tld{T}_1 - T_0} + \abs{\bd S}.
\end{align}

Here we make an important assumption about the values of parameters $\lambda > 0$ and $\e > 0$,
formalized below, that allows us to circumvent the potential unboundedness of the perimeter component. 
We mention that the problem of identifying the ranges of values of parameters that fit the assumption is out of the scope of this paper, and will be the focus of future research.

\begin{assumption}[Filled voids]
\label{ref:main-assumption}
  For any original current $T_0$ embedded in $\real^2$,  
  the scale parameter $\lambda = \lambda(T_0) > 0$ is small enough 
  such that for any size of the perturbation $\e = \e(\lambda, T_0) > 0$ taken within some range $0 < \e < M_{\lambda}(T_0)$,
  \textbf{the optimal decomposition of the flat norm distance between $T_0$ and its consecutive perturbations always fills in all the voids that appear as a result of the perturbations}.
\end{assumption}
We refer to Figure \ref{fig:currents:example-currents-and-neighborhoods} for an illustration, in which the current $T_1$ in blue can be considered as a perturbed version of the original current $T_0$ in red.
Assumption \ref{ref:main-assumption} specifies that the optimal flat norm decomposition shown in orange always fills in the entire space between these two currents.
We note that this assumption does not introduce any additional challenges in implementing the flat norm distance, 
since we always can find a small enough $\lambda$ by scaling it by one-half until all gaps between the input geometries are filled.

Given parameters $\lambda > 0$ and $\e > 0$ under Assumption \ref{ref:main-assumption},
the minimization objective of $\mathbb{F}_{\lambda}(\tld{T}_1 - T_0)$ in Eq.~\eqref{eq:flat-norm-1-currents} is achieved by a 2-current $S \equiv X$, where $X$ is a 2D-void in between $\tld{T}_1$ and $T_0$. 
Hence, $\bd S = \bd X = \tld{T}_1 - T_0$, i.e., $S \in \mc{C}_2[\bd X ] = \mc{C}_2[\tld{T}_1 - T_0]$.
This implies that the length component in Eq.~\eqref{eq:pwl-fn-length-bound} renders to 0---its minimum value.
Conversely, if we would leave the void unfilled, i.e., $S \equiv 0$,
then the area component as well as its boundary become zero, $\mc{A}(S) = 0$ and $\abs{\bd S} = 0$,
while the length component is equal to the void's perimeter: $\abs{\tld{T}_1 - T_0 - \bd S} = \abs{\tld{T}_1 - T_0} = \abs{\bd X}$.
Hence, under Assumption \ref{ref:main-assumption}, the optimization objective of $\mathbb{F}_{\lambda}(\tld{T}_1 - T_0)$ reduces to the scaled area of $S$ for $S \in \mc{C}_2[\tld{T}_1 - T_0]$: 
\begin{align}
\label{eq:pwl-fn-objective}     
  \mathbb{F}_{\lambda}(\tld{T}_1 - T_0) 
   &= \min\limits_{S \in \mc{C}_2[\bd X] }
           \left\{ \abs{\tld{T}_1 - T_0 -\bd S}  + \lambda \mc{A} \left(S\right) \right\}
 = \lambda \mc{A}(S).
\end{align}

This collapsed minimization objective implies that for the scale and the perturbation radius given by our assumption, 
the void produced by perturbing the copy of $T_0$ 
is filled in by a 2-current $S$ spanned by $\tld{T}_1 - T_0$,
such that the scaled area of $S$ is upper bounded by its (non-scaled) perimeter:
\begin{align}
\label{eq:pwl-fn-upper-bound}   
  \mathbb{F}_{\lambda}(\tld{T}_1 - T_0) = 
  \lambda \mc{A}(S) < \abs{\bd S} = \abs{\tld{T}_1 - T_0} .
\end{align}

\subsubsection{\texorpdfstring{$\e$}{Delta}-perturbations of PWL currents}
\label{subsec:PWL-PERTURB}

We consider a PWL current $T_0 \in \mc{L}_{n}[s, t]$ spanned by $s, t \in \real^2$,
called the \emph{original current},
and its copy $T_1 = T_0$, e.g., see Fig.~\ref{fig:pwl-perturbation-00}.
Obviously, $\mathbb{F}_\lambda\left(T_1 - T_0\right) = 0$ at any scale $\lambda > 0$. 

\begin{figure}[ht!]
   \centering
   \includegraphics[width=0.60\textwidth]{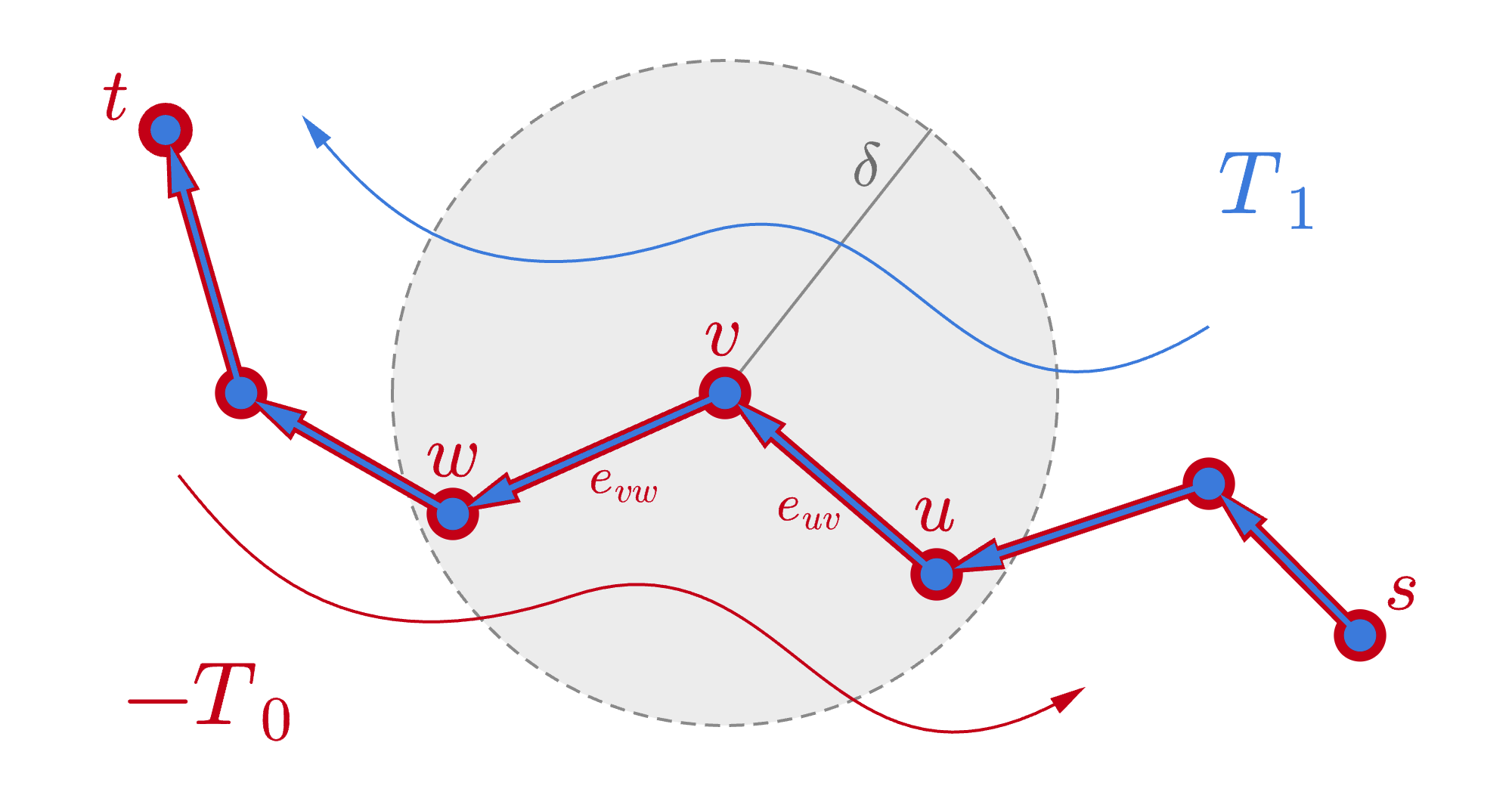}
   \caption{
        The original PWL current $\textcolor{BrightRed}{T_0}  \in \mc{L}_n[s, t]$ and its copy \textcolor{AzureBlue}{$T_1$}.
        The arrows show the orientations in $T_1 - T_0$,
        and the gray disk is the $\e$-ball \textcolor{gray}{$\mc{B}_{\e}(v)$} centered at an interior vertex $v$.
   }
\label{fig:pwl-perturbation-00}
\end{figure}

Let $v = v_i \in T_1$  be an interior  vertex of $T_1$, $1 \leq i \leq n - 1$.
Given a \emph{radius of perturbation} $\e > 0$, 
consider a mapping $v \to \tld{v}$ such that $\tld{v}$ stays within an open $\e$-ball centered at $v$, 
i.e., $\tld{v} \in \mc{B}_{\e}(v) = \left\{ x \in \real^2 \mid \norm{v - x}_2 < \e \right\}$.
Let $u = v_{i-1}$ and $w = v_{i + 1}$ be the adjacent nodes of $v$,
and let the corresponding incident edges be $e_{uv} = (u, v) = (v_{i-1}, v_{i}) = e_{i}$ and $e_{vw} = (v, w) = (v_{i}, v_{i + 1}) = e_{i + 1}$.
Then, let $\tld{e}_{uv} = (u, \tld{v})$ and $\tld{e}_{vw} = (\tld{v},  w)$
be the edge-like currents connecting the neighbors of $v$ to its perturbation $\tld{v}$, see Fig.~\ref{fig:pwl-perturbation-QQ}. 
We define a \textbf{$\boldsymbol{\e}$-perturbation} of $v \in T_1$ as the collection of mappings of $v$ along with its two connected edges under the assumption that the new edges $\tld{e}_{uv}$ and $\tld{e}_{vw}$ \emph{\textbf{do not cross any of the original edges}}.
\begin{align*}
  v \perturbe \tld{v} = \big\{
          v \to \tld{v}, e_{uv} \to \tld{e}_{u v}, e_{vw} \to \tld{e}_{v w} 
    \big\}
    ~\text{ subject to }~ \tld{v} \in \mc{R}_{\e}(v)
\end{align*}
where $\mc{R}_{\e}(v) \subseteq \mc{B}_{\e}(v)$ is the \emph{allowed region} of $\e$-perturbation
defined as a subregion of the $\e$-ball centered at $v$ that is in the direct line of sight of the vertices adjacent to $v$ (see Fig.~\ref{fig:pwl-perturbation-AA}):
\begin{align}
\label{eq:pwl-e-perturb-region}   
  \mc{R}_{\e}(v) = \mc{R}_{\e}(v; T_1) &= 
  \big\{ 
          \tld{v} \in \mc{B}_{\e}(v) \mid \tld{e}_{uv} \cap T_1 = \emptyset,\, \tld{e}_{vw} \cap T_1 = \emptyset 
  \big\}.
\end{align}

\begin{figure}[ht!]
   \centering
   \includegraphics[width=0.55\textwidth]{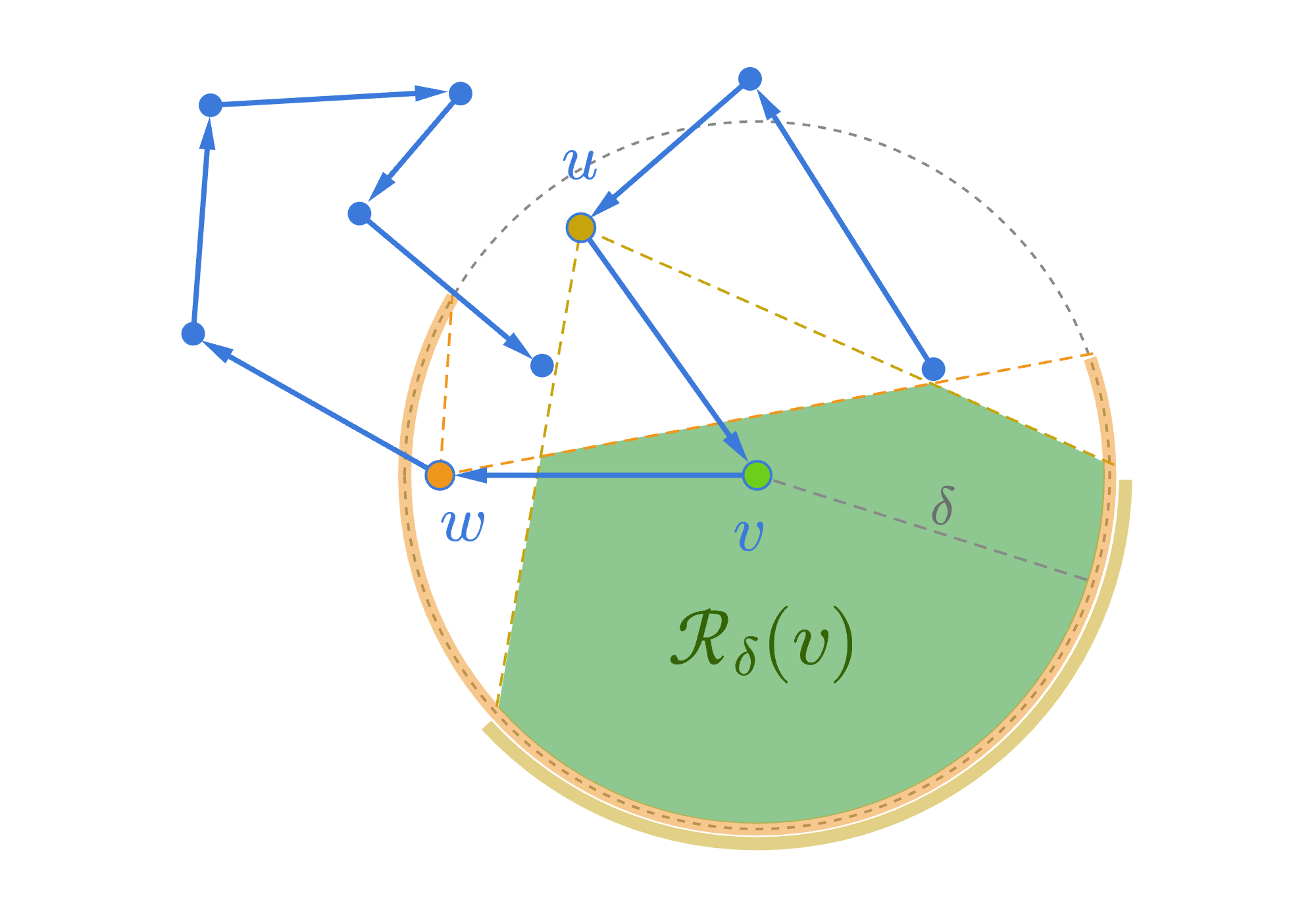}
   \caption{
     The \textcolor{DarkYellow}{yellow} and \textcolor{DarkOrange}{orange} circular arcs indicate the subregions of the $\e$-ball that are in the line of sight of nodes $u$ and $w$.
     The allowed region \textcolor{ForestGreen}{$\mc{R}_{\e}(v)$} of a $\e$-perturbation $v \perturbe \tld{v}$ is shown in green,
     in which it is guaranteed that the perturbed edges will not cross any of the original edges (see Eq.~\eqref{eq:pwl-e-perturb-region}). 
   } 
\label{fig:pwl-perturbation-AA}
\end{figure}

A {$\e$-perturbation} of a vertex $v \perturbe \tld{v}$ maps $T_1$ into $\tld{T}_1$ by acting on its node and edge sets as defined below,
and defines a $\e$-perturbation  $T_1 \perturbe \tld{T}_1$.
We say that the $\e$-perturbation of a current $T_1 \perturbe \tld{T}_1$ is \emph{induced} by $v \perturbe \tld{v}$.
Note that the non-overlapping conditions  in Eq.~\eqref{eq:pwl-e-perturb-region},
implies that $\tld{e}_{uv} \cap \tld{T}_1 = \tld{e}_{vw} \cap \tld{T}_1= \emptyset$,
which means that $\tld{T}_1$ is injective (has no self intersections). 
\begin{align*}
  \mathbf{N}(\tld{T}_1)
  &= v \perturbe \tld{v} \big( \mathbf{N}(T_1) \big)
   = \left\{ s,  v_1, \ldots, u, \tld{v}, w,  \ldots, v_{n-1}, t \right\} ~~~\text{ and }
   \\
  \mathbf{E}(\tld{T}_1)
  &= v \perturbe \tld{v} \big( \mathbf{E}(T_1) \big)
   = \left\{ e_s,  e_2, \ldots, \tld{e}_{u v}, \tld{e}_{v w}, \ldots, e_{n-1}, e_t \right\}.
\end{align*}

\noindent The $\e$-perturbations of boundaries are given by the following maps:
\begin{align*}
    s \perturbe \tld{s} &= 
    \big\{
        s \to \tld{s}, e_{s} \to \tld{e}_{s}
    \big\}
    \text{\hspace*{0.6in} and }
    &
    t \perturbe \tld{t} &= 
    \big\{
        t \to \tld{t}, e_{t} \to \tld{e}_{t}
    \big\},
\end{align*} 
where $\tld{s} \in \mc{R}_{\e}(s) = \mc{B}_{\e}(s) \cap \{ \tld{s} \mid \tld{e}_{s} \cap T_1 = \emptyset \}$ 
and $\tld{t} \in \mc{R}_{\e}(t) = \mc{B}_{\e}(t) \cap \{ \tld{t} \mid \tld{e}_{t} \cap T_1 = \emptyset \}$.

\subsubsection{Flat Norm of \texorpdfstring{$\e$}{delta}-perturbations}
\label{subsec:FN-PWL-PERTURB}

\begin{figure}[b!]
   \centering
   \includegraphics[width=0.50\textwidth]{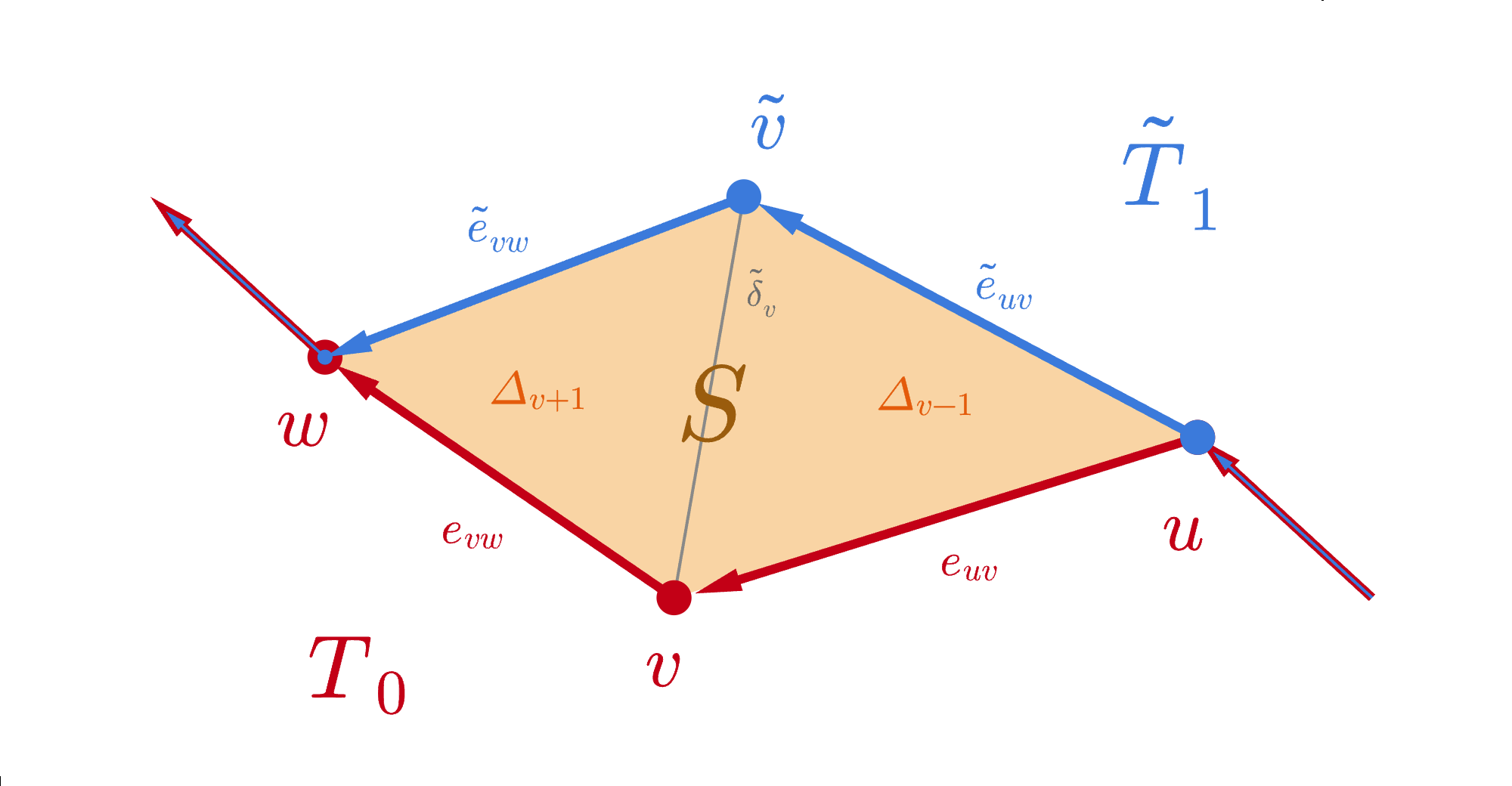}
   \caption{
       The quadrilateral $S = [u, v, w, \tld{v}]  = \Delta_{v + 1} + \Delta_{v - 1}$
       appears as result of a $\e$-perturbation $T_1 \perturbe \tld{T}_1$ induced by a perturbation of a non-boundary vertex $v \perturbe \tld{v}$.
   } 
\label{fig:pwl-perturbation-QQ}
\end{figure}

Having specified the necessary definitions and procedures,
we are now ready to prove our first result regarding a single $\e$-perturbation of an interior node.

\begin{theorem}
\label{thm:fn-bound:single-perturbation}
  Let $T_0 \in \mc{L}_n[s,t; \real^2]$ for $n \geq 2$ and $T_1 = T_0$ be its copy.
  Given a small enough scale $\lambda > 0$ and an appropriate radius of perturbation $\e > 0$,
  consider a $\e$-perturbation $T_1 \perturbe \tld{T}_1$ induced by perturbation of an interior node $v \perturbe \tld{v}$.
  Then the flat norm distance between $T_0$ and $\tld{T}_1$ 
  is upper bounded as follows:
  \begin{align*}
    \mathbb{F}_{\lambda}(\tld{T}_1 - T_0) \leq \frac{\lambda \e}{2} \aBs{ T_0[u, w] } 
  \end{align*}
  where $u$ and $w$ are vertices of $T_0$ adjacent to $v$.
\end{theorem}
\begin{proof}
    Recall that by the Assumption~\ref{ref:main-assumption},
    the optimal 2-current $S$ is spanned by $\tld{T}_1 - T_0$,
    and by the construction of $\mc{R}_{\e}(v)$ the new edges $\tld{e}_{uv}$ and $\tld{e}_{vw}$ do not overlap with the original current. 
    Then $\bd S =  \tld{T}_1 - T_0 =  e_{uv} +  e_{vw} - \tld{e}_{vw} - \tld{e}_{uv}$,
    which is the boundary of a quadrilateral spanned by the vertices $[u, v, w, \tld{v}]$, see Fig.~\ref{fig:pwl-perturbation-QQ}.
    To keep the notation simple, we just write $S = [u, v, w, \tld{v}] \in \mc{C}_2[\tld{T}_1 - T_0 ]$.
    Consider the diagonal $\tld{\e}_{v} = (v, \tld{v}) \in \mc{L}_1[v, \tld{v}]$ that connects $v$ to its perturbation.
    Its length is bounded by the radius of perturbation $\abs{\tld{\e}_{v}} \leq \e$ and
    it partitions $S$ into a pair of triangles $\Delta_{v + 1}$ and $\Delta_{v - 1}$:
    \begin{align}         
    \label{eq:pwl-perturb-quad-S}
      S = [u, v, w, \tld{v}] = [v, w, \tld{v}] + [\tld{v}, u, {v}] = \Delta_{v + 1} + \Delta_{v - 1}
    \end{align}
    with boundaries given by $\bd \Delta_{v + 1} = e_{vw} - \tld{e}_{vw} - \tld{\e}_{v}$
    and $\bd \Delta_{v - 1} = e_{uv}  + \tld{\e}_{v} - \tld{e}_{uv}$.
    The corresponding areas are 
    $\mc{A}(\Delta_{v + 1}) = \Half \abs{\tld{\e}_{v}} \abs{e_{v w}} \sin \theta_{v + 1}$ and 
    $\mc{A}(\Delta_{v - 1}) = \Half \abs{\tld{\e}_{v}} \abs{e_{u v}} \sin \theta_{v - 1}$,
    where $\theta_{i}$ is the angle between the diagonal $\tld{\e}_{v}$ and the original edge in the corresponding triangle, (see Fig.~\ref{fig:pwl-perturbation-QQ}).
    Note that both $\Delta_{v + 1}$ and $\Delta_{v - 1}$, and consequently $S$,  
    attain the largest possible area
    when the perturbed vertex $\tld{v}$ lands on the boundary of the $\e$-ball $\mc{B}_{\e}(v)$
    and $\tld{\e}_{v}$ is perpendicular to the original edges, i.e., $\theta_{v \pm 1} = \sfrac{\pi}{2}$.
    Let $\Delta_v$ be one of the triangles and $e_v = e_{vw}$ or $e_v = e_{uv}$ be its original edge,
    then the following are the upper bounds on the area of $\Delta_v$ and $S$:
    \begin{align}
    \label{eq:pwl-perturb-area-upper-bound}         
      \mc{A}(\Delta_{v}) = \Half \abs{\tld{\e}_{v}} \abs{e_{v}} \sin \theta
                      &\leq \frac{\e}{2} \abs{e_{v}}
      \\
    \label{eq:pwl-perturb-quad-area-upper-bound}    
      \mc{A}(S) = \mc{A}(\Delta_{v + 1}) + \mc{A}(\Delta_{v - 1})
              &\leq \frac{\e}{2} \big( \abs{e_{vw}}  + \abs{e_{uv}} \big)
    \end{align}
    
    Note that $e_{uv} + e_{vw}$ defines a subcurrent of $T_0$ spanned by $u$ and $w$,
    namely $T_0[u, w]$.
    Finally, recall that under our main Assumption~\ref{ref:main-assumption},
    $\mathbb{F}_{\lambda}(\tld{T}_1 - T_0)$ is given by $\lambda \mc{A}(S)$ (see Eq.~\eqref{eq:pwl-fn-objective}), 
    which together with the upper bound on $\mc{A}(S)$ in Eqn.~\ref{eq:pwl-perturb-quad-area-upper-bound} implies the main statement of the Theorem.
\end{proof}

Furthermore, observe that by the triangle inequality (see Fig.~\ref{fig:pwl-perturbation-QQ})
the length of the perturbed edges in the boundary of $S$ is bounded as follows:
\begin{align}
\label{eq:pwl-perturb-edges-upper-bound}
  \abs{e_v} - \e \leq \abs{e_v} - \abs{\tld{\e}_v} \leq
  \abs{\tld{e}_v} 
  \leq \abs{e_v} + \abs{\tld{\e}_v} \leq \abs{e_v} + \e
\end{align}
which implies that $\aBs{\bd \Delta_v} = \abs{e_v} + \abs{ \tld{e}_v} + \abs{\tld{\e}_v} \geq 2 \abs{{e}_v} - \abs{\tld{\e}_v} + \abs{\tld{\e}_v} = 2 \abs{e_v}$,
and thus the following bounds hold: 
\begin{align*}
  2 \abs{e_v}  \leq & ~\aBs{\bd \Delta_v} \leq \dsp 2 \abs{e_v} + 2 \e \quad \text{ and }  \\
  2\big( \abs{e_{v w}} + \abs{e_{u v}} \big) - 2\e  \leq & ~~ \aBs{\bd S} \,\leq \dsp
                        2 \big( \abs{e_{v w }} + \abs{e_{u v}} \big) + 2\e.  \\
\end{align*}

\subsubsection{Sequential $\e$-perturbations} \label{sssec:FN-PWL-PERTURB:SEQUENTIAL}

We now want to derive results similar to Theorem \ref{thm:fn-bound:single-perturbation}
for the case where perturbations are applied to a subcurrent of $T_1$ given by a subset of its adjacent nodes.
To this end, we consider a sequence of $\e$-perturbations of the copy current
$T_1 \perturbe \tld{T}_1 \perturbe \ldots \perturbe \tld{T}^{n - 1}_1$ 
induced by sequential perturbations of the interior points 
$v_1 \perturbe \tld{v}_1, \ldots, v_{n - 1} \perturbe \tld{v}_{n - 1}$,
which we denote as $[v_1, \ldots, v_{n - 1}] \perturbe [\tld{v}_1, \ldots, \tld{v}_{n - 1}]$.

The procedure of $\e$-perturbations described above does not guarantee ``out of the box''
additivity of the area components $\mc{A}(S_i)$ of the corresponding flat norm distances 
$\mathbb{F}_\lambda(\tld{T}^{i}_1 - \tld{T}^{i - 1}_1)$.
The problem arises when a $\e$-perturbation $\tld{v}_i$ lands within a region $S_j$ produced by $v_j \perturbe \tld{v}_j$ for some $j < i$, where $\bd S_j = \tld{T}^{j}_{1} - \tld{T}^{j - 1}_{1}$.
It means that $S_i \cap S_j \neq \emptyset$ and $\mc{A}(S_i + S_j) \neq \mc{A}(S_i) + \mc{A}(S_j)$. 
But since $S_i$ and $S_j$ are embedded in $\real^2$ 
the area of a formal sum $S_i + S_j$ is not larger than the area of their union: 
$\mc{A}\big( S_i + S_j \big) \leq \mc{A} \left(S_i \cup S_j \right) \leq \mc{A}(S_i) + \mc{A}(S_j)$.
Therefore, to find an upper bound for $\mathbb{F}_\lambda(\tld{T}^{n-1}_1 - T_0)$
we can consider only perturbation sequences $T_1 \perturbe \ldots \perturbe \tld{T}^{n - 1}_1$ that produce non-overlapping regions $S_i$,
and hence, additive area components $\mc{A}(S_i)$.

One way to achieve this is to force $\tld{v}$ for any choice of $v$ to land on the same ``side'' of $T_1$, and hence of $T_0$,
by restricting the $\e$-ball $\mc{B}_{\e}(v)$
to a \emph{positive cone} $\mc{B}^{\splus}_{\e}(v)$ given by the edges incident to $v$ in $T_1$. 
As an example in Fig.~\ref{fig:pwl-perturbation-AA-d}, the light green circular arc indicates the segment of the $\e$-ball that corresponds to $\mc{B}^{\splus}_{\e}(v)$.
Let $v = v_i \in \mathbf{N} (T_1)$ be an interior vertex of $T_1$, i.e., $1 \leq i \leq n -1$,
and $e_{uv} = e_{i} = (v_{i - 1}, v_{i})$ and $e_{vw} = e_{i + 1} = (v_{i}, v_{i + 1})$
are the incident edges of $v$ in $T_1$.
Then we get that
\begin{align}
\label{eq:pwl-eplus-perturb-cone}       
  \mc{B}^{\splus}_{\e}(v) = \mc{B}^{\splus}_{\e}(v; T_1)
  &= \bigg\{
      x \in \mc{B}_{\e}(v) 
      \,\bigg\vert 
      \det\big[x - v,\, e_{uv} \big]^{\tr} > 0 \text{ and } \det\big[x - v,\, e_{vw} \big]^{\tr} > 0
  \bigg\}
\end{align}
where $\det [ x, y ]^{\tr} = x_1 y_2 - x_2 y_1$ for $x = (x_1, x_2) \in \real^2$  and $y = (y_1, y_2) \in \real^2$.
In the case of boundary vertices $s$ or $t$, 
the positive cone reduces to the half-disk bounded by a line that contains $e_s$ or $e_t$, respectively.
As previously specified in Eq.~\eqref{eq:pwl-e-perturb-region}, the allowed region of $\tld{v}$ is specified by the non-overlapping conditions:
\begin{align}
\label{eq:pwl-eplus-perturb-region}     
  \mc{R}^{\splus}_{\e}(v) = \mc{R}^{\splus}_{\e}(v; T_1) 
  &= 
    \bigg\{ 
            \tld{v} \in \mc{B}^{\splus}_{\e}(v) 
            \,\bigg\vert \,
            \tld{e}_{uv} \cap T_1 = \tld{e}_{vw} \cap T_1 = \emptyset 
    \bigg\}
  = \mc{B}^{\splus}_{\e}(v; T_1) \cap \mc{R}_{\e}(v; T_1).
\end{align}

\begin{figure}[ht!]
   \centering
   \includegraphics[width=0.55\textwidth]{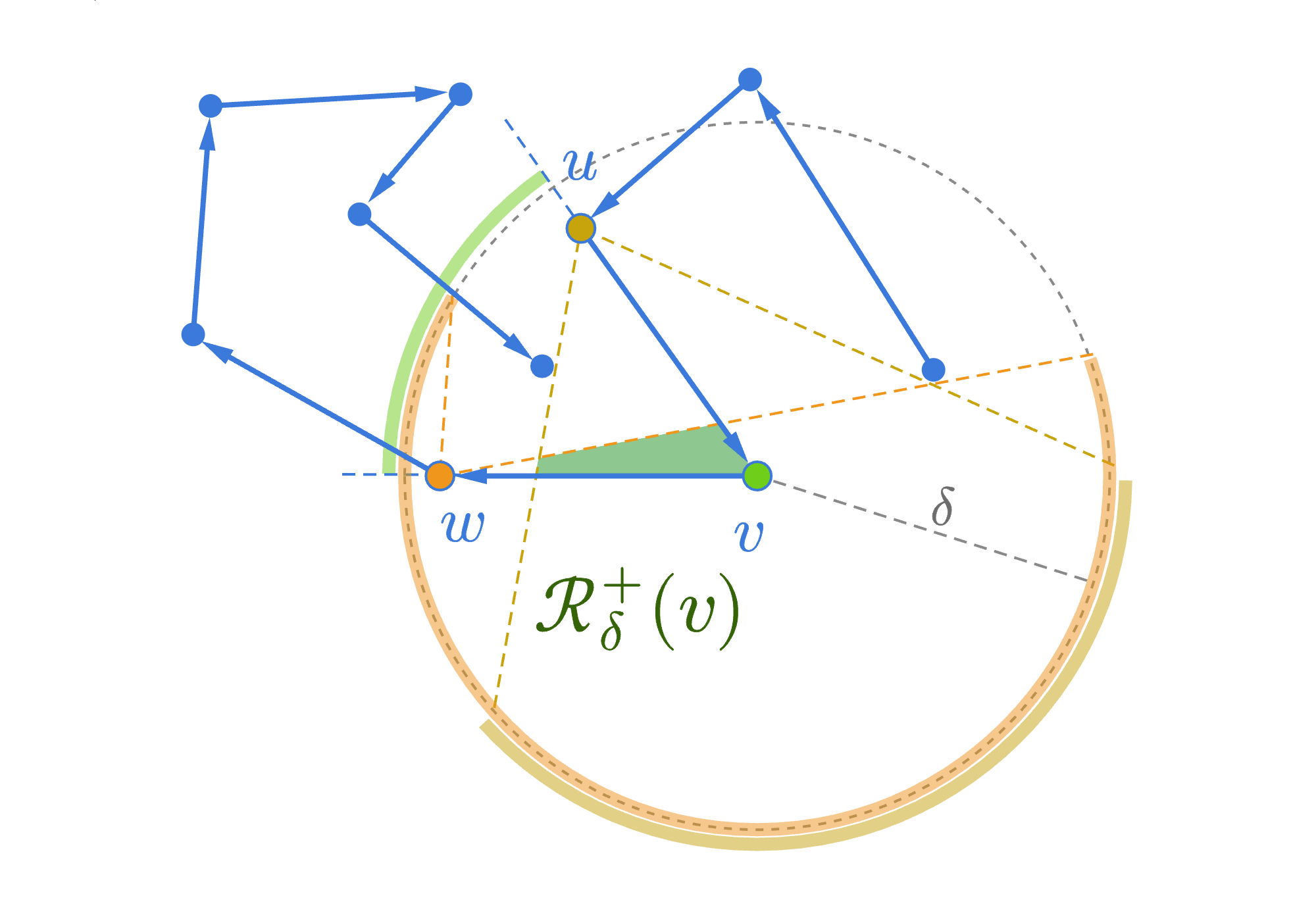}
   \caption{
     The \textcolor{GrassGreen}{light green} arc indicates the positive cone $\mc{B}^{\splus}_{\e}(v)$ (see Eq.~\eqref{eq:pwl-eplus-perturb-cone}).
     The allowed region \textcolor{ForestGreen}{$\mc{R}^{\splus}_{\e}(v)$} (see Eq.~\eqref{eq:pwl-eplus-perturb-region}) of a positive $\e$-perturbation $v \perturbeplus \tld{v}$ is shown in green.
     It ensures that \emph{the sequential perturbations produce non-overlapping regions $S_i$} (Proposition~\ref{thm:eplus-non-overlapping}).
     Compare this figure to the allowed region shown in Fig.~\ref{fig:pwl-perturbation-AA}.
   } 
\label{fig:pwl-perturbation-AA-d}
\end{figure}

We call a perturbation $v_i \perturbe \tld{v}_i$
a \textit{\textbf{positive $\boldsymbol{\e}$-perturbation}} of $v_i \in \tld{T}^{i - 1}_1$
if $\tld{v}_i \in \mc{R}^{\splus}_{\e}(v_i; \tld{T}^{i - 1}_1)$.
We denote it as $v_i \perturbeplus \tld{v}_i$, and it induces $\tld{T}^{i - 1}_1 \perturbeplus \tld{T}^{i}_1$.
It easy to see that a sequence of $k$ positive $\e$-perturbations $T_1 \perturbeplus \ldots \perturbeplus \tld{T}^{k}_1$
produces non-overlapping regions $S_1, \ldots, S_k$.

\begin{proposition}
\label{thm:eplus-non-overlapping}
  Let $T_1 \in \mc{L}_n[s,t; \real^2]$ for $n \geq 2$.
  Consider a sequence  of perturbations 
  $T_1 \perturbeplus \tld{T}^{k_0}_1 \perturbeplus \ldots \perturbeplus \tld{T}^{k}_1$ 
  induced by positive $\e$-perturbations of adjacent nodes $[v_{k_0}, \ldots, v_{k}]$
  for some $0 \leq k_0 < k \leq n$.
  Then the regions $S_{k_0}, \ldots, S_k$ produced by the corresponding perturbations in the sequence do not overlap.
\end{proposition}

\begin{proof}
  
  First, when two adjacent nodes are perturbed, $v_{i - 1} \perturbeplus \tld{v}_{i - 1}$ and $v_{i} \perturbeplus \tld{v}_{i}$,
  the edge $e_{i} = (v_{i - 1}, v_{i})$ is perturbed twice:
  $e_i \to \tld{e}_i \to \tld{\tld{e}}_i$,
  where $\tld{e}_i \in \mc{L}_{1}[\tld{v}_{i-1}, v_i]$
  and $\tld{\tld{e}}_i \in \mc{L}_{1}[\tld{v}_{i-1}, \tld{v}_i]$.
  Therefore, the  boundaries of corresponding regions $S_{i - 1}$ and $S_i$ are:
  \begin{align}
    \bd S_{i - 1} = & \, \tld{T}^{i - 1}_1 - \tld{T}^{i - 2}_1 = \tld{e}_{i - 1} + e_{i} - \tld{\tld{e}}_{i - 1} -\tld{e}_{i} ~\text{ and }
    \\
    \bd S_{i} = & ~~ \tld{T}^{i}_1 - \tld{T}^{i - 1}_1 ~= \tld{e}_{i} + e_{i + 1} - \tld{\tld{e}}_{i} -\tld{e}_{i + 1}. 
  \end{align}
  
  It is worth pointing out that when only interior nodes are perturbed, i.e., $1 \leq k_0$ and $k \leq n$,
  the edges $e_{k_0} = (v_{k_0 - 1}, v_{k_0}) \in S_{k_0}$ and $e_{k + 1} = (v_{k}, v_{k + 1}) \in S_k$ are perturbed only once.
  Thus, the edges incident to the boundaries of $T_1$, namely $e_s$ and $e_t$, 
  are perturbed to $\tld{\tld{e}}_s$ or $\tld{\tld{e}}_t$ if and only if the boundary nodes $s$ and $t$ were perturbed
  together with all the nodes in between. 
  In this case the corresponding area components $S_0$ and $S_n$
  are given by  triangles $\Delta_s = [{s}, \tld{v}_1, \tld{s}]$ 
  and $\Delta_t = [\tld{t}, \tld{v}_{n-1}, t]$, respectively.

  Second, observe that any point $x$ within $S_i$ is in $\mc{B}^{\splus}_{\e}(v_i)$ as well, 
  which means that $\det\big[x - v_i,\, \tld{e}_{i} \big]^{\tr} > 0 $ and  $\det\big[x - v_i,\, e_{i + 1} \big]^{\tr} > 0$.
  This point also belongs to $\mc{B}^{-}_{\e}(\tld{v}_i) = \big\{ x \in \mc{B}_{\e}(\tld{v}_i)  \mid \det\big[x - \tld{v}_i,\, \tld{\tld{e}}_{i} \big]^{\tr} < 0 \text{ and } \det\big[x - \tld{v}_i,\, \tld{e}_{i + 1} \big]^{\tr} < 0 \big\}$,
  since $x$ is enclosed by $\bd S_i$.
  
  Given a sequence $T_1 \perturbeplus \tld{T}^{k_0}_1 \perturbeplus \ldots \perturbeplus \tld{T}^{k}_1$, 
  let us assume that $\tld{v}_k \in S_j$ for some $k_0 \leq j < k$,
  while $\tld{v}_i \in \mc{R}^{\splus}_{\e}(v_i; \tld{T}^{i - 1}_1)$ for all $k_0 \leq i \leq k$. 
  If $k - j > 1$ then $\bd S_k$ and $\bd S_j$ do not share any edges,
  and hence placing $\tld{v}_k$ inside $S_j$ requires an intersection of edges, which contradicts $\tld{v}_k \in \mc{R}^{\splus}_{\e}(v_k; \tld{T}^{k-1}_1)$.
  If $j = k - 1$ then $\tld{v}_k \in S_k \iff \det\big[\tld{v}_k - v_k,\, \tld{e}_{k} \big]^{\tr} > 0 $
  and $\tld{v}_k \in S_{k - 1} \iff \det\big[\tld{v}_k - \tld{v}_{k - 1},\, \tld{e}_{k} \big]^{\tr} < 0 $,
  where $\tld{e}_k = (\tld{v}_{k-1}, v_k)$,
  which is also a contradiction.
  Therefore $S_k$ does not overlap with any of the previously produced regions $S_{k_0}, \ldots, S_{k -1}$.
\end{proof}

\begin{corollary}[Additivity of $\mathbb{F}_{\lambda}$]
\label{thm:corollary:fn-additivity}
Let $T_0 \in \mc{L}_n[s,t; \real^2]$ for $n \geq 2$, and $T_1 = T_0$ is its copy.
  Given a small enough scale $\lambda > 0$ and an appropriate radius of perturbation $\e > 0$,
  consider a sequence of perturbations $T_1 \perturbeplus \tld{T}^{k_0}_1 \perturbeplus \ldots \perturbeplus \tld{T}^{k}_1$ 
  induced by the positive $\e$-perturbations  of adjacent nodes $[v_{k_0}, \ldots, v_k]$ in $T_1$
  for some $0 \leq k_0 < k \leq n$.
  Then the flat norm distance between the adjacent perturbations is additive
  and sums up to $ \mathbb{F}_{\lambda}(\tld{T}^{k}_1 - T_0)$:
  \begin{align*}
    \mathbb{F}_{\lambda}(\tld{T}^{k}_1 - {T}_0) 
    &= \sum_{i = k_0}^{k} \mathbb{F}_{\lambda}(\tld{T}^{i}_1 - \tld{T}^{i - 1}_1)  
  \end{align*}
  where we set $\tld{T}^{k_0 - 1} = T_1$.
\end{corollary}

\begin{proof}
    Let $S = S_{k_0} + \ldots + S_k$ be the 2-current that combines all the regions produced by $T_1 \perturbeplus \tld{T}^{k_0}_1 \perturbeplus \ldots \perturbeplus \tld{T}^{k}_1$. 
    Since $S_i \cap S_j = \emptyset$ for $i \neq j$  and $\bd S_i = \tld{T}^{i}_1 - \tld{T}^{i - 1}_1$, the boundary of $S$ is given by
    \begin{align*}
      \bd S &= \bd S_{k_0} + \bd S_{k_0 + 1} + \ldots + \bd S_{k - 1} + \bd S_{k} 
    \\
            &= (\tld{T}^{k_0}_1 - {T}_1) + (\tld{T}^{k_0 + 1}_1 - \tld{T}^{k_0}_1) + \ldots + (\tld{T}^{k - 1}_1 - \tld{T}^{k - 2}_1) + (\tld{T}^{k}_1 - \tld{T}^{k - 1}_1)
            = \tld{T}^{k}_1 - T_1. 
    \end{align*}
    
    This means that $S$ is spanned by $\tld{T}^{k}_1 - T_1 = \tld{T}^{k}_1 - T_0$, 
    which under the Assumption~\ref{ref:main-assumption} implies that
    \begin{align*}
      \mathbb{F}_{\lambda}(\tld{T}^{k}_1 - T_0) 
       = \lambda  \mc{A}(S)
       = \lambda \sum_{i = k_0}^{k} \mc{A}(S_i)
      &= \sum_{i = k_0}^{k} \mathbb{F}_{\lambda}(\tld{T}^{i}_1 - \tld{T}^{i - 1}_1).
    \end{align*}
\end{proof}

\begin{figure}[ht!]
   \centering
   \includegraphics[width=0.70\textwidth]{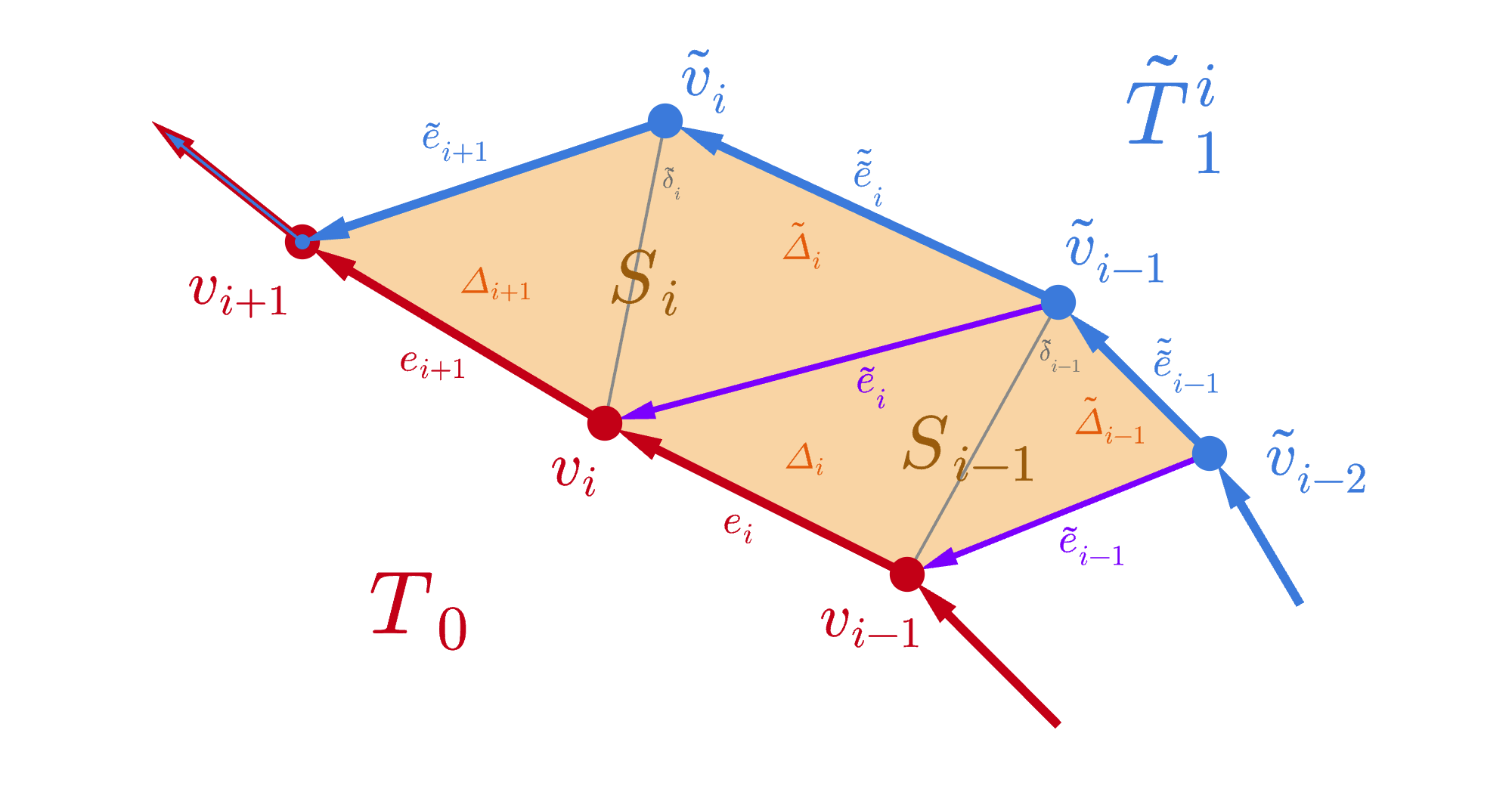}
   \caption{
        The $i$-th positive $\e$-perturbation
        $\tld{T}^{i-1}_1 \perturbeplus \tld{T}^{i}_1$ induced by $v_i \perturbeplus \tld{v}_i$.
        The purple edges show history of the previously perturbed edges \textcolor{DodgerPurple}{$\tld{e}_v$}
        that have been mapped to \textcolor{AzureBlue}{$\tld{\tld{e}}_v$} in $\tld{T}^{i}_1$.
   } 
\label{fig:pwl-perturbation-01-v}
\end{figure}

We now present in Theorem \ref{thm:fn-bound:k-perturbations} an upper bound on the flat norm distance between the input current $T_0$ and its perturbed version resulting from a sequence of perturbations of its \emph{internal} nodes, i.e., nodes in $\mathbf{N} (T_0) \setminus \{s,t\}$.
We subsequently extend the result to include perturbations of all nodes including the boundary nodes in Corollary \ref{thm:corollary:fn-bound-all}.

\begin{theorem}
\label{thm:fn-bound:k-perturbations}
  Let $T_0 \in \mc{L}_n[s,t; \real^2]$ for $n \geq 3$, and $T_1 = T_0$ is its copy.
  Given a small enough scale $\lambda > 0$ and an appropriate radius of perturbation $\e > 0$,
  consider a sequence of perturbations $T_1 \perturbeplus \tld{T}^{k_0}_1 \perturbeplus \ldots \perturbeplus \tld{T}^{k}_1$ 
  induced by the positive $\e$-perturbations  of adjacent \emph{non-boundary} nodes $[v_{k_0}, \ldots, v_k]$ in $T_1$
  for some $1 \leq k_0 <  k \leq n - 1$.
  Then the flat norm distance between $T_0$ and $\tld{T}^{k}_1$ is upper bounded as follows:
  \begin{align*}
    \mathbb{F}_{\lambda}(\tld{T}^{k}_1 - {T}_0) 
    &\leq 
          \frac{\lambda \e}{2} \bigg( 
                    \aBs{T_0[v_{k_0 - 1}, v_{k + 1}]} + \aBs{T_0[v_{k_0}, v_k]} +  (k - k_0)\e 
               \bigg).
  \end{align*}
  
\end{theorem}
\begin{proof}
    As previously was observed in Eq.~\eqref{eq:pwl-perturb-quad-S},
    $S_i$ can be partitioned by $\tld{\e}_{i} = (v_i, \tld{v}_i)$ into a pair of triangles
    $\Delta_{v_{i + 1}} = [v_i, v_{i+1}, \tld{v}_i]$ and $\tld{\Delta}_{v_i} = [\tld{v}_i, \tld{v}_{i - 1}, v_i]$,
    with the respective boundaries $\bd \Delta_{v_{i + 1}} = e_{i + 1} - \tld{e}_{i + 1} - \tld{\e}_{i}$ and 
    $\bd \tld{\Delta}_{v_{i}} = \tld{e}_{i} - \tld{\tld{e}}_{i} + \tld{\e}_{i}$.
    The area of these triangles is bounded as in Eq.~\eqref{eq:pwl-perturb-area-upper-bound}, 
    and the upper bound on the length of the perturbed edges $\tld{e}_{v}$ is given by corresponding triangle inequalities in Eq.~\ref{eq:pwl-perturb-edges-upper-bound}.
    Hence, we get the following upper bound for $\mc{A}(S_i)$:
    \begin{align}
    \label{eq:pwl-perturb-seq-area-bound}     
      \mc{A}(S_i) = \mc{A}(\Delta_{v_{i + 1}}) + \mc{A}(\tld{\Delta}_{v_{i}})
      &\leq \frac{\e}{2} \big( \abs{e_{i + 1}}   + \abs{\tld{e}_{i}} \big)
       \leq \frac{\e}{2} \big( \abs{e_{i + 1}}   + \abs{{e}_{i}} + \e \big)
    \end{align}

    Note that the upper bound on the flat norm distance for the first perturbation in the sequence $T_1 \perturbeplus \tld{T}^{k_0}$
    is given by Theorem~\ref{thm:fn-bound:single-perturbation} instead of Eq.~\eqref{eq:pwl-perturb-seq-area-bound},
    namely $\mathbb{F}_{\lambda}(\tld{T}^{k_0}_1 - T_1) = \lambda \mc{A}(S_{k_0}) \leq \frac{\lambda \e}{2} \big( \abs{e_{k_0}}  + \abs{e_{k_0 + 1}} \big)$.
    As was shown above in the Corollary~\ref{thm:corollary:fn-additivity} 
    under the Assumption~\ref{ref:main-assumption},
    we get additivity of the flat norm distances between consecutive perturbations, and hence the additivity of the upper bounds.
    Then we get the Theorem's claim after doing some algebra:
    \begin{align*}
      \mathbb{F}_{\lambda}(\tld{T}^{k}_1 - T_0) 
      &= \sum_{i = k_0}^{k} \mathbb{F}_{\lambda}(\tld{T}^{i}_1 - \tld{T}^{i - 1}_1)  
       = \lambda \sum_{i = k_0 + 1}^{k} \mc{A}(S_i) + \lambda \mc{A}(S_{k_0})
    \notag
    \\
      &\leq \frac{\lambda \e}{2} \left( 
                  \sum_{i = k_0 + 1}^{k}(\abs{\tld{e}_i} + \abs{e_{i + 1}}) + (\abs{e_{k_0}} + \abs{e_{k_0 + 1}}) 
            \right)
    \notag
    \\
      &\leq \frac{\lambda \e}{2} \left( 
                  \sum_{i = k_0}^{k + 1} \abs{e_{i}} + \sum_{i = k_0 + 1}^{k}\abs{{e}_i} +  (k - k_0)\e 
            \right)
    \notag
    \\
      &= \frac{\lambda \e}{2} \bigg( 
              \aBs{T_0[v_{k_0 - 1}, v_{k + 1}]} + \aBs{T_0[v_{k_0}, v_k]} +  (k - k_0)\e 
         \bigg).
    \notag
    \end{align*}

\end{proof}

We immediately get the flat norm bounds on the positive $\e$-perturbations of all interior nodes,
as well as for perturbations of all nodes as corollaries.

\begin{corollary}[Complete perturbation of interior]
\label{thm:corollary:fn-bound-interior}
  Given the setup of Theorem~\ref{thm:fn-bound:k-perturbations},
  consider a sequence of perturbations $T_1 \perturbeplus \ldots \perturbeplus \tld{T}^{n-1}_1$ 
  induced by the positive $\e$-perturbations of all interior nodes $[v_{1}, \ldots, v_{n - 1}]$.
  Then an upper bound on $\mathbb{F}_{\lambda}\big(\tld{T}^{n-1}_1 - T_0 \big)$ 
  is given by 
  \begin{align*}
    \mathbb{F}_{\lambda}\big(\tld{T}^{n-1}_1 - T_0 \big) 
    &\leq    
          \frac{\lambda \e}{2} \bigg( 
              \aBs{T_0} +  \aBs{ T_0[v_1, v_{n - 1}] } + (n - 2) \e 
          \bigg)
    \\
    &=    
          \lambda \e \bigg( 
              \aBs{T_0 } - \left( \frac{\abs{e_s} + \abs{e_t}}{2} \right) + (n - 2) \frac{\e}{2} 
          \bigg).
  \end{align*}
\end{corollary}

\begin{corollary}[Complete perturbation of $T_1$]
\label{thm:corollary:fn-bound-all}
  Given the setup of Theorem~\ref{thm:fn-bound:k-perturbations} with $n \geq 2$,
  consider a sequence of perturbations $T_1  \perturbeplus \ldots \perturbeplus \tld{T}^{n+1}_1$ 
  induced by the positive $\e$-perturbations of each node of $T_1$. 
  Then an upper bound on $\mathbb{F}_{\lambda}\big(\tld{T}^{n + 1}_1 - T_0 \big)$ 
  is given by 
  \begin{align}
  \label{eq:pwl-perturb-all-fn-bound}     
    \mathbb{F}_{\lambda}(\tld{T}^{n + 1}_1 - T_0)
    &\leq    \frac{\lambda \e}{2} \big( 2 \abs{T_0} + n \e \big)
    =        \lambda \e \left( \abs{T_0} + \frac{n \e}{2} \right).
  \end{align}
\end{corollary}

\begin{proof}
    
    Due to the additivity of the flat norm distance  between  adjacent $\e^{\scalebox{0.4}{+}}$-perturbations, 
    we can perturb the interior nodes $[v_{1}, \ldots, v_{n - 1}]$ first,
    and the boundary nodes after that, since now all edges will be perturbed twice and the order will not affect the upper bound. 
    The perturbation of the interior nodes induces a sequence $T_1 \perturbeplus \ldots \perturbeplus \tld{T}^{n-1}_1$ from Corollary~\ref{thm:corollary:fn-bound-interior},
    while $[s, t] \perturbeplus [\tld{s}, \tld{t}]$ induces 
    $\tld{T}^{n - 1}_1 \perturbeplus \tld{T}^{n}_1 \perturbeplus \tld{T}^{n + 1}_1$.

    The area components produced by $s \perturbeplus \tld{s}$ and $t \perturbeplus \tld{t}$ are spanned by the differences of corresponding perturbations
    are given by triangles $S_0 = \Delta_s = [{s}, \tld{v}_1, \tld{s}] \in \mc{C}_2[\tld{T}^{n}_1 - \tld{T}^{n-1}_1] $ 
    and $S_n = \Delta_t = [\tld{t}, \tld{v}_{n-1}, t]  \in \mc{C}_2[\tld{T}^{n+1}_1 - \tld{T}^{n}_1] $.
    Their boundaries are:
    \begin{align*}
      \bd S_0 = \bd \Delta_{s} = \tld{e}_s - \tld{\tld{e}}_s - \tld{\e}_{s}
      &&\text{ and }&&
      \bd S_n = \bd \Delta_{t} = \tld{e}_t - \tld{\tld{e}}_t + \tld{\e}_{t}
    \end{align*}
    where $\tld{\e}_{v} \in \mc{L}_1(v, \tld{v})$ such that $\abs{\tld{\e}_{v}} \leq \e$, see Fig.~\ref{fig:pwl-perturbation-BB}.
    Then the upper bound on the flat norm distance becomes
    \begin{align}
      \mathbb{F}_{\lambda}(\tld{T}^{n + 1}_1 - T_0) 
      &=     \mathbb{F}_{\lambda}(\tld{T}^{n + 1}_1 - \tld{T}^{n}_1) 
             + \mathbb{F}_{\lambda}(\tld{T}^{n}_1 - \tld{T}^{n - 1}_1) 
             + \mathbb{F}_{\lambda}(\tld{T}^{n-1}_1 - T_0) 
    \notag
    \\
      &\leq    
             \frac{\lambda \e}{2} \left( 
                  \sum_{i = 2}^{n-1} \abs{\tld{e}_{i}} 
                  + \sum_{i = 1}^{n} \abs{{e}_{i}} 
                  + \abs{\tld{e_s}} + \abs{\tld{e_t}}
            \right) 
    \notag
    \\
      &=    
           \frac{\lambda \e}{2} \left( 
                 \sum_{i = i}^{n} \abs{\tld{e}_{i}} 
                 + \sum_{i = 1}^{n} \abs{{e}_{i}}  
            \right)
    \label{eq:fn-bound:pwl-perturb-all-intermedite}     
    \\
      &\leq   \frac{\lambda \e}{2} \big( 
                    2 \abs{T_0} + n \e 
              \big)
      =    \lambda \e \left( 
              \abs{T_0} + \frac{n \e}{2} 
            \right).
    \notag
    \end{align}

\end{proof}

\begin{figure}[ht!]
   \centering
   \includegraphics[width=0.70\textwidth]{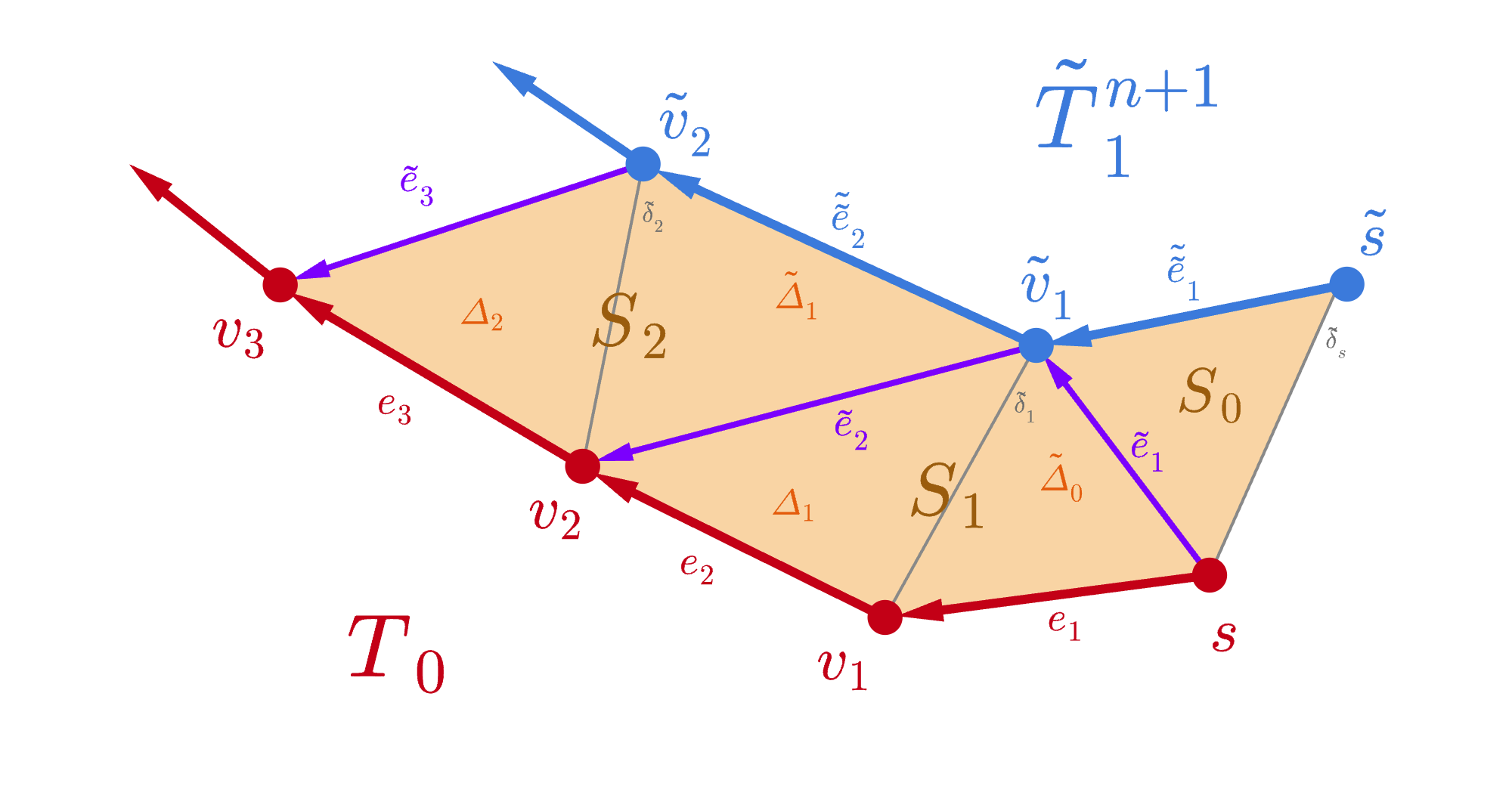}
   \caption{
        The positive $\e$-perturbations $\tld{T}^{n-1}_1 \perturbeplus \tld{T}^{n}_1 \perturbeplus \tld{T}^{n+1}_1$
        induced by perturbations of the boundary vertices $[s, t] \perturb [\tld{s}, \tld{t}]$.
        The area components $S_0$ and $S_n$ that were produced in the result are given by triangles instead of quadrilaterals.
   }
\label{fig:pwl-perturbation-BB}  
\end{figure}

\subsubsection{Normalized flat norm of $\e$-perturbations}
\label{sssec:FN-PWL-PERTURB:NORMALIZED}

Recall from Section~\ref{subsec:normalized-fn} (Eq.~\eqref{eq:flat-norm-normalized-def}) that the normalized flat norm distance is given as
\begin{align*}
    \tld{\mathbb{F}}_{\lambda}\big(\tld{T}_1 - T_0 \big) 
    &=
        \frac{\mathbb{F}_{\lambda}\big(\tld{T}_1 - T_0 \big) }{\aBs{T_0} + \aBs{\tld{T}_1}}.
\end{align*}
Note that the normalized flat norm distance has a natural upper bound given by $\sfrac{\mathbb{F}_{\lambda}\big(\tld{T}^{n-1}_1 - T_0 \big) }{\aBs{T_0}}$
that may be of interest when measurements of one of the input geometries are not available or are not reliable. 
The following corollary of Theorem~\ref{thm:fn-bound:k-perturbations} follows from Eq.~\eqref{eq:pwl-perturb-all-fn-bound}
by dividing both sides by $\Abs{T_0}$.

\begin{corollary}
\label{thm:norm-fn-bound:loose}
  Given the setup of Theorem~\ref{thm:fn-bound:k-perturbations} for $n \geq 2$,
  consider a sequence of perturbations $T_1  \perturbeplus \ldots \perturbeplus \tld{T}^{n+1}_1$ 
  induced by the positive $\e$-perturbations of each node of $T_1$.
  Then an upper bound on the normalized flat norm distance  $\tld{\mathbb{F}}_{\lambda}\big(\tld{T}^{n + 1}_1 - T_0 \big)$ 
  is given by 
  \begin{align*}
    \tld{\mathbb{F}}_{\lambda}\big(\tld{T}^{n+1}_1 - T_0 \big) 
    &\leq \lambda \e \left( 1 + \frac{n}{2 \aBs{T_0}} \e \right)
     = \lambda \e \left( 1 + \frac{\e}{2 \hat{e}} \right)
  \end{align*}
  where $\hat{e} = \frac{1}{n} \abs{T_0}$ is the average edge length of $T_0$.
\end{corollary}


\begin{theorem}
\label{thm:norm-fn-bound:tight}
  Let $T_0 \in \mc{L}_n[s,t; \real^2]$ for $n \geq 2$, and $T_1 = T_0$ be its copy.
  Given a small enough scale $\lambda > 0$ and an appropriate radius of perturbation $\e > 0$,
  consider a sequence of perturbations $T_1  \perturbeplus \ldots \perturbeplus \tld{T}^{n+1}_1$ 
  induced by the positive $\e$-perturbations of each node of $T_1$. 
  Then an upper bound on the normalized flat norm distance 
  $\mathbb{F}_{\lambda}\big(\tld{T}^{n + 1}_1 - T_0 \big)$ 
  is given by 
  \begin{align}
  \label{eq:norm-fn-bound:tight}      
    \tld{\mathbb{F}}_{\lambda}\big(\tld{T}^{n + 1}_1 - T_0 \big) 
    &\leq
        \frac{\lambda \e}{2} \left( 
                \frac{1}{2} + \frac{\hat{e} + \e}{\hat{e} + \hat{e}_{\tld{n}} }
        \right)
       = \lambda \e \left( 
                \frac{1}{4} + \frac{\hat{e} + \e}{ 2(\hat{e} + \hat{e}_{\tld{n}}) }
        \right)
  \end{align}
  where $\hat{e}_{\tld{n}} = \frac{1}{n} \aBs{\tld{T}^{n + 1}_1}$ is the average length of perturbed edges.
\end{theorem}
\begin{proof}

    To derive an upper bound for the normalized flat norm distance  
    observe that by the triangle inequality in $\Delta_i$ and $\tld{\Delta}_{i}$ we get the following bounds:
    \begin{align*}
    &\begin{array}{ccccc}
      \Delta_i: &&\dsp e_i - \e \leq &\tld{e}_i& \leq e_i + \e 
      \\[1em]
      \tld{\Delta}_i: &&\dsp \tld{\tld{e}}_i - \e \leq &\tld{e}_i& \leq \tld{\tld{e}} + \e 
      \\[1em]
      \iimplies &&\dsp \frac{e_i + \tld{\tld{e}}_i}{2} - \e \leq &\tld{e}_i&\dsp \leq \frac{e_i + \tld{\tld{e}}_i}{2} + \e .
    \end{array}
    \end{align*}

    Let $\bbar{T}_1 = \tld{e}_1 + \ldots + \tld{e}_n$.
    Note that $T_0 = e_1 + \ldots + e_n$ and $\tld{T}^{n + 1}_1 = \tld{\tld{e}}_1 + \ldots + \tld{\tld{e}}_n$.
    We then continue the derivation from the intermediate result in Eq.~\eqref{eq:fn-bound:pwl-perturb-all-intermedite}
    obtained during the proof of Corollary~\ref{thm:corollary:fn-bound-all} of Theorem~\ref{thm:fn-bound:k-perturbations}:
    \begin{align*}
      \mathbb{F}_{\lambda}(\tld{T}^{n + 1}_1 - T_0)  
      &\leq \frac{\lambda \e}{2} \left( 
                   \sum_{i = 1}^{n} \abs{\tld{e}_{i}} 
                   + \sum_{i = 1}^{n} \abs{{e}_{i}}  
              \right)
        = \frac{\lambda \e}{2} \bigg( \aBs{\bbar{T}_1 } + \aBs{T_0} \bigg)
      \\
      &\leq \frac{\lambda \e}{2} \left(
          \frac{\aBs{T_0} + \aBs{ \tld{T}^{n + 1}_1 } }{2}
          + \abs{T_0} + n \e
      \right).
    \end{align*}
    
    Recall that $\hat{e}_{\tld{n}} = \sfrac{\aBs{\tld{T}^{n + 1}_1}}{n}$.
    Then,
    \begin{align*}
        \frac{\mathbb{F}_{\lambda}(\tld{T}^{n + 1}_1 - T_0)}{\aBs{T_0} + \aBs{ \tld{T}^{n + 1}_1}}
      &\leq
        \frac{\lambda \e}{2} \left( 
              \frac{1}{2} + \frac{\aBs{T_0}}{\aBs{T_0} + \aBs{ \tld{T}^{n + 1}_1}}
              + \frac{n}{\abs{T_0} + \abs{ \tld{T}^{n + 1}_1}} \e
        \right)
      \\
      &=
        \frac{\lambda \e}{2} \left( 
              \frac{1}{2} + \frac{\hat{e}}{\hat{e} + \hat{e}_{\tld{n}} }
              + \frac{\e}{\hat{e} + \hat{e}_{\tld{n}} } 
        \right).
    \end{align*}

\end{proof}

Combining the generic upper bound in Eq.~\eqref{eq:pwl-fn-upper-bound} that holds for small enough $\lambda > 0$ and an appropriate $\e > 0$,
we can rewrite Eq.~\eqref{eq:norm-fn-bound:tight} as
\begin{align}
\label{eq:norm-fn-bound:tight-truncated}
  \tld{\mathbb{F}}_{\lambda}(\tld{T}^{n + 1}_1 - T_0) 
  &\leq \min \left\{ 
          \lambda \e \left( \frac{1}{4} + \frac{\hat{e} + \e}{ 2(\hat{e} + \hat{e}_{\tld{n}}) }\right),
          \,1
        \right\}.
\end{align}

\begin{corollary}
  In the case of a non-shrinking perturbation sequence $T_1 \perturbeplus \ldots \perturbeplus \tld{T}^{n + 1}_1$ such that $\aBs{\tld{T}^{n + 1}_1} \geq \aBs{T_1} = \aBs{T_0}$,
  the upper bound on the normalized flat norm distance is given as
  \begin{align*}
    \tld{\mathbb{F}}_{\lambda}(\tld{T}^{n + 1}_1 - T_0) 
    &\leq 
      \lambda \e \left( 
              \frac{3}{4} 
              + \frac{\e}{4 \hat{e} } 
      \right).
  \end{align*}

\end{corollary}

\section{Statistical Analysis of Normalized Flat Norm}
\label{sec:stats}

We use the proposed multiscale flat norm to compare a pair of network geometries
from power distribution networks for a region in a county in USA.
The two networks considered are the actual power distribution network for the region and the synthetic network generated using the methodology proposed by Meyur et al.~\cite{rounak2020}.
We provide a brief overview of these networks.

\medskip
\noindent\textbf{Actual network.}~The actual power distribution network was obtained from the power company serving the location.
Due to its proprietary nature, node and edge labels were redacted from the shared data.
Further, the networks were shared as a set of handmade drawings, many of which had not been drawn to a well-defined scale.
We digitized the drawings by overlaying them on OpenStreetMaps~\cite{osm} and georeferencing to particular points of interest~\cite{arcgis}.
Geometries corresponding to the actual network edges are obtained as shape files.

\medskip
\noindent\textbf{Synthetic network.}~The synthetic power distribution network is generated using a framework with the underlying assumption that the network follows the road network infrastructure to a significant extent~\cite{rounak2020}.
To this end, the residences are connected to local pole top transformers located along the road network to construct the low voltage (LV) secondary distribution network.
The local transformers are then connected to the power substation following the road network leading to the medium voltage (MV) primary distribution network.
That is, the primary network edges are chosen from the underlying road infrastructure network such that the structural and power engineering constraints are satisfied. 

\smallskip
In this section we study the empirical distribution of the normalized flat norm $\widetilde{\mathbb{F}}_{\lambda}$ for different local regions and argue that it indeed captures the similarity between input geometries.
We use Algorithm~(\ref{alg:sample}) to sample random square shaped regions of size $2\epsilon\times 2\epsilon$ steradians from a given geographic location.
\begin{algorithm}[tbhp]
\caption{Sample square regions from location}
\label{alg:sample}
\textbf{Input}: Geometries $\mathscr{E}_1,\mathscr{E}_2$, number of regions $N$\\
\textbf{Parameter}: Size of region $\epsilon$
\begin{algorithmic}[1]
\STATE Find bounding rectangle for the pair of geometries: $\mathscr{E}_{\textrm{bound}}=\mathsf{rect}\left(\mathscr{E}_1,\mathscr{E}_2\right)$.
\STATE Initialize set of regions: $\mathscr{R} \leftarrow \{\}$.
\WHILE{$\left|\mathscr{R}\right| \leq N$}
\STATE Sample a point $\left(x,y\right)$ uniformly from region bounded by $\mathscr{E}_{\textrm{bound}}$.
\STATE Define the square region $r\left(x,y\right)$ formed by the corner points $\left\{ \left(x - \epsilon, y - \epsilon\right), \left(x + \epsilon, y + \epsilon\right) \right\}$.
\IF{$r\left(x,y\right) \cap \mathscr{E}_1 \cap \mathscr{E}_2 \neq \emptyset$}
\STATE Add region $r\left(x,y\right)$ to the set of sampled regions: $\mathscr{R} \leftarrow \mathscr{R} \cup \left\{r\left(x,y\right)\right\}$.
\ENDIF
\ENDWHILE
\end{algorithmic}
\textbf{Output}: Set of sampled regions: $\mathscr{R}$.
\end{algorithm}
We perform our empirical studies for two urban locations of a county in USA. These locations have been identified as `Location A' and `Location B' for the remainder of this paper. We consider local regions of sizes characterized by $\epsilon\in\{0.0005, 0.001, 0.0015, 0.002\}$. For each location, we randomly sample $N=50$ local regions for each value of $\epsilon$ using Algorithm~(\ref{alg:sample}) and hence we consider $50\times4=200$ regions. For every sampled region, we use Algorithm~(\ref{alg:distance}) to compute the multiscale flat norm between the network geometries contained within the region with scale parameter $\lambda\in\{10^3,25\times10^3,50\times10^3,75\times10^3,10^5\}$.
The choices of $\epsilon$ and $\lambda$ values are made to cover a sizeable range of all possible values and obtain enough data for robust statistical analysis while still keeping the overall computational load within reasonable limits.
Thereafter, we normalize the computed flat norm using Eq.~(\ref{eq:flat-norm-normalized-def}). Additionally, we compute the global normalized flat norm for the entire location and indicate it by $\widetilde{\mathbb{F}}_{\lambda}^{G}$. The corresponding square box bounding the entire location is characterized by $\epsilon_{G}$. We also denote the total length of networks in each location scaled by the size of the location by the ratio $|T_G|/\epsilon_{G}$. The detailed statistical results for the experiments are included in the Appendix.

\subsection{Empirical distribution of \texorpdfstring{$\widetilde{\mathbb{F}}_{\lambda}$}{normalized flat norm}}

First, we show the histogram of normalized flat norms for Location A and Location B with the five different values for the scale parameter $\lambda$, Fig.~\ref{fig:flatnorm-hists-lambdas}.
Each histogram shows the empirical distribution of normalized flat norm values $\widetilde{\mathbb{F}}_{\lambda}$ for $200$ uniformly sampled local regions ($50$ regions for each $\epsilon$).
We also record the global normalized flat norm between the network geometries of the location $\widetilde{\mathbb{F}}_{\lambda}^{G}$ and denote it by the solid blue line in each histogram.
We show the mean normalized flat norm $\widehat{\mathbb{F}}_{\lambda}$ using the solid green line, with dashed green lines indicating the standard deviation of the distribution.
As we can see from the above histograms, the distribution is skewed toward the right for {high values of the scale parameter $\lambda$}.
This follows from our previous discussion of the dependence of the flat norm on the scale parameter:
for a large $\lambda$, the area patches are weighed higher in the objective function of the flat norm LP (Eq.~(\ref{eq:opt-flatnorm})).
Therefore, the contribution of lengths of the input currents $T_1$ and $T_2$ towards the flat norm distance becomes more dominant 
at \emph{the high values of the scale parameter $\lambda$},
so that the flat norm  $\mathbb{F}_\lambda\left(T_1-T_2\right)$ is slowly approaching the total network length $|T_1| + |T_2|$.
Hence, the normalized flat norm is approaching $1$.
For the remainder of the paper, we will continue our discussion with scale parameter $\lambda=1000$ since the empirical distributions of normalized flat norm corresponding to $\lambda=1000$ indicate almost Gaussian distribution.

\begin{figure}[ht!]
    \centering
    \includegraphics[width=\textwidth]{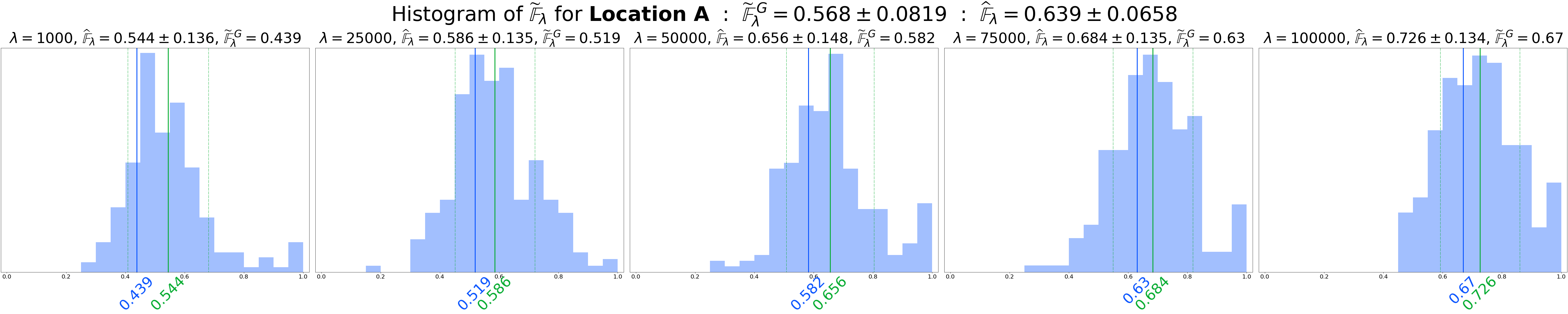}\\
    \smallskip
    \includegraphics[width=\textwidth]{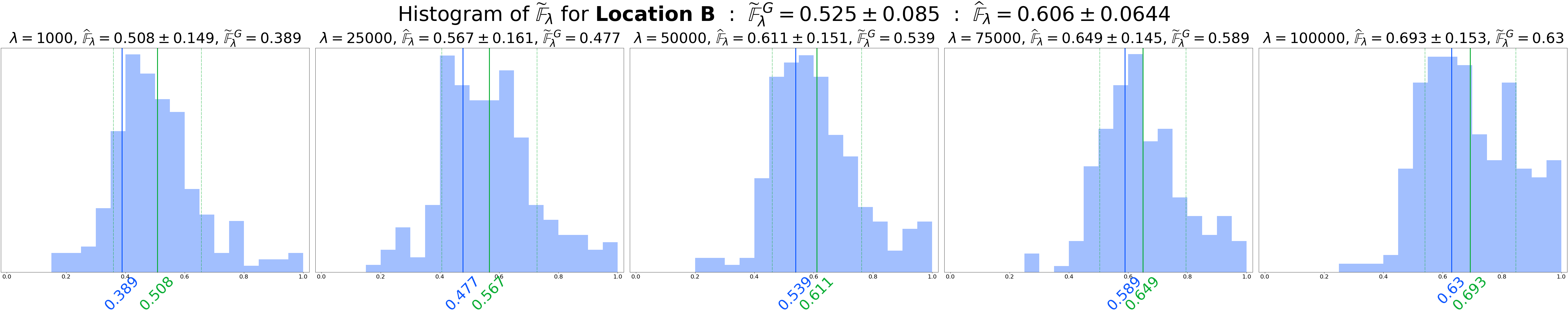}
    \caption{Distribution of normalized flat norm computed using five different values of $\lambda$ for $200$ uniformly sampled local regions in Location A (top) and Location B (bottom). The blue line in each histogram denotes the global normalized flat norm $\widetilde{\mathbb{F}}_{\lambda}^{G}$ computed for the location with the corresponding scale $\lambda$. The solid green line denotes the mean normalized flat norm $\widehat{\mathbb{F}}_{\lambda}$ for the uniformly sampled local regions computed with scale $\lambda$. The dashed green lines show the spread of the distribution.}
    \label{fig:flatnorm-hists-lambdas}
\end{figure}

Next, we consider the empirical distribution of normalized flat norm computed with scale parameter $\lambda=1000$ for uniformly sampled local regions in Location A and Location B, Fig.~\ref{fig:flatnorm-hists-epsilons}.
We show separate histograms for four different-sized local regions (different values of $\epsilon$).
Note that for \emph{small-sized local regions} (low $\epsilon$), the distribution is skewed toward the right. 
This is because when we consider small regions, we often capture very isolated network geometries and the flat norm computation is close to the total network length $\mathbb{F}_\lambda\left(T_1-T_2\right) \rightarrow |T_1|+|T_2|$, which again leads the normalized flat norm to be close to $1$.
Such occurrences are avoided in larger local regions, and therefore we do not observe skewed distributions.

\begin{figure*}[ht!]
    \centering
    \includegraphics[width=\textwidth]{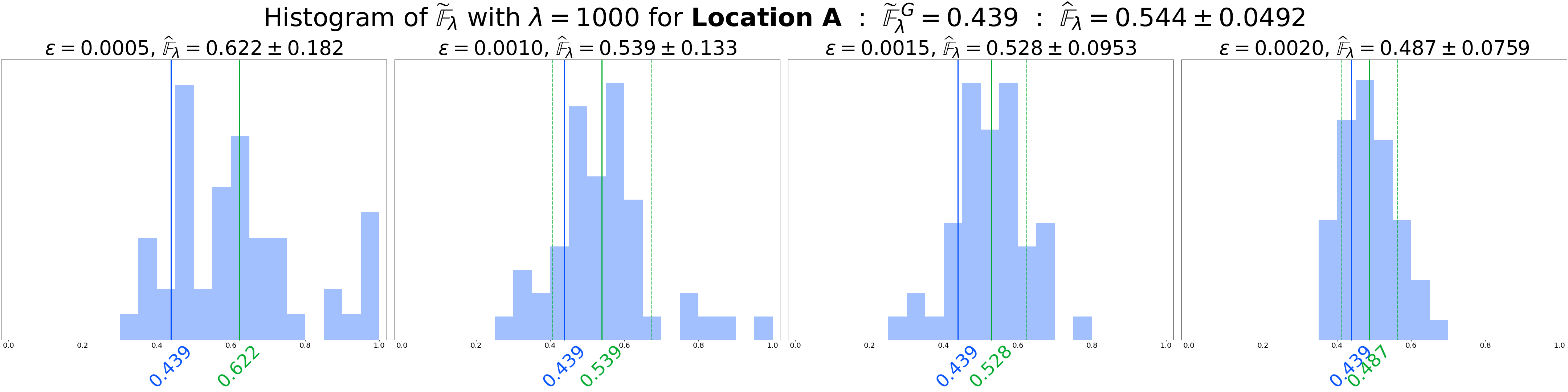}\\
    \smallskip
    \includegraphics[width=\textwidth]{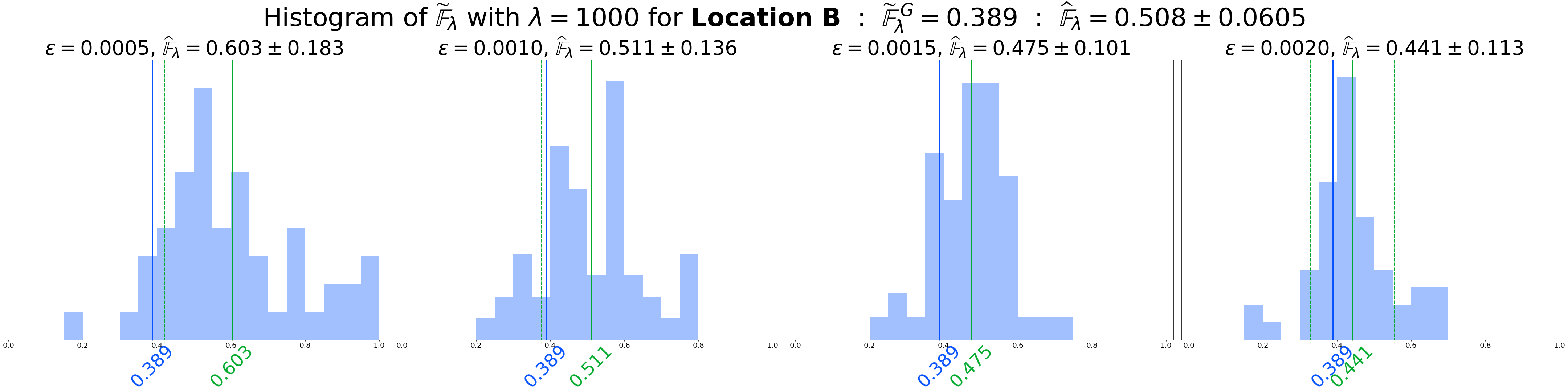}
    \caption{Distribution of normalized flat norm computed using $\lambda=1000$ for $50$ uniformly sampled local regions with four different sizes $\epsilon$ in Location A (top) and Location B (bottom). The blue line in each histogram denotes the global normalized flat norm $\widetilde{\mathbb{F}}_{\lambda}^{G}$. The solid green line denotes the mean normalized flat norm $\widehat{\mathbb{F}}_{\lambda}$ for the uniformly sampled local regions. The dashed green lines show the spread of the distribution.}
    \label{fig:flatnorm-hists-epsilons}
\end{figure*}

\subsection{Distribution of \texorpdfstring{$\widetilde{\mathbb{F}}_{\lambda}$}{normalized flat norm distance} across local regions}
\begin{figure*}
    \centering
    \includegraphics[width=0.99\textwidth]{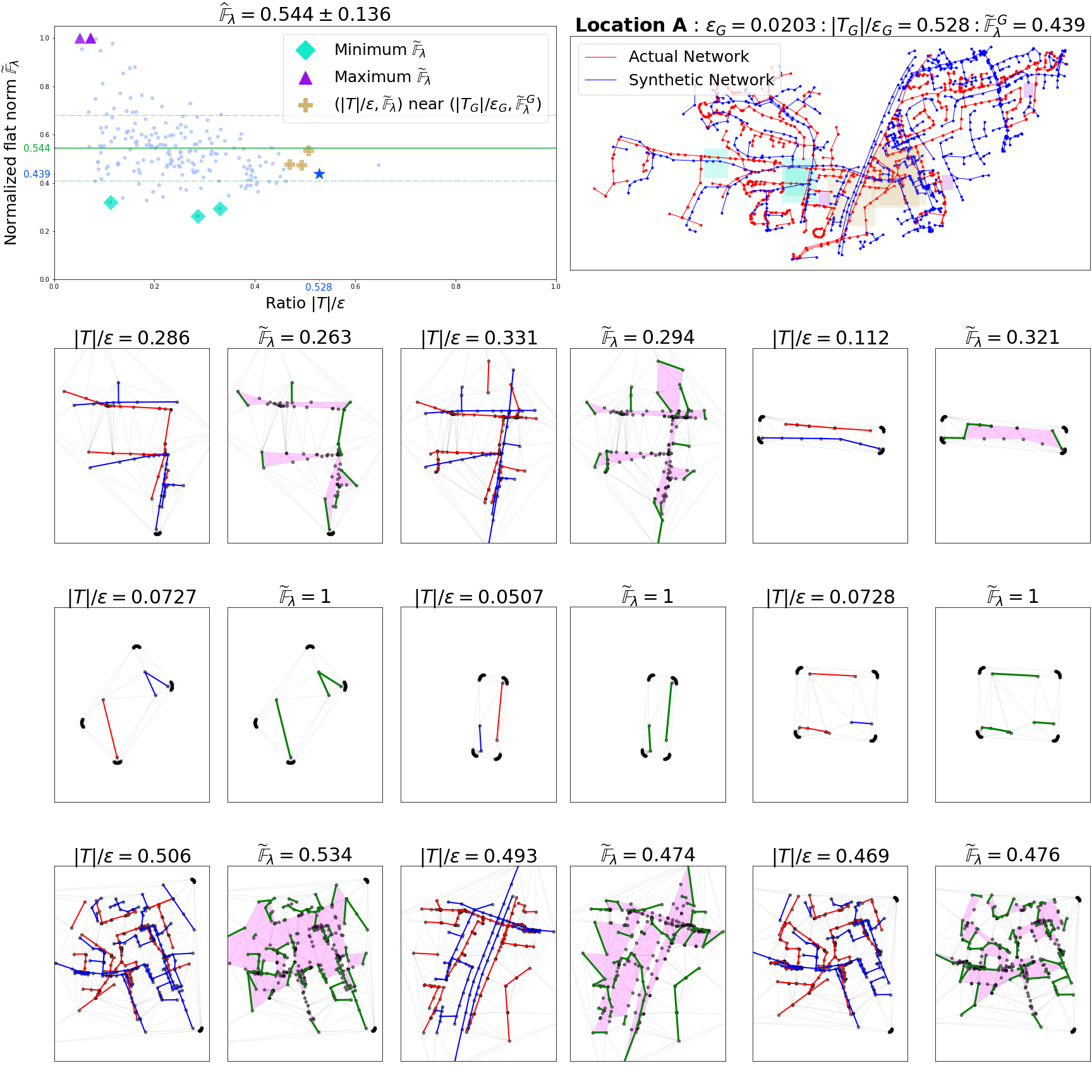}
    \caption{Plots showing normalized flat norm computed for entire Location A and few local regions within it. The scatter plot (top left plot) shows the empirical distribution of $\left(|T|/\epsilon,\widetilde{\mathbb{F}}_{\lambda}\right)$ values with the global normalized flat norm $\left(|T_G|/\epsilon_G,\widetilde{\mathbb{F}}_{\lambda}^{G}\right)$ for the region (blue star). Nine local regions (three with small $\widetilde{\mathbb{F}}_{\lambda}$, three with large $\widetilde{\mathbb{F}}_{\lambda}$ and three with $\left(|T|/\epsilon,\widetilde{\mathbb{F}}_{\lambda}\right)$ values close to the global value $\left(|T_G|/\epsilon_G,\widetilde{\mathbb{F}}_{\lambda}^{G}\right)$) are additionally highlighted. The local regions are highlighted along with the pair of network geometries (top right plot). The normalized flat norm computation (with scale $\lambda=1000$) for the local regions are shown in bottom plots.}
    \label{fig:regionA}
\end{figure*}

  The scatter plot in the top left of Fig.~\ref{fig:regionA} shows the empirical distribution of $\left(|T|/\epsilon,\widetilde{\mathbb{F}}_{\lambda}\right)$ values.
  The scatter plot highlights as a blue star the global value $\left(|T_G|/\epsilon_G,\widetilde{\mathbb{F}}_{\lambda}^{G}\right)$ of Location A, which indicates the normalized flat norm computed for the entire location.
  The global normalized flat norm (with a scale parameter $\lambda=1000$) for Location A is $\widetilde{\mathbb{F}}_{\lambda}^{G}=0.439$ and the ratio $|T_G|/\epsilon_G=0.528$.
  Further, nine additional points are highlighted in the scatter plot denoting nine local regions within Location A.
  The solid green line denotes the mean of the normalized flat norm values and the dashed green lines indicate the spread of the values around the mean.

\begin{figure*}
    \centering
    \includegraphics[width=0.99\textwidth]{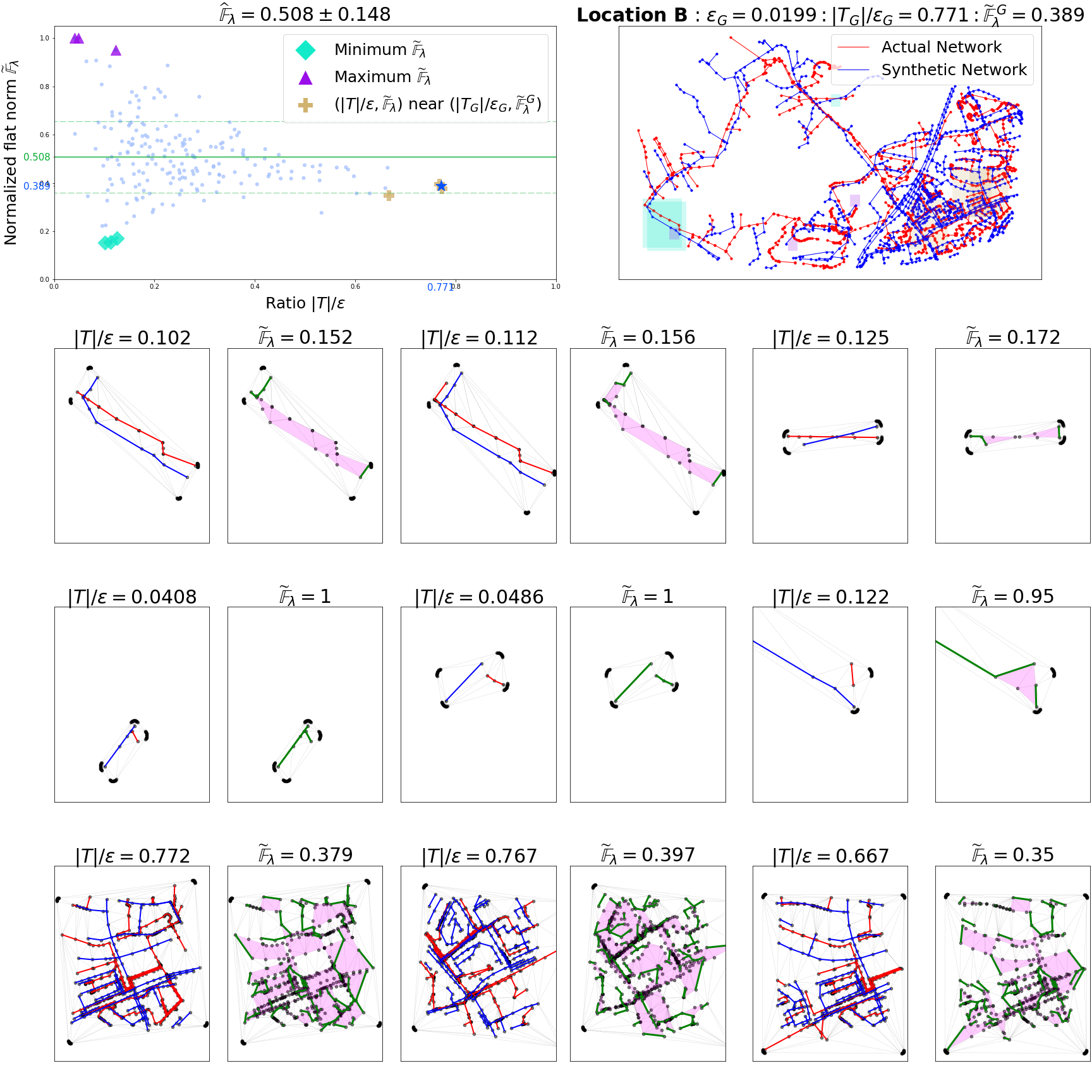}
    \caption{Plots showing normalized flat norm computed for entire Location B and few local regions within it. The scatter plot (top left plot) shows the empirical distribution of $\left(|T|/\epsilon,\widetilde{\mathbb{F}}_{\lambda}\right)$ values with the global normalized flat norm $\left(|T_G|/\epsilon_G,\widetilde{\mathbb{F}}_{\lambda}^{G}\right)$ for the region (blue star). Nine local regions (three with small $\widetilde{\mathbb{F}}_{\lambda}$, three with large $\widetilde{\mathbb{F}}_{\lambda}$ and three with $\left(|T|/\epsilon,\widetilde{\mathbb{F}}_{\lambda}\right)$ values close to the global value $\left(|T_G|/\epsilon_G,\widetilde{\mathbb{F}}_{\lambda}^{G}\right)$) are additionally highlighted. The local regions are highlighted along with the pair of network geometries (top right plot). The normalized flat norm computation (with scale $\lambda=1000$) for the local regions are shown in bottom plots.}
    \label{fig:regionB}
\end{figure*}

The nine local regions are selected such that three of them have the minimum $\widetilde{\mathbb{F}}_{\lambda}$ in the location (highlighted by cyan colored diamonds), three of them have the maximum $\widetilde{\mathbb{F}}_{\lambda}$ in the location (highlighted by purple triangles), and the remaining three local regions have the $\left(|T|/\epsilon,\widetilde{\mathbb{F}}_{\lambda}\right)$ values close to the global value $\left(|T_G|/\epsilon_G,\widetilde{\mathbb{F}}_{\lambda}^{G}\right)$ for the location (highlighted by tan plus symbols).
The network geometries within each region and the flat norm computation with scale $\lambda=1000$ are shown in the bottom plots.
The computed flat norm $\widetilde{\mathbb{F}}_{\lambda}$ and ratio $|T|/\epsilon$ values are shown above each plot.
The local regions are also highlighted (cyan, purple, and tan colored boxes, respectively) in the top right plot where the actual and synthetic network geometries within the entire location are overlaid.
Fig.~\ref{fig:regionB} shows similar local regions from Location B.

From a mere visual inspection of both Figs.~\ref{fig:regionA} and~\ref{fig:regionB}, we notice that the network geometries in each local region shown in the first row of the bottom plots resemble and almost overlap each other.
The computed normalized flat norm $\widetilde{\mathbb{F}}_{\lambda}$ values for these local regions agree with this observation.
Similarly, the large value of the normalized flat norm justifies the observation that network geometries for local regions depicted in the second row of the bottom plots do not resemble each other.
These observations validate our choice of using normalized flat norm as a suitable measure to compare network geometries for local regions.

\section{Conclusions}

We have proposed a fairly general metric to compare a pair of network geometries embedded on the same plane.
Unlike standard approaches that map the geometries to points in a possibly simpler space and then measuring distance between those points~\cite{KeBaCaLe2009}, or comparing ``signatures'' for the geometries,  our metric works directly in the input space and hence allows us to capture all details in the input.
The metric uses the multiscale flat norm from geometric measure theory, and can be used in more general settings as long as we can triangulate the region containing the two geometries.
It appears challenging to derive \emph{standard} stability results for this distance measure that imply only small changes in the flat norm metric when the inputs change by a small amount---there is no obvious alternative metric to measure the \emph{small change in the input} that is also stable on its own in the same manner.
For instance, our theoretical example (in Fig.~\ref{fig:currents:example-currents-and-neighborhoods}) shows that the commonly used Hausdorff metric cannot be used for this purpose.
Instead, we have derived upper bounds on the flat norm distance between a piecewise linear 1-current and its perturbed version as a function of the radius of perturbation under certain assumptions provided the perturbations are performed carefully (see Section \ref{subsec:FN-BOUND}).
On the other hand, we do get natural stability results for our distance following the properties of the flat norm~\cite{Federer1969,Morgan2016}---small changes in the input geometries lead to only small changes in the flat norm distance between them.

A highly desirable result would be to derive stability of the flat norm distance in terms of other recently studied distances for geometric graphs such as the geometric graph distance or the graph edit distance \cite{Ma2023,MaWe2024}.
  On a related note, applying the standard framework for stability of persistence modules \cite{ChdeSGlOu2016} to the flat norm distance may prove challenging, since the flat norm decompositions may not form a filtration as the scale parameter $\lambda$ varies over its range.
  In fact, one may have to consider the setting of multiparameter persistence \cite{BoLe2023} since the scale parameter $\lambda$ varies independent of the default distance parameter.
  But the non-monotone dependence on the scale parameter $\lambda$ would again pose a challenge to apply results on the stability of 2-parameter persistence \cite{BlLe2024}.

We use the proposed metric to compare a pair of power distribution networks: (i) actual power distribution networks of two locations in a county of USA obtained from a power company and (ii) synthetically generated digital duplicate of the network created for the same geographic location.
The proposed comparison metric is able to perform global as well as local comparison of network geometries for the two locations.
We discuss the effect of different parameters used in the metric on the comparison.
Further, we validate the suitability of using the flat norm metric for such comparisons using computation as well as theoretical examples.


\begin{thebibliography}{10}

\bibitem{roads_2015}
Mahmuda Ahmed, Brittany~Terese Fasy, Kyle~S. Hickmann, and Carola Wenk.
\newblock A path-based distance for street map comparison.
\newblock {\em ACM Trans. Spatial Algorithms Syst.}, 1(1), 2015.

\bibitem{roads2014}
Mahmuda Ahmed, Brittany~Terese Fasy, and Carola Wenk.
\newblock Local persistent homology based distance between maps.
\newblock In {\em Proceedings of the 22nd ACM SIGSPATIAL International
  Conference on Advances in Geographic Information Systems}, SIGSPATIAL '14,
  page 43–52, New York, NY, USA, 2014. Association for Computing Machinery.

\bibitem{BaDiBiChSuWa2019}
Yunsheng Bai, Hao Ding, Song Bian, Ting Chen, Yizhou Sun, and Wei Wang.
\newblock {SimGNN: A Neural Network Approach to Fast Graph Similarity
  Computation}.
\newblock In {\em Proceedings of the Twelfth ACM International Conference on
  Web Search and Data Mining}, WSDM '19, page 384–392, New York, NY, USA,
  2019. Association for Computing Machinery.

\bibitem{highres_net}
Antoine Bidel, Tom Schelo, and Thomas Hamacher.
\newblock Synthetic distribution grid generation based on high resolution
  spatial data.
\newblock In {\em 2021 IEEE International Conference on Environment and
  Electrical Engineering and 2021 IEEE Industrial and Commercial Power Systems
  Europe (EEEIC / ICPS Europe)}, pages 1--6, Bari, Italy, 2021. IEEE.

\bibitem{BlLe2024}
Andrew~J. Blumberg and Michael Lesnick.
\newblock {Stability of 2-Parameter Persistent Homology}.
\newblock {\em Foundations of Computational Mathematics}, 24:385--427, 2024.


\bibitem{BoLe2023}
Magnus Botnan and Michael Lesnick.
\newblock {\em An introduction to multiparameter persistence}, pages 77--150.
\newblock EMS Press, 2023.


\bibitem{roads_haus}
Maria~Antonia Brovelli, Marco Minghini, Monia Molinari, and Peter Mooney.
\newblock Towards an automated comparison of openstreetmap with authoritative
  road datasets.
\newblock {\em Transactions in GIS}, 21(2):191--206, 2017.

\bibitem{streets2006}
Alessio Cardillo, Salvatore Scellato, Vito Latora, and Sergio Porta.
\newblock Structural properties of planar graphs of urban street patterns.
\newblock {\em Phys. Rev. E}, 73:066107, 2006.

\bibitem{ChdeSGlOu2016}
Fr{\'e}d{\'e}ric Chazal, Vin de~Silva, Marc Glisse, and Steve Oudot.
\newblock {\em {The Structure and Stability of Persistence Modules}}.
\newblock SpringerBriefs in Mathematics. Springer Cham, 1 edition, 2016.


\bibitem{eppstein_1995}
David Eppstein.
\newblock Subgraph isomorphism in planar graphs and related problems.
\newblock In {\em Proceedings of the Sixth Annual ACM-SIAM Symposium on
  Discrete Algorithms}, SODA '95, page 632–640, USA, 1995. Society for
  Industrial and Applied Mathematics.

\bibitem{arcgis}
ESRI.
\newblock
  \href{https://desktop.arcgis.com/en/arcmap/latest/manage-data/raster-and-images/georeferencing-a-raster-to-a-vector.html}{Georeferencing
  a raster to a vector}, 2022.
\newblock
  \href{https://desktop.arcgis.com/en/arcmap/latest/manage-data/raster-and-images/georeferencing-a-raster-to-a-vector.html}{https://desktop.arcgis.com/en/arcmap/
  latest/manage-data/raster-and-images/georeferencing-a-raster-to-a-vector.html};
  last accessed 26 Sept 2022.

\bibitem{Federer1969}
Herbert Federer.
\newblock {\em Geometric Measure Theory}.
\newblock Die Grundlehren der mathematischen Wissenschaften, Band 153.
  Springer-Verlag, New York, 1969.

\bibitem{IbKrVi2013}
Sharif Ibrahim, Bala Krishnamoorthy, and Kevin~R. Vixie.
\newblock Simplicial flat norm with scale.
\newblock {\em Journal of Computational Geometry}, 4(1):133--159, 2013.
\newblock \href{http://arxiv.org/abs/1105.5104}{arXiv:1105.5104}.

\bibitem{KeBaCaLe2009}
David~G. Kendall, Dennis~M. Barden, {T.~Keith} Carne, and Huiling Le.
\newblock {\em Shape and Shape Theory}.
\newblock Wiley Series in Probability and Statistics. Wiley, Hoboken, NJ, USA,
  2009.

\bibitem{validate2020}
Venkat Krishnan, Bruce Bugbee, Tarek Elgindy, Carlos Mateo, Pablo Duenas,
  Fernando Postigo, Jean-S\'{e}ebastien Lacroix, Tom\'{a}s G\'{o}mez~San Roman,
  and Bryan Palmintier.
\newblock Validation of synthetic u.s. electric power distribution system data
  sets.
\newblock {\em IEEE Transactions on Smart Grid}, 11(5):4477--4489, 2020.

\bibitem{overbye_2020}
Hanyue Li, Jessica~L. Wert, Adam~Barlow Birchfield, Thomas~J. Overbye, Tomas
  Gomez~San Roman, Carlos~Mateo Domingo, Fernando Emilio~Postigo Marcos,
  Pablo~Duenas Martinez, Tarek Elgindy, and Bryan Palmintier.
\newblock Building highly detailed synthetic electric grid data sets for
  combined transmission and distribution systems.
\newblock {\em IEEE Open Access Journal of Power and Energy}, 7:478--488, 2020.

\bibitem{feeder_gan}
Ming Liang, Yao Meng, Jiyu Wang, David~L. Lubkeman, and Ning Lu.
\newblock Feedergan: Synthetic feeder generation via deep graph adversarial
  nets.
\newblock {\em IEEE Transactions on Smart Grid}, 12(2):1163--1173, 2021.

\bibitem{Ma2023}
Sushovan Majhi.
\newblock {Graph mover’s distance: An efficiently computable distance measure
  for geometric graphs}.
\newblock In {\em Proceedings of the 35th Annual Canadian Conference on
  Computational Geometry}, CCCG 2023, pages 207--212, 2023.

\bibitem{MaWe2024}
Sushovan Majhi and Carola Wenk.
\newblock Distance measures for geometric graphs.
\newblock {\em Computational Geometry}, 118:102056, 2024.

\bibitem{nrel_net}
Carlos Mateo, Fernando Postigo, Fernando de~Cuadra, Tom\'{a}s G\'{o}mez~San
  Roman, Tarek Elgindy, Pablo Due\~{n}as, Bri-Mathias Hodge, Venkat Krishnan,
  and Bryan Palmintier.
\newblock Building large-scale u.s. synthetic electric distribution system
  models.
\newblock {\em IEEE Transactions on Smart Grid}, 11(6):5301--5313, 2020.

\bibitem{rounak2020}
Rounak Meyur, Madhav Marathe, Anil Vullikanti, Henning Mortveit, Samarth
  Swarup, Virgilio Centeno, and Arun Phadke.
\newblock Creating realistic power distribution networks using interdependent
  road infrastructure.
\newblock In {\em 2020 IEEE International Conference on Big Data (Big Data)},
  pages 1226--1235, Atlanta, GA, USA, 2020. IEEE.

\bibitem{rounak_pnas}
Rounak Meyur, Anil Vullikanti, Samarth Swarup, Henning Mortveit, Virgilio
  Centeno, Arun Phadke, Vince Poor, and Madhav Marathe.
\newblock Ensembles of realistic power distribution networks.
\newblock {\em Proceedings of the National Academy of Sciences},
  119(26):e2123355119, 2022.

\bibitem{efficiency2020}
Ignacio Morer, Alessio Cardillo, Albert D\'{\i}az-Guilera, Luce Prignano, and
  Sergi Lozano.
\newblock Comparing spatial networks: A one-size-fits-all efficiency-driven
  approach.
\newblock {\em Phys. Rev. E}, 101:042301, 2020.

\bibitem{Morgan2016}
Frank Morgan.
\newblock {\em Geometric Measure Theory: A Beginner's Guide}.
\newblock Academic Press, Cambridge, MA, USA, 5th edition, 2016.

\bibitem{MoVi2007}
Simon~P. Morgan and Kevin~R. Vixie.
\newblock {$L^1$TV} computes the flat norm for boundaries.
\newblock {\em Abstract and Applied Analysis}, 2007:Article ID 45153,14 pages,
  2007.
\newblock \href{http://arxiv.org/abs/math/0612287}{arXiv:0612287}.

\bibitem{graphsim-2020}
Seongmin Ok.
\newblock A graph similarity for deep learning.
\newblock In H.~Larochelle, M.~Ranzato, R.~Hadsell, M.F. Balcan, and H.~Lin,
  editors, {\em Advances in Neural Information Processing Systems}, volume~33,
  pages 1--12, Vancouver, BC, Canada, 2020. Curran Associates, Inc.

\bibitem{osm}
{Open Street Map Foundation}.
\newblock \href{https://www.openstreetmap.org}{Open Street Maps}, 2022.
\newblock \href{https://www.openstreetmap.org}{www.openstreetmap.org}; Last
  accessed 20 Aug 2022.

\bibitem{PaGaMiHa2018}
Benjamin Paa{\ss}en, Claudio Gallicchio, Alessio Micheli, and Barbara Hammer.
\newblock Tree edit distance learning via adaptive symbol embeddings.
\newblock In Jennifer Dy and Andreas Krause, editors, {\em Proceedings of the
  35th International Conference on Machine Learning}, volume~80 of {\em
  Proceedings of Machine Learning Research}, pages 3976--3985. PMLR, 10--15 Jul
  2018.

\bibitem{riba2020}
Pau Riba, Andreas Fischer, Josep Llad\'{o}s, and Alicia Forn\'{e}s.
\newblock Learning graph edit distance by graph neural networks.
\newblock {\em Pattern Recognition}, 120:108132, 2021.

\bibitem{samiul2021}
Kamol~Chandra Roy, Samiul Hasan, Aron Culotta, and Naveen Eluru.
\newblock Predicting traffic demand during hurricane evacuation using real-time
  data from transportation systems and social media.
\newblock {\em Transportation Research Part C: Emerging Technologies},
  131:103339, 2021.

\bibitem{anna_naps}
Shammya~Shananda Saha, Eran Schweitzer, Anna Scaglione, and Nathan~G. Johnson.
\newblock A framework for generating synthetic distribution feeders using
  openstreetmap.
\newblock In {\em 2019 North American Power Symposium (NAPS)}, pages 1--6,
  Wichita, KS, USA, 2019. IEEE.

\bibitem{schweitzer}
Eran Schweitzer, Anna Scaglione, Antonello Monti, and Giuliano~Andrea Pagani.
\newblock Automated generation algorithm for synthetic medium voltage radial
  distribution systems.
\newblock {\em IEEE Journal on Emerging and Selected Topics in Circuits and
  Systems}, 7(2):271--284, 2017.

\bibitem{Si2010}
Hang Si.
\newblock Constrained {D}elaunay tetrahedral mesh generation and refinement.
\newblock {\em Finite Elements in Analysis and Design}, 46:33--46, 2010.

\bibitem{Tantardini2019}
Mattia Tantardini, Francesca Ieva, Lucia Tajoli, and Carlo Piccardi.
\newblock Comparing methods for comparing networks.
\newblock {\em Scientific Reports}, 9(1):17557, 2019.

\bibitem{xu2015}
Hangjun Xu.
\newblock An algorithm for comparing similarity between two trees.
\newblock \url{https://arxiv.org/abs/1508.03381}, 2015.
\newblock Last accessed 26 Sept 2022.

\end{thebibliography}

\end{document}